\let\ifarxiv\iftrue
\let\ifappendix\iftrue
\ifcsname ifdraft\endcsname\else
\expandafter\let\csname ifdraft\expandafter\endcsname
\csname iffalse\endcsname
\fi

\ifcsname ifappendix\endcsname\else
\expandafter\let\csname ifappendix\expandafter\endcsname
\csname iffalse\endcsname
\fi

\ifcsname ifsub\endcsname\else
\expandafter\let\csname ifsub\expandafter\endcsname
\csname iffalse\endcsname
\fi

\ifcsname ifarxiv\endcsname\else
\expandafter\let\csname ifarxiv\expandafter\endcsname
\csname iffalse\endcsname
\fi

\ifarxiv
\documentclass[acmsmall,screen,nonacm]{acmart}
\pdfoutput=1 %
\else
\ifsub
\documentclass[acmsmall]{acmart}
\else
\documentclass[acmsmall,review]{acmart}
\fi
\fi

\setcopyright{none}

\acmJournal{PACMPL}
\acmVolume{1}
\acmNumber{CONF} %
\acmArticle{1}
\acmYear{2021}
\acmMonth{1}
\acmDOI{} %
\startPage{1}

\usepackage{mathtools}
\usepackage{hyperref}
\usepackage[plain]{algorithm}
\usepackage[noend]{algpseudocode}
\usepackage{skull,ifoddpage,marginnote}
\usepackage{stmaryrd}
\usepackage[capitalize,compress]{cleveref}
\usepackage{tikz}
\usepackage{scalerel}
\usepackage[inline]{enumitem}
\usepackage{multicol}
\usepackage{array}
\usepackage{multirow}
\usepackage{subcaption}
\usepackage{soul}
\usepackage{stackengine}
\usepackage{refcount}

\input{macros}

\begin{document}

\title{Visibility Reasoning for Concurrent Snapshot Algorithms}
\ifarxiv
\subtitle{Extended Version}
\fi

\author{Joakim Öhman}
\orcid{0000-0002-9284-3886}
\email{joakim.ohman@imdea.org}
\affiliation{%
  \institution{IMDEA Software Institute}
  \city{Madrid}
  \country{Spain}
}
\affiliation{%
  \institution{Universidad Politécnica de Madrid}
  \city{Pozuelo de Alarcón}
  \country{Spain}
}
\author{Aleksandar Nanevski}
\email{aleks.nanevski@imdea.org}
\affiliation{%
  \institution{IMDEA Software Institute}
  \city{Madrid}
  \country{Spain}
}

\begin{abstract}
  Visibility relations have been proposed by Henzinger et al. as an
abstraction for proving linearizability of concurrent algorithms that
obtains modular and reusable proofs. This is in contrast to the
customary approach based on exhibiting the algorithm's linearization
points.
In this paper we apply visibility relations to develop modular proofs
for three elegant concurrent snapshot algorithms of Jayanti.  The
proofs are divided by signatures into components of increasing level
of abstraction;
the components at higher abstraction levels are
shared, i.e., they apply to all three algorithms
simultaneously. Importantly, the interface properties mathematically
capture Jayanti's original intuitions that have previously been given
only informally.

\end{abstract}

\begin{CCSXML}
<ccs2012>
<concept>
<concept_id>10003752.10010124.10010138.10010142</concept_id>
<concept_desc>Theory of computation~Program verification</concept_desc>
<concept_significance>500</concept_significance>
</concept>
<concept>
<concept_id>10003752.10003809.10011778</concept_id>
<concept_desc>Theory of computation~Concurrent algorithms</concept_desc>
<concept_significance>500</concept_significance>
</concept>
</ccs2012>
\end{CCSXML}

\ccsdesc[500]{Theory of computation~Program verification}
\ccsdesc[500]{Theory of computation~Concurrent algorithms}

\maketitle

\section{Introduction}\label{sec:intro}

Linearizability~\cite{herlihy:90} is a standard correctness condition
for concurrent data structures. It requires that the operations in any
execution over the data structure may be ordered sequentially, without
incurring changes to the observed results. In other words, the methods
of a linearizable concurrent structure exhibit the same external
behavior as the sequential equivalent. Programmers can use the
concurrent variant to efficiently utilize modern systems' multi-core
setup, and rely on the sequential variant for understanding and formal
reasoning.

Many methods exist for proving linearizability of a data
structure. The standard idea shared by most of them involves finding,
for each method of the structure, a point in time in the concurrent
execution when the method can be considered as \emph{logically}
occurring. In other words, a method may execute over a period of time,
admitting interference from other concurrent threads; nevertheless,
for all reasoning intents and purposes, the execution is
indistinguishable from one where the method executes atomically at a
single point in time, without any interference.
This point in time is referred to as the~\emph{linearization point}.
Where the linearization point lies for a given method may depend on
the run-time behavior and interleaving of other methods executing
concurrently, sometimes even including behavior that occurs after the
original method has already terminated (the latter are often termed
\emph{far future} linearization points). Although significant progress
has been made recently in the verification of concurrent structures
and algorithms, and in particular by the introduction of so-called
\emph{prophecy
  variables}~\cite{aba+lam:91,LynchV+IC95,prophecy18,prophecy20} to
model the dependence of linearization points on run-time behavior
(including in the far future), establishing linearizability by
explicitly exhibiting the linearization points of an algorithm remains
a highly complex task in general.

A different approach, which avoids explicit reasoning about
linearization points, has been proposed
by~\citet{henzinger:concur13}.
In this approach, one first specifies a set of properties and proves
that the properties hold of every execution history over the data
structure. One then constructs a proof of linearizability using only
the specified properties as axioms, thereby abstracting from the
underlying executions. In other words, the low-level reasoning about
the concrete placement of linearization points is replaced by
higher-level reasoning out of data structure axioms, in turn
facilitating proof decomposition, modularity, abstraction, and reuse.

The properties used in the axiomatization are defined under relations
tying an event in the execution to other events that depend on it,
i.e., that \emph{observe} it.
For example, the event of reading of a pointer will be related to the
write event that was responsible for mutating the pointer. Once the
pointer is mutated by another write, the subsequent reads will observe
(i.e., be related to) the new write, or possibly another later write.
Similar ideas of reasoning about event observations have been used in
the axiomatizations of weak memory models~\cite{raad:popl19} and in
distributed systems~\cite{viotti:acmsurv16} where they have been
captured by means of so-called \emph{visibility} relations.
Reasoning by visibility has also been applied in different ways to
prove linearizability of concrete data structures, such as
queues~\cite{henzinger:concur13} and stacks~\cite{dodds:popl15},
including reasoning automation~\cite{enea:cav17}, and reasoning about
relaxed notions of linearizability~\cite{emm+ene:popl19,enea:esop20}.

In this paper, we demonstrate the further applicability of reasoning
by visibility, by applying it in a novel way to even more nuanced
algorithms and proofs. In particular, we show how visibility can be
used to express and axiomatize the important \emph{internal} properties
shared by several \emph{snapshot} algorithms.  Similarly to the
approach of~\citet{henzinger:concur13} to queues, the axiomatization
enables a modularization of the linearizability proof: a significant
portion of the proof is carried out once, and then reused for each
snapshot algorithm.

More specifically, we verify the three snapshot algorithms
by~\citet{jayanti:stoc05}. A snapshot algorithm scans a memory array
and returns the values read, so that the obtained values reflect the
state of the array at one point in time.  In a sequential setting this
is trivial to achieve since the array remains unchanged during the
scan. However, Jayanti's algorithms are concurrent, allowing
interfering threads to modify the array while a scan is in progress.
Jayanti's algorithms are of increasing efficiency and generality. The
simplest is the single-writer/single-scanner algorithm, which assumes
that no two scanners run concurrently, and that no two writers
concurrently modify the same array element. Jayanti's second
algorithm generalizes to a multi-writer/single-scanner setting, and
the third is the most general and ultimately desirable
multi-writer/multi-scanner version.
Jayanti describes the linearization points only for the
single-writer/single-scanner algorithm, but already this description
is quite challenging to transform into a fully formal correctness
proof~\cite{delbianco:ecoop17} because the algorithm exhibits far
future linearization points.

Of interest to us in the current paper is that each of the three
algorithms builds on the previous one by relaxing some part of the
previous algorithm's implementation, while preserving the essential
invariants that Jayanti calls \emph{forwarding principles}. Jayanti
credits the forwarding principles as the key idea behind his design,
because the principles are shared by the three algorithms, and
abstractly govern how a concurrent write into the array should be
``forwarded'' to a scanner that is in progress, but has already read
past the written element.

In this paper, we show how the forwarding principles can be
axiomatized mathematically in terms of visibility, which we present in
\cref{sec:overview} (Jayanti states the principles in English, much
less formally).  In \cref{sec:proofsketches}, we develop the
linearizability proof out of the axioms alone, so that it applies to
all three algorithms simultaneously.
We then establish that the axioms hold for each of the three
algorithms: in \cref{sec:proofsketches} for the first algorithm and
\cref{sec:multiwriter} for the second and third. In \cref{sec:afek} we
apply our method to verify another snapshot algorithm, that of
\citet{afek:acm93}.

By employing visibility, we sidestep the difficulties inherent in
reasoning about linearization points in general, and far future
linearization points in particular.
Our development substantiates that visibility is a natural abstraction
to use in specifications and proofs, as it enables formally capturing
the intuition---that of forwarding principles---that motivated the
design of Jayanti's algorithms in the first place.
In summary, our contributions are as follows.
\begin{itemize}
\item This is the first formal proof of all of Jayanti's three
  snapshot algorithms. Moreover, it efficiently reuses proofs between
  algorithms. \citet{delbianco:ecoop17} employed a
  variation of the linearization point approach to prove only
  Jayanti's first algorithm, with no direct way of extending the proof
  to the other two algorithms, which are significantly more involved.
That proof was mechanized in Coq. Based on \emph{loc.~cit.},
\citet{jacobs:jayanti} developed a mechanized proof of Jayanti's first
algorithm in VeriFast, using prophecy variables.  

\item We have axiomatized the forwarding principles, which were
  Jayanti's motivating insight, and key properties of his snapshot
  algorithms, but have so far been out of reach of formal
  mathematics. The axiomatization is non-trivial, and required
  generalizing from Jayanti's English description in order to apply to
  all three algorithms. It also enabled us to prove, for the first
  time, that forwarding principles formally imply linearizability.
  While visibility relations have been used before to axiomatize
  concurrent structures, this shows that they can also usefully
  capture important \emph{internal} properties.

\end{itemize}

\section{Overview}\label{sec:overview}

\subsection{Jayanti's First Snapshot Algorithm}
\label{subsec:overview-jayanti}

The $n$-snapshot data-structure is an array of length $n$, which
consists of two kinds of operations: $\scanProc$ which reads the
memory and returns a list of length $n$ reflecting the state of the
memory at a point in time; and $\writeProc(i,v)$ which writes the
value $v$ into memory cell $i$. 
We use $\abs{s}$ to range over instances of $\scanProc$, $\abs{w_i}$
to range over instances of $\writeProc(i,v)$ for some $v$, and
$\abs{e}$ to range over all operations in general. In line with
standard terminology of linearizability, we call these operations
\emph{abstract events} or \emph{abs events} and color them
\absColorName{}.

A snapshot (or any other) algorithm is linearizable if for every
concurrent execution, there exists some way of sequentially ordering
the \emph{overlapping} (abs) events so that the value returned by each
of the events remains unchanged compared to the concurrent execution,
and moreover, matches the intended semantics of the data structure.
Intuitively, linearizability implies that we can re-run the
computation sequentially to obtain the same results as in the original
concurrent run; however, the internal state of the algorithm during
and after the runs need not match.

Implementing an efficient and correct concurrent snapshot algorithm is
more challenging than it may seem. To highlight this point, consider a
na\"ive implementation, where writers simply write to a shared array
$\resA$ and scanners simply iterate over the array $\resA$ to obtain a
snapshot. To allow this implementation to operate efficiently, we
allow writes and scans to run concurrently.  For this simple
implementation, consider a snapshot array of length $2$, where we have a
scan $\abs{s}$ running concurrently with writes $\abs{w_0}$ writing
$2$ and $\abs{w_1}$ writing $3$. Let $0$ be the initial value of each
array cell, and consider the following execution:
\begin{itemize}
\item Scan $\abs{s}$ starts and reads $0$ from $\resA[0]$, after which
  the scheduler interrupts $\abs{s}$.
\item Write $\abs{w_0}$ starts, writing $2$ to $\resA[0]$, and after
  $\abs{w_0}$ finishes, write $\abs{w_1}$ writes $3$ to $\resA[1]$.
\item Scan $\abs{s}$ resumes, reading $3$ from $\resA[1]$. It then
  returns the snapshot $(0,3)$.
\end{itemize}
The snapshot $(0,3)$ should indicate that there exists a point in time
when the array consisted of that pair, however that is not the case.
The array started as $(0,0)$, followed by $(2,0)$ after write
$\abs{w_0}$, and $(2,3)$ after write $\abs{w_1}$, but none of these
states are reflected in the result. In a sense, the scan $\abs{s}$
missed the write $\abs{w_0}$, yet it caught the write $\abs{w_1}$,
which occurred after $\abs{w_0}$. Jayanti's snapshot algorithms ensure
that writes are not missed by the scanner, as we explain next.

\newcommand{\AComment}[1]{\Comment{\makebox[8mm][l]{\abs{#1}}}}
\newcommand{\RComment}[1]{\Comment{\makebox[8mm][l]{\rep{#1}}}}

\begin{algorithm}[t]
  \setlength\multicolsep{0pt}
  \begin{multicols}{2}
  \begin{algorithmic}[1]
    \Resource $\resA : \arrayType{n}{\valType}$
    \Resource $\resB : \arrayType{n}{\valType \cup \left\{\bot\right\}}$
    \Resource $\resX : \boolType \set \false$
    \Statex

    \Procedure{\writeProc}{$i : \natType, v : \valType$}{} \AComment{$w_i$}
    \miniskip
      \State $\resA[i] \set v$ \RComment{$\wra{w_i}$}
      \State $x \gets \resX$ \RComment{$\wrx{w_i}$}
      \If{$x$}
        $\resB[i] \set v$ \RComment{$\wrb{w_i}$}
      \EndIf
    \EndProcedure
    \columnbreak
    \Procedure{\scanProc}{}{$\arrayType{n}{\valType}$} \AComment{$s$}
    \miniskip
      \State $\resX \set \true$ \RComment{$\scon{s}$}
      \For{$i \in \left\{ 0 \dots n-1 \right\}$}
        \State $\resB[i] \set \bot$ \RComment{$\scr{s}{i}$}
      \EndFor
      \For{$i \in \left\{ 0 \dots n-1 \right\}$}
        \State $a \gets \resA[i]$ \RComment{$\sca{s}{i}$}
        \State $V[i] \set a$
      \EndFor
      \State $\resX \set \false$ \RComment{$\scoff{s}$}
      \For{$i \in \left\{ 0 \dots n-1 \right\}$}
        \State $b \gets \resB[i]$ \RComment{$\scb{s}{i}$}
        \If{$b \neq \bot$}
          $V[i] \set b$
        \EndIf
      \EndFor
      \State \textbf{return} $V$
    \EndProcedure
  \end{algorithmic}
  \end{multicols}\vspace{-2mm}
  \caption{\label{alg:jay1} Jayanti's single-writer, single-scanner
    snapshot algorithm over a memory of length $n$. Mnemonics on the
    right identify the corresponding commands.\vspace{-4mm}
  }
\end{algorithm}

\cref{alg:jay1} is the first and simplest of Jayanti's snapshot
algorithms. The idea is for a scan to make two passes over the memory,
first over the main array $\resA$, and then over the auxiliary array
$\resB$. A writer updates the array $\resB$ if it detects a concurrent
scan via the boolean flag $\resX$, to \emph{forward} its value. That
is, in case the scanner missed the writer's value when scanning
$\resA$, it will have a chance to catch the value when scanning
$\resB$.  We thus refer to $\resB$ as the \emph{forwarding array}.
\cref{alg:jay1} is a single-writer/single-scanner algorithm, meaning
that for it to behave correctly, two scans must not run concurrently
and two writers must not concurrently mutate the same array cell. In
an implementation, this can be enforced by explicit locking, which we
elide from~\cref{alg:jay1} following Jayanti's original presentation.

Describing the algorithm in more detail, $\scanProc$ works by first
setting $\resX$ to $\true$ by event $\rep{\scon{s}}$, signaling that a
scan is running the first pass. This is followed by clearing the
forwarding array $\resB$ by setting all its cells to $\bot$ via the
$\rep{\scr{s}{i}}$ events, for each $i$. The clearing ensures that the
current scan cannot consider the forwards left over from the previous
scans.
Next, the scanner creates a na\"ive snapshot by copying the main array
$\resA$ by the events $\rep{\sca{s}{i}}$ for each $i$ into the local
array $V$. However, as we argued before, this is insufficient for a
correct snapshot. This is where the second pass comes in, which starts
after $\resX$ is set to $\false$ by event $\rep{\scoff{s}}$.
In the second pass, the procedure repairs the na\"ive snapshot by
stepping through the forwarding array $\resB$ by the events
$\rep{\scb{s}{i}}$ for each $i$. If a non-$\bot$ (i.e., forwarded)
value is found, it overwrites the original value in $V$, thus
repairing the snapshot and preventing missing writes. Finally, the
scanner returns $V$, which contains the complete snapshot.

For $\writeProc(i,v)$, it starts by writing its value to $\resA$ by
event $\rep{\wra{w_i}}$, followed by a check for a concurrently
running scanner performing the na\"ive pass of the scan by event
$\rep{\wrx{w_i}}$. If such a scan is detected, the writer forwards its
value to $\resB$ by event $\rep{\wrb{w_i}}$.

Events marked in \repColorName{} are called
\emph{representation} or \emph{rep} events, and are used internally in
the implementation of abs (i.e., \absColorName{}) events. The
distinction between abs and rep events is standard in the theory of
linearizability~\cite{herlihy:90}.
We will use the naming convention whereby abs events and rep events
with the same priming belong together, e.g.\ we assume
$\rep{\sca{s}{i}}$ belongs to $\abs{s}$ and $\rep{\wra{w'_i}}$ belongs
to $\abs{w'_i}$.

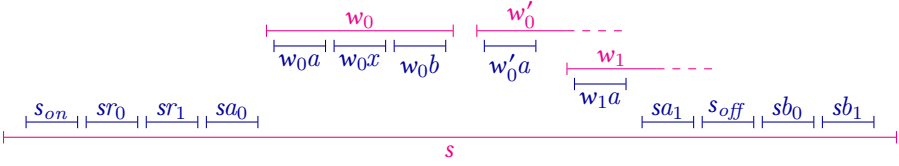
\begin{figure}
  \centering
  \begin{tikzpicture}
    \draw[|-|,draw=abs] (2.7,-1.4) -- node (S) [label={[yshift=4]below:$\abs{s}$}] {} (14.6,-1.4);
    \draw[|-|,draw=rep] (3,-1.2) -- node (SON) [label={[yshift=-6]$\rep{\scon{s}}$}] {} (3.7,-1.2);
    \draw[|-|,draw=rep] (3.8,-1.2) -- node (SR0) [label={[yshift=-6]$\rep{\scr{s}{0}}$}] {} (4.5,-1.2);
    \draw[|-|,draw=rep] (4.6,-1.2) -- node (SR1) [label={[yshift=-6]$\rep{\scr{s}{1}}$}] {} (5.3,-1.2);
    \draw[|-|,draw=rep] (5.4,-1.2) -- node (SA0) [label={[yshift=-6]$\rep{\sca{s}{0}}$}] {} (6.1,-1.2);
    \draw[|-|,draw=rep] (11.2,-1.2) -- node (SA1) [label={[yshift=-6]$\rep{\sca{s}{1}}$}] {} (11.9,-1.2);
    \draw[|-|,draw=rep] (12,-1.2) -- node (SOFF) [label={[yshift=-6]$\rep{\scoff{s}}$}] {} (12.7,-1.2);
    \draw[|-|,draw=rep] (12.8,-1.2) -- node (SB0) [label={[yshift=-6]$\rep{\scb{s}{0}}$}] {} (13.5,-1.2);
    \draw[|-|,draw=rep] (13.6,-1.2) -- node (SB1) [label={[yshift=-6]$\rep{\scb{s}{1}}$}] {} (14.3,-1.2);
    
    \draw[|-|,draw=abs] (6.2,0) node [left] {} -- node (W0) [label={[yshift=-6]$\abs{w_0}$}] {} (8.7,0);
    \draw[|-|,draw=rep] (6.3,-0.2) -- node (WA0) [label={[yshift=4]below:$\rep{\wra{w_0}}$}] {} (7,-0.2);
    \draw[|-|,draw=rep] (7.1,-0.2) -- node (WX0) [label={[yshift=4]below:$\rep{\wrx{w_0}}$}] {} (7.8,-0.2);
    \draw[|-|,draw=rep] (7.9,-0.2) -- node (WB0) [label={[yshift=4]below:$\rep{\wrb{w_0}}$}] {} (8.6,-0.2);
    
    \draw[|-,draw=abs] (9,0) node [left] {} -- node (W0') [label={[yshift=-6]$\abs{w'_0}$}] {} (10.2,0);
    \draw[dashed,draw=abs] (10.2,0) -- (11,0);
    \draw[|-|,draw=rep] (9.1,-0.2) -- node (WA0') [label={[yshift=4]below:$\rep{\wra{w'_0}}$}] {} (9.8,-0.2);

    \draw[|-,draw=abs] (10.2,-0.5) node [left] {} -- node (W1) [label={[yshift=-6]$\abs{w_1}$}] {} (11.4,-0.5);
    \draw[dashed,draw=abs] (11.4,-0.5) -- (12.2,-0.5);
    \draw[|-|,draw=rep] (10.3,-0.7) -- node (WA1) [label={[yshift=4]below:$\rep{\wra{w_1}}$}] {} (11,-0.7);
  \end{tikzpicture}\vspace{-4mm}
  \caption{Diagram illustrating an execution scenario for
    \cref{alg:jay1} with an array length of two.\vspace{-4mm}}
  \label{fig:jay1scenario}
\end{figure}

Now, consider the scenario illustrated in \cref{fig:jay1scenario},
where we have \cref{alg:jay1} being executed over an array of length
two, with each value initially set to $0$, with write events
$\abs{w_0}$, $\abs{w'_0}$ and $\abs{w_1}$ writing 2, 3 and 4
respectively, and a scan $\abs{s}$. When $\abs{s}$ reads $\resA[0]$
with $\rep{\sca{s}{0}}$ for its first pass, it reads 0, while when it
reads $\resA[1]$ with $\rep{\sca{s}{1}}$ it reads 4 written by
$\abs{w_1}$ with $\rep{\wra{w_1}}$.
Between the two passes of $\abs{s}$, $\rep{\wra{w_0}}$ is missed,
however since $\resX$ is set to $\true$ by $\rep{\scon{s}}$,
$\abs{w_0}$ will forward the value 2 with $\rep{\wrb{w_0}}$.
At the second pass, $\abs{s}$ reads $\resB[0]$ with
$\rep{\scb{s}{0}}$, reading 2 written by $\rep{\wrb{w_0}}$, and when
it reads $\resB[1]$ with $\rep{\scb{s}{1}}$, it reads $\bot$ written
by $\rep{\scr{s}{1}}$, meaning $\abs{s}$ will use the original value 4
for its final snapshot, thus returning (2,4).

In contrast to $(0, 3)$ before, linearizability admits $(2, 4)$ as a
correct snapshot even though $\resA$ never contained $(2,4)$ during
the execution. Indeed, $\resA$ only contained $(0,0)$, $(2,0)$,
$(3,0)$ and $(3,4)$. This is actually fine, because we can reorder the
concurrent events with the order
$\abs{w_0} \to \abs{w_1} \to \abs{s} \to \abs{w'_0}$, which, when
executed sequentially, result in $\abs{s}$ having snapshot $(2, 4)$.
In the physical execution, the event $\abs{w_0}$ returned before
$\abs{w'_0}$ and $\abs{w_1}$ started. The reordering respects this by
listing $\abs{w_0}$ before $\abs{w'_0}$ and $\abs{w_1}$. In other
words, the reordering affects only events that physically overlapped,
as required by linearizability.

Jayanti sketches the correctness proof of~\cref{alg:jay1} by
describing its linearization points.
The linearization point for a scan $\abs{s}$ is always when the scan
performs $\rep{\scoff{s}}$. However, the linearization point of a write
$\abs{w_i}$ varies. If there is no scan concurrent to $\abs{w_i}$, or
there is a concurrent scan, but it reads the value of $\abs{w_i}$
either from $\resA$ or from $\resB$, then the writer's linearization
point is at $\rep{\wra{w_i}}$. If there is a concurrent scan $\abs{s}$
that misses $\abs{w_i}$, which can occur if the scanner misses
$\abs{w_i}$ in its $\resA$ pass, and $\abs{w_i}$ either does not write
into $\resB[i]$ due to $\resX$ being set to $\false$ before the writer
could forward, or $\abs{w_i}$ writing into $\resB[i]$ too late,
then $\abs{w_i}$ must be logically considered as occurring after
$\abs{s}$. Thus, the linearization point of $\abs{w_i}$ is immediately
after $\rep{\scoff{s}}$, making the linearization point of
$\writeProc$ \emph{external}, as its position is given in terms of
another procedure, in this case $\scanProc$.
The observation that a scan missed a write $\abs{w_i}$, which occurs
when the scan reads $\resB[i]$ with $\rep{\scb{s}{i}}$, can be made
after the writer has already terminated, making $\writeProc$ exhibit a
far future linearization point.

We proceed to show how to organize the linearizability proof
of~\cref{alg:jay1}, and the other two Jayanti algorithms, by
axiomatizing forwarding via visibility, without using linearization
points. %

\subsection{Basic Abstractions of Visibility Reasoning}
\label{subsec:overview-observation}

\subsubsection{Events and Their Structure}\label{subsub:eventstruct}

An event is an object consisting of fields $\evStart$, $\evEnd$,
$\evOp$, $\evIn$, and $\evOut$, describing the following aspects of
the execution of some operation of the data structure: $\evStart$ is
the event's beginning time, $\evEnd$ is the ending time, $\evOp$ is
the operation name (e.g., $\scanProc$ or $\writeProc$), $\evIn$ is
the operation's input, and $\evOut$ is the output. We refer to the
elements of an event $e$ by projection, e.g. $e.\evStart$ and
$e.\evOp$.
The fields $e.\evStart$ and $e.\evEnd$ are natural numbers, and
$e.\evEnd$ may be $\infty$ (infinity) to represent that $e$ has not
terminated yet, i.e., $e$ is terminated iff $e.\evEnd \neq \infty$,
which we denote with $\term(e)$.
For every $e$, $e.\evEnd > e.\evStart$. The types of $e.\evIn$ and
$e.\evOut$ depend on $e.\evOp$, and $e.\evOut$ is undefined iff
$e.\evEnd = \infty$.

Additionally, each event $e$ contains the optional field $\evParent$,
corresponding to the event that invoked $e$, if any. For example, if
$\rep{e}$ is a \rep{rep} event, then $\rep{e}.\evParent$ is the
\abs{abs} event that contains $\rep{e}$. Note that $e.\evParent$ and
$e$ need not have the same $\evOp$, $\evIn$ and $\evOut$ fields; e.g.,
a write \rep{rep} event can be invoked both by \abs{abs} writer and
\abs{abs} scanner.

Finally, we require that each \abs{abs} event is single-threaded, and
thus cannot fork children threads.

We denote the set of all events of a given execution history by
$\events$.
If $E \subset \events$, then
$\mathit{op}(E) = \left\{ e \in E \mid e.\evOp = op \right\}$ selects
the events with operation $\mathit{op}$, and we overload $\term$ over
sets with $\term(E) = \left\{ e \in E \mid \term(e) \right\}$ to
select the terminated events.

\axiomset{RB}

As linearization order can only affect physically overlapping events,
proving linearizability requires reasoning about non-overlapping
events, which is captured by the \emph{returns-before} relation
\[
  e \rb e' \wideDefeq e.\evEnd < e'.\evStart
\]
denoting that $e$ terminated before $e'$ started. Two events are
overlapping if they are unrelated by $\rb$. As customary, we write
$\rbeq$ for the reflexive closure of $\rb$.
The relation $\rb$ is irreflexive (i.e. \emph{acyclic}) partial order,
and moreover, an \emph{interval} order~\cite{Felsner:92}, as it
satisfies the following property
\begin{equation}
  e_1 \rb e_2 \land e'_1 \rb e'_2 \implies e_1 \rb e'_2 \lor e'_1 \rb e_2
  \label[prop]{ax:interval}
\end{equation}

\begin{figure}[t]
  \centering

    \centering
    \begin{tikzpicture}
      \draw[|-|] (0,0) node[above] {$\scriptscriptstyle \evStart$} -- (2,0) node [midway, above] {$e_1$} node[above] {$\scriptscriptstyle \evEnd$};
      \draw[|-|] (2.5,0) node[above] {$\scriptscriptstyle \evStart$} -- (6.5,0) node [midway, above] {$e_2$} node[above] {$\scriptscriptstyle \evEnd$};
      
      \draw[|-|] (-1,-0.5) node[below] {$\scriptscriptstyle \evStart$} -- (3,-0.5) node [midway, below] {$e'_1$} node[below] {$\scriptscriptstyle \evEnd$};
      \draw[|-|] (3.5,-0.5) node[below] {$\scriptscriptstyle \evStart$} -- (5.5,-0.5) node [midway, below] {$e'_2$} node[below] {$\scriptscriptstyle \evEnd$};
    \end{tikzpicture}\vspace{-2mm}

    \caption{Visual representation of events highlighting the interval
      and subevent \cref{ax:interval,ax:rb-absrep}.\vspace{-4mm}}
  \label{fig:rb-prop}
\end{figure}
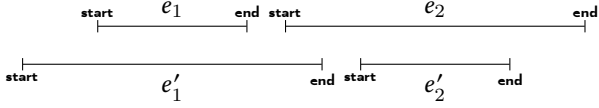

\cref{fig:rb-prop} illustrates why \cref{ax:interval} must hold. The
figure shows events $e_1, e_2, e'_1, e'_2$ such that $e_1 \rb e_2$,
and $e'_1 \rb e'_2$, and \cref{ax:interval} holds because also
$e_1 \rb e'_2$. We can try to invalidate the latter by shifting $e'_1$
and $e'_2$ to the left so that $e'_2.\evStart < e_1.\evEnd$ while
maintaining $e'_1.\evEnd < e'_2.\evStart$. But then we're forced to
have $e'_1.\evEnd < e_2.\evStart$, i.e. $e'_1 \rb e_2$ which
re-establishes \cref{ax:interval}.

We also say that $e$ is a \emph{subevent} of $e'$ (alternatively, $e'$
\emph{contains} $e$) if
\[
  e \subev e' \wideDefeq e'.\evStart \leq e.\evStart \land
  e.\evEnd \leq e'.\evEnd
\]
For any \rep{rep} event $\rep{e}$, we require that $\rep{e}.\evParent$
must be an \abs{abs} event such that $\rep{e} \subev \rep{e}.\evParent$.
Additionally, the following property holds for
subevents and returns-before.
\begin{equation}
  e_1 \subev e'_1 \land e_2 \subev e'_2 \land e'_1 \rb e'_2
  \implies e_1 \rb e_2
  \label[prop]{ax:rb-absrep}
\end{equation}
That is, if $e'_1$ returned before $e'_2$ then all subevents of $e'_1$
must return before any subevent of $e'_2$. For example, in
\cref{fig:rb-prop}, $e_1 \subev e'_1$ and $e'_2 \subev e_2$, and
$e_1 \rb e'_2$. If we shift $e'_1$ to the left so that $e'_1 \rb e_2$,
then we just increase the distance between $e_1$ and $e'_2$,
maintaining $e_1 \rb e'_2$.

\subsubsection{Visibility Relations}\label{sec:vis}
\axiomset{V}
If the event $e'$ depends on the result of $e$, we say that $e$ is
\emph{visible} to $e'$, or alternatively that $e$ is \emph{observed} by
$e'$. We denote the relationship as
\[
e \obs e'
\]
Depending on the data structure being verified, we will often require
several different observation relations to differentiate how the
observation came about. For example, in the case of Jayanti, we will
use $\robs$ for a ``reads-from observation'' (a reader observes a
writer by reading what was written), and $\fobs$ for a ``forwarding
observation'' (a scanner observes a writer by having the written value
forwarded). 
The former relation is tied to the physical act of reading a pointer,
and will be the same in all three algorithms. The latter relation
differs for different algorithms and is typically provided by the
human verifier, similarly to how loop invariants must often be
provided.

We will typically obtain the visibility relation $\obs$ by unioning
all the different observation subrelations. We then require the
following property of $\obs$
\begin{align}
e \obs^+ e' \implies e' \not\rbeq e \label[prop]{ax:wfobs}
\end{align}
where $\obs^+$ is the transitive closure of $\obs$. Read
contrapositively, the property says that if $e'$ terminated before $e$
started, then $e'$ cannot end a non-empty sequence of observations
starting from $e$. In particular, as a special case, $e'$ cannot
observe $e$ or be equal to $e$, meaning $\obs$ has to be irreflexive.
We write $\obseq$ and $\robseq$ for the reflexive closure of $\obs$
and $\robs$ respectively.

\subsubsection{Happens-Before Relation}
Given two events $e$ and $e'$, if $e \rb e'$ or $e \obs e'$, then
clearly, in the ultimate linearization order we want to construct, $e$
must appear before $e'$. To capture this intuition, we define
the \emph{happens-before} relation as the transitive closure
\[
  \mathord{\hb} \wideDefeq (\mathord{\rb}\ \cup \mathord{\obs})^+
\]
We also name the \emph{single-step happens-before} relation
$\mathord{\hb_1}\,{=}\,(\mathord{\rb}\ \cup\ \mathord{\obs})$, so that
$\mathord{\hb} = \mathord{\hb_1 \hbeq}\ =\ \mathord{\hbeq \hb_1} =
\mathord{\hb_1^+}$, where $\mathord{\hbeq}$ is the
reflexive-transitive closure $\mathord{\hb_1^*}$.
The $\rb$, $\obs$ and $\hb$ relations are all standard in the
literature~\cite{viotti:acmsurv16}.
An important property is that $\hb$ is an irreflexive (i.e.,
\emph{acyclic}) partial order, which is ensured by $\obs$
satisfying~\cref{ax:wfobs}.
\begin{lemma}\label{lem:hb-irrefl}
  If visibility relation $\obs$ satisfies \cref{ax:wfobs} then $\hb$
  is irreflexive.
\end{lemma}
\begin{proof}
  We assume that $\hb$ is not irreflexive (i.e., there exists $e$ such
  that $e \hb e$), and derive contradiction. The relation $e \hb e$ is
  a cyclic chain of $\hb_1$, each of which is either $\rb$ or
  $\obs$. We are justified in considering chains with only one or zero
  occurrences of $\rb$, as chains with more occurrences of $\rb$ can
  be shortened, thus we can recursively shorten a chain until it
  consists of one or zero $\rb$. Indeed, if the chain has
  more than one $\rb$, it has the form
  \[e \hb_1 \cdots \hb_1 e_1 \rb e_2 \obs \cdots \obs e_3 \rb e_4
    \hb_1 \cdots \hb_1 e\ ,\] 
  where $e_2$ and $e_3$ are related by a chain of zero or more
  $\obs$'s (i.e., $e_2 = e_3$ or $e_2 \obs^+ e_3$). But then, we can
  remove one $\rb$ as follows. If $e_2 = e_3$, by transitivity of
  $\rb$, we can shorten the chain to $e \hbeq e_1 \rb e_4 \hbeq e$. If
  $e_2 \obs^+ e_3$, by \cref{ax:interval}, it must be either
  $e_1 \rb e_4$ or $e_3 \rb e_2$. If $e_1 \rb e_4$, we again shorten
  to $e \hbeq e_1 \rb e_4 \hbeq e$. Otherwise, $e_3 \rb e_2$ and
  $e_2 \obs^+ e_3$ contradict \cref{ax:wfobs}.

  On the other hand, if there is exactly one $\rb$ in the chain, i.e.,
  $e \obs^* e_1 \rb e_2 \obs^* e$, then $e_2 \obs^* e \obs^* e_1$,
  i.e., $e_2 \obs^* e_1$ which, along with $e_1 \rb e_2$, contradicts
  \cref{ax:wfobs}. Finally, if there is no $\rb$ in the chain, i.e.,
  $e \obs^+ e$, then \cref{ax:wfobs} directly implies the
  contradiction $e \not\rbeq e$.
\end{proof}

We will usually color the relations with the same colors as the events
they are relating (e.g., $\abs{e \obs e'}$ for abs events $\abs{e}$
and $\abs{e'}$ and $\rep{e_r \hb e'_r}$ for rep events $\rep{e_r}$ and
$\rep{e'_r}$).

\subsubsection{Memory Model}\label{sec:memorymodel}

\begin{figure}
  \centering
  \renewcommand{\arraystretch}{\sigSpacing}%

\axiomset{M}
\signatureHeader{$\aregSig$}{Atomic Register}

\begin{letbox}{l}
  & $\rep{W}$ & Set of all writes of the register of the history\\
  & $\rep{R}$ & Set of all reads of the register of the history\\
\end{letbox}

\begin{tabular*}{0.9\textwidth}{l@{\hskip 1em}A@{\extracolsep{\fill}}r}
  $\bigSum$
  & \rep{\obs} &\subseteq (\rep{W} \cup \rep{R})^2 & Visibility relation\\
  & \rep{\robs} &\subseteq \mathord{\rep{\obs}} \cap \rep{W} \times \rep{R} & Reads-from visibility
\end{tabular*}\vspace{\sepAfterVar}

\begin{letbox}{A}
  & \rep{\hb} &\wideDefeq (\rep{\rb} \cup \rep{\obs})^+ & Happens-before order\\
\end{letbox}

\begin{tabular*}{0.9\textwidth}{l@{\hskip 2em}A@{\extracolsep{\fill}}r}
  $\forall \rep{e}, \rep{e'}.$%
  &\rep{e \obs^+ e'}
  &\implies \rep{e' \not\rbeq e}
  &\eqref{ax:wfobs}\\
  $\forall \rep{r} \in \term(\rep{R}).$
  &\multicolumn{2}{c}{$\exists \rep{w} \in \rep{W}.\ \rep{w \robs r} \land \rep{w}.\evIn = \rep{r}.\evOut$}
  &\insertEq\label[prop]{ax:mem-io}\\
  $\forall \rep{w} \in \rep{W}, \rep{r}\in \rep{R} .$
  &\rep{w \robs r}
  &\implies \nexists \rep{w'}.\ \rep{w \hb w' \hb r}
  &\insertEq\label[prop]{ax:mem-nowrbetween}\\
  $\forall \rep{w}, \rep{w'} \in \rep{W}, \rep{r} \in \rep{R}.$
  &\rep{w \robs r} \land \rep{w' \robs r} &\implies \rep{w = w'}
  &\insertEq\label[prop]{ax:mem-robsuniq}\\
  $\forall \rep{w}, \rep{w'} \in \rep{W}.$
  &\rep{w \neq w'} &\implies \rep{w \hb w'} \lor \rep{w' \hb w}
  &\insertEq\label[prop]{ax:mem-wrtotal}\\[\sepLastProp]
  \bottomrule
  \end{tabular*}
  \caption{Signature representation of the properties of an atomic
    register.}
  \label{fig:memory}
\end{figure}

With the above relations we can now state as axioms the properties
that we expect of the underlying memory model. We start with a simple
axiomatization of memory that is sufficient for \cref{alg:jay1}. We
will extend this axiomatization in \cref{sec:multiwriter} to account
for the $\LLins$ (load-link) and $\SCins$ (store-conditional) memory
operations, required by~\cref{alg:jay2,alg:jay3}.

The memory model only considers operations over individual memory
cells, aka.\ atomic registers~\cite{Herlihy-Shavit:08}. Their
axiomatization in terms of visibility relations is given in
\cref{fig:memory}, in the form of a signature we assume the memory to
satisfy. We use \repColorName{} in this figure, since we will only use
events of atomic registers as rep events.
The sets $\rep{W}$ and $\rep{R}$ are defined by the history as the
sets of writes and reads respectively. The quantifier $\Sigma$
signifies that the relations $\rep{\obs}$ and $\rep{\robs}$ are
abstract
components of the specification; the clients do not know anything
about these components outside of the listed axioms. Given an instance
$X$ of the signature, the clients can refer to the relations and to
the let definitions by projection as $X.\rep{\obs}$, $X.\rep{\robs}$,
$X.\rep{W}$ and $X.\rep{R}$. 
We will present similar signatures for snapshot data-structures and Jayanti-style forwarding.

Looking at \cref{fig:memory},
an atomic register is mathematically represented by two sets of
events: $\rep{W}$, corresponding to all the writes into the register's
memory cell, and $\rep{R}$, corresponding to all the reads. These are
related by $\rep{w \robs r}$, stating that read $\rep{r}$ \emph{reads
from} (i.e., observes) write $\rep{w}$. These events may also be
related by some other relation part of the total set of observations
$\rep{\obs}$, which includes $\rep{\robs}$; \cref{ax:wfobs} must hold
of $\rep{\obs}$ to ensure that $\rep{\hb}$ is irreflexive (thus, a
partial order) as per \cref{lem:hb-irrefl}.
\cref{ax:mem-io} states that for any terminated read $\rep{r}$ there
exists an observed write $\rep{w}$ with its input being the same as
the output of $\rep{r}$. \cref{ax:mem-nowrbetween} states that the
read $\rep{r}$ can observe only the latest write $\rep{w}$ into the
given memory cell. \cref{ax:mem-robsuniq} states that a read may
observe at most one write. \cref{ax:mem-wrtotal} states that the
writes into the given cell are totally ordered. These are standard
properties of sequentially consistent memory
\cite{Herlihy-Shavit:08}.

We want to reason about $\rep{\hb}$ order arbitrarily between rep
events, even if they belong to distinct register objects. As an
analogue to the locality property of linearizability, which says that
the union of multiple linearizable objects is itself linearizable, we
establish that the union of objects satisfying \cref{ax:wfobs} itself
satisfies \cref{ax:wfobs}. By \cref{lem:hb-irrefl}, this implies that
$\rep{\hb}$ order between any rep event is a partial order.
More formally, let $\rep{R_1} \dots \rep{R_m}$ be
all objects satisfying $\aregSig$ and representing our memory. We
define the top-level $\rep{\obs}$ (and top-level $\rep{\hb}$) relation
over rep events as the union of visibility of all register objects:
\begin{align*}
  \rep{\mathord{\obs}} &\wideDefeq (\rep{R_1.\mathord{\obs}})\ \cup \dots \cup\ (\rep{R_m.\mathord{\obs}})\\
  \rep{\mathord{\hb}} &\wideDefeq (\rep{\mathord{\rb}}\ \cup\ \rep{\mathord{\obs}})^+
\end{align*}
Since each register satisfies \cref{ax:wfobs} and each individual
$\rep{\obs}$ only relates events from the same register, it follows
that the combined $\rep{\obs}$ also satisfies \cref{ax:wfobs}. By
\cref{lem:hb-irrefl}, the top-level $\rep{\hb}$ is then a partial
order. The combination similarly satisfies
\cref{ax:mem-io,ax:mem-nowrbetween,ax:mem-robsuniq,ax:mem-wrtotal}.
We will use these top-level definition when relating
rep events of distinct objects.

\subsubsection{Linearizability}\label{sec:linearizability}
We now define linearizability formally in terms of visibility relations.

\begin{definition}\label{def:lin}
  History $E$ is linearizable with respect to data structure
  $\mathcal{D}$ if there exists a visibility relation $\obs$ and a total order $<$ on the set
  of events
  $E_c = \{ e \mid \exists e' \in \term(E).\ e \obs^* e' \}$, such
  that: (1) $\mathord{\hb} = (\rb \cup \obs)^+ \subseteq \mathord{<}$
  on $E_c$, and (2) executing the events in $E_c$ in the order of $<$
  is a legal sequential behavior of $\mathcal{D}$. The order $<$ is
  the \emph{linearization order} (or \emph{linearization}, for short)
  of $E$. Algorithm (or structure) $\mathcal{A}$ is linearizable
  wrt.~$\mathcal{D}$ if every history of $\mathcal{A}$ is linearizable
  wrt.~$\mathcal{D}$.
\end{definition}

Definition~\ref{def:lin} differs somewhat from the original
one~\cite{herlihy:90} in that we use the visibility relation $\obs$ to
complete $\term(E)$ into $E_c$ with events that have been observed,
but have not yet terminated. The original definition existentially
abstracts over the set of completing events; formally, it does not
organize them into a visibility relation, but in practice these events
are always added because they have been observed by some terminated
event, and are necessary to ensure legal sequential behavior.
The requirement $\mathord{\hb} \subseteq \mathord{<}$ means that $<$
must respect $\rb$ (in addition to $\obs$), and in particular that $<$
can reorder only overlapping events.
In this paper, $\mathcal{D}$ is the snapshot data structure that has
the following sequential behavior over its state (the array $\resA$):
\begin{itemize}
\item $w_i \in \writeProc_i$ writes $w_i.\evIn$ to $\resA[i]$ and
  returns nothing, i.e., $w_i.\evOut = ()$.
\item $s \in \scanProc$ does not modify $\resA$ and returns a copy of
  $\resA$, i.e., $s.\evOut = \resA$.
\end{itemize}

\subsection{Hierarchical Structure of the Proof}

\begin{figure}
  \centering

\begin{tikzpicture}[auto=left,scale=0.89, every node/.style={scale=0.89}]
  \node[draw] (Jay1) at (0,0) {\cref{alg:jay1}};
  \node[draw, align=center] (Fwd) at (4.5,0) {$\fwdSig$~sig.\\(\cref{fig:forward})};
  \node[draw, align=center] (Snap) at (9,0) {$\snapshotSig$~sig.\\(\cref{fig:snapshot})};
  \node[draw, align=center] (Lin) at (13.5,0) {Linearizability\\(Def. \ref{def:lin})};

  \node[draw] (Jay2) at (0,-1.3) {\cref{alg:jay2}};
  \node[draw] (Jay3) at (0,-2.6) {\cref{alg:jay3}};
  \node[draw, align=center] (Fwd2) at (4.5,-1.95) {$\mwFwdSig$~sig.\\(\cref{fig:forward2})};

  \draw[->] (Jay1) -- node {Step (1)} node[swap] {\cref{lem:jay1-fwd}} (Fwd);
  \draw[->] (Fwd) -- node {Step (2)} node[swap] {\cref{lem:fwd-ss}} (Snap);
  \draw[->] (Snap) -- node {Step (3)} node[swap] {\cref{lem:ss-lin}} (Lin);

  \draw[->] (Jay2.east) -- node[sloped] {Step (1a)} node[sloped,swap] {\cref{lem:jay2-fwd2}} ([yshift=-3]Fwd2.north west);
  \draw[->] (Jay3.east) -- node[sloped] {Step (1a)} node[sloped,swap] {\cref{lem:jay3-fwd2}} ([yshift=3]Fwd2.south west);
  \draw[->] (Fwd2) -- node[align=center] {Step\\(2a)} node[swap,align=center] {Lemma\\\ref{lem:fwd2-fwd}} (Fwd);

\end{tikzpicture} 

\caption{%
  Overview of the structure of the linearizability proof for each
  of Jayanti's snapshot algorithms.\vspace{-2mm}
}
\label{fig:overview}
\end{figure}
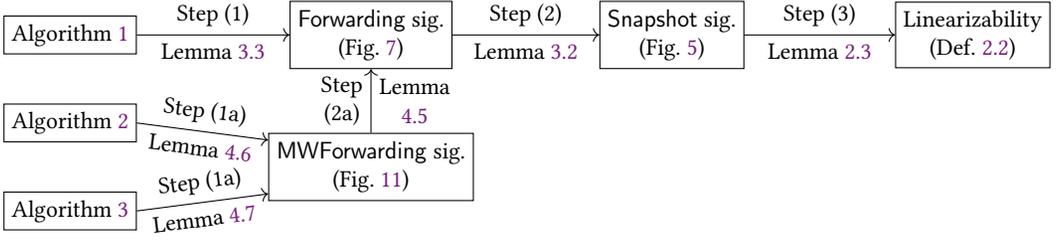

Ultimately, our goal is to prove that each of Jayanti's three snapshot
algorithms is linearizable. To accomplish this in a way that maximizes
proof reuse, we divide our proofs into multiple segments, which we
illustrate in \cref{fig:overview}. The proof of Jayanti's first
algorithm is split into the following steps:
\begin{enumerate}
\item Jayanti's single-writer/single-scanner algorithm
  (\cref{alg:jay1}) consists of the rep events described in
  \cref{fig:struct} and satisfies the forwarding snapshot
  signature that we give in \cref{fig:forward}.
\item Any algorithm consisting of the rep events described in
  \cref{fig:struct} and satisfying the forwarding snapshot signature
  from \cref{fig:forward}, also satisfies the general snapshot
  signature in \cref{fig:snapshot}.
\item Any algorithms satisfying the general snapshot signature in
  \cref{fig:snapshot} is linearizable.
\end{enumerate}
For Jayanti's remaining algorithms, we incorporate the following two
additional steps.
\begin{enumerate}
\item[(1a)] Jayanti's two multi-writer algorithms
  (\cref{alg:jay2,alg:jay3}) consist of the rep events described in
  \cref{fig:struct2} and satisfy the multi-writer forwarding snapshot
  signature in \cref{fig:forward2}.
\item[(2a)] Any algorithm consisting of the rep events described in
  \cref{fig:struct2} and satisfying the multi-writer forwarding
  snapshot signature in \cref{fig:forward2}, also satisfies the
  forwarding snapshot signature in \cref{fig:forward}. The
  linearizability then follows by steps (2) and (3) above, which are
  thus reused and shared by all three algorithms.
\end{enumerate}

We proceed to describe the intuition behind the steps (1) to (3),
together with the forwarding and snapshot signatures. The detailed
proofs of (1) to (3) are in \cref{sec:proofsketches}. The description
and proofs of (1a) and (2a) are in \cref{sec:multiwriter}.

\subsection{Snapshot Signature}

\begin{figure}[t]
  \centering
  \renewcommand{\arraystretch}{\sigSpacing}%

\axiomset{S}
\signatureHeader{$\snapshotSig$}{Snapshot data-structure}

\begin{letbox}{l}
  & $\abs{W_i}$ & \WrSetDesc{}\\
  & $\abs{S}$ & \ScSetDesc{}\\
\end{letbox}

\begin{tabular*}{0.9\textwidth}{l@{\hskip 1em}A@{\extracolsep{\fill}}r}
  $\bigSum$
  & \abs{\WrEff_i} &\subseteq \abs{W_i} & Set of effectful writes of index $i$\\
  & \abs{\obs} &\subseteq (\bigcup_i \abs{\WrEff_i} \cup \abs{S})^2 & Visibility relation\\
  & \abs{\robs} &\subseteq \mathord{\abs{\obs}} \cap \bigcup_i \abs{\WrEff_i} \times \abs{S} & Reads-from visibility
\end{tabular*}\vspace{\sepAfterVar}

\begin{letbox}{A}
  & \abs{\hb} &\wideDefeq (\abs{\rb} \cup \abs{\obs})^+ & Happens-before order\\
\end{letbox}

\begin{tabular*}{0.9\textwidth}{l@{}A@{\extracolsep{\fill}}r}
  $\forall \abs{e}, \abs{e'}.$
  &\abs{e \obs^+ e'}
  &\implies \abs{e' \not\rbeq e}
  &\eqref{ax:wfobs}\\
  $\forall i, \abs{s} \in \term(\abs{S}).$
  &\multicolumn{2}{c}{$\exists \abs{w_i}.\ \abs{w_i \robs s} \land \abs{w_i}.\evIn = \abs{s}.\evOut[i]$}
  &\insertEq\label[prop]{ax:ss-io}\\
  $\forall \abs{w_i}, \abs{s}.$
  &\abs{w_i \robs s}
  &\implies \nexists \abs{w'_i}.\ \abs{w_i \hb w'_i \hb s}
  &\insertEq\label[prop]{ax:ss-nowrbetween}\\
  $\forall \abs{w_i}, \abs{w'_i}, \abs{s}.$
  &\abs{w_i \robs s} \land \abs{w'_i \robs s} &\implies \abs{w_i = w'_i}
  &\insertEq\label[prop]{ax:ss-robsuniq}\\
  $\forall \abs{w_i}, \abs{w'_i} \in \abs{\WrEff_i}.$
  &\abs{w_i \neq w'_i} &\implies \abs{w_i \hb w'_i} \lor \abs{w'_i \hb w_i}
  &\insertEq\label[prop]{ax:ss-wrtotal}\\
  &\multicolumn{2}{c}{\hspace{-22mm}$\term(\abs{W_i}) \subseteq \abs{\WrEff_i}$}
  &\insertEq\label[prop]{ax:ss-wrterm}\\
  $\forall \abs{w_i}, \abs{w'_i}, \abs{w_j}, \abs{w'_j}, \abs{s}, \abs{s'}.$
  &\multicolumn{2}{c}{$\abs{w_i,w_j \robs s} \land \abs{w'_i,w'_j \robs s'} \land \abs{w_i \hb w'_i}
  \implies \abs{w'_j \not\hb w_j}$}
  &\insertEq\label[prop]{ax:ss-mono}\\[\sepLastProp]
  \bottomrule
  \end{tabular*}

  \caption{ %
    Signature representation of the properties of a
    snapshot data-structure.\vspace{-2mm}
  }
  \label{fig:snapshot}
\end{figure}

Going backwards, we start with the general snapshot signature
(\cref{fig:snapshot}). It describes the axioms for reasoning about the
snapshot data structure as a whole, and will interface the
linearizablity proofs by the following lemma, which we will prove in
\cref{sec:proofsketches}.
\begin{lemma}\label{lem:ss-lin}
  Histories satisfying the $\snapshotSig$~signature
  (\cref{fig:snapshot}) are linearizable.
\end{lemma}
The $\snapshotSig$ signature is almost identical to that for atomic
registers (\cref{fig:memory}); the main distinction is that we have
multiple sets of writes into multiple pointers, scans (observing
multiple writes) instead of reads, the sets of \emph{effectful} writes
$\abs{\WrEff_i}$, and the extra \cref{ax:ss-wrterm,ax:ss-mono}.

Not all writes are immediately observable once they have started,
therefore we introduce the set of effectful writes $\abs{\WrEff_i}$
for each $i$. These are the writes that are available for scans to
observe (by definition of $\abs{\robs}$) and which we can order by
\cref{ax:ss-wrtotal}. \cref{ax:ss-wrterm} encodes that each terminated
write must have been effectful.

\cref{ax:ss-mono} imposes a form of monotonicity on the ordering of
the writes observed by scans over multiple memory cells. More
specifically, if we have two writes $\abs{w_i}$, $\abs{w_j}$ observed
by scan $\abs{s}$, and two writes $\abs{w'_i}$, $\abs{w'_j}$ observed
by scan $\abs{s'}$, then the writes into $i$ and $j$ cannot be ordered
in the opposite way.
The property ensures that the scans can be totally ordered in an
eventual linearization order. Indeed, if $\abs{w_i \hb w'_i}$ and
$\abs{w'_j \hb w_j}$, and we ordered $\abs{s}$ before $\abs{s'}$, then
we reach a contradiction by \cref{ax:ss-nowrbetween}, because the
event $\abs{w_j}$ occurs between $\abs{w'_j}$ and the scan $\abs{s'}$
that observes it. Similar argument applies if we try to order the
scans the other way around.

\subsection{Forwarding Signature}

The $\snapshotSig$ signature in \cref{fig:snapshot} captures the key
properties of snapshot algorithms, and we shall prove linearizability
solely out of the axioms of this signature. However, Jayanti's
algorithms have more in common than merely being snapshot algorithms,
as they share the same design principle of forwarding. We can thus
further modularize the proofs by axiomatizing forwarding itself.

To see what the axiomatization should accomplish, we must first
discuss the high-level differences between Jayanti's algorithms that
the axiomatization must abstract over. We present
\cref{alg:jay2,alg:jay3} in detail in \cref{sec:multiwriter},
but for now it suffices to know that in order to support the
multi-writer functionality of \cref{alg:jay2,alg:jay3}, we need to
allow for different kinds of forwarding that go beyond simply writing
into the auxiliary array $\resB$ that \cref{alg:jay1} does.

Additionally, to support multi-scanner functionality of
\cref{alg:jay3}, we need to allow for an abstract notion of
\emph{virtual} scan. While in a multi-scanner setting several physical
scans may run concurrently, the key idea of Jayanti for \cref{alg:jay3}
is that these scans collaborate to create a virtual scan (and a
corresponding snapshot), of which at most one may exist at any given
moment. Thus, even though \cref{alg:jay3} is physically a
multi-scanner, in \cref{sec:jay3} we will still be able to
conceptually see it as a single-scanner one.

While rep and abs events originate from execution histories, the
virtual scans are artificial events that the human verifier creates
themselves for purposes of verification. In analogy with the concept
of ``ghost state'' that is frequently used in verification of
concurrent programs, we can say that virtual scans are ``ghost''
events.
We will further require a mapping that sends each abs scan to a
virtual scan to which the abs scan contributed. Thus, in a
multi-scanner algorithm, a number of abs scans may be mapped to the
same virtual scan. The mapping may be partial because some abs scan
may not have yet reached a point for which the virtual scan
representing it has been determined. For single-scanner algorithms,
this map is a (total) bijection that identifies virtual and abs scans.
In particular, the intuition about abs scans from \cref{alg:jay1} will
suffice to understand virtual scans in our axiomatization of forwarding
principles.
In the rest of the text, we use the color \vrtColorName{} to visually
mark the virtual scans.

\subsubsection{Event Signatures}

\begin{figure}
  \centering
  \renewcommand{\arraystretch}{\sigSpacing}%

  \begin{tabular*}{0.9\textwidth}{l@{\extracolsep{\fill}}r}
    \toprule
    \textbf{resource}\ \ $\rep{\resA} \wideColon \arrayType{n}{\aregSig}$ & Main memory array\\
    \bottomrule
  \end{tabular*}\vspace{\sepEvSig}

\begin{tabular*}{0.9\textwidth}{@{\hskip 2em}A@{\extracolsep{\fill}}r}
  \toprule
  \multicolumn{2}{l}{\evsignature{$\writeProc_i$}} & \textbf{Write into cell $i$}\\
  \midrule
  \rep{a} &\wideColon \rep{\resA[i].W} & Rep event of writer writing value into $\resA[i]$\\
  \bottomrule
\end{tabular*}\vspace{\sepEvSig}

\begin{tabular*}{0.9\textwidth}{@{\hskip 2em}A@{\extracolsep{\fill}}r}
  \toprule
  \multicolumn{2}{c}{\evsignature{$\vscanProc$}} & \textbf{Virtual scan}\\
  \midrule
  \rep{\resB} &\wideColon \arrayType{n}{\aregSig} & Forwarding array of this virtual scan\\
  \rep{r_i} &\wideColon \rep{\resB[i].W} & Rep event of scanner resetting $\resB[i]$ by writing $\bot$\\
  \rep{a_i} &\wideColon \rep{\resA[i].R} & Rep event of scanner reading from $\resA[i]$\\
  \rep{b_i} &\wideColon \rep{\resB[i].R} & Rep event of scanner reading from $\resB[i]$\\
  \bottomrule
\end{tabular*}

\caption{Resources and event signatures for abs writes and
  virtual scans corresponding to $\fwdSig$ signature.\vspace{-4mm}
}
\label{fig:struct}
\end{figure}

For the $\fwdSig$ signature, it is not enough to just consider
properties over the abs events, as for the $\snapshotSig$ signature,
but we also need to consider their rep events. For this reason, we
introduce a special kind of signature, called \emph{event signature},
to encode which rep events an abs or virtual event consists of. In
particular, for an \abs{abs} event $\abs{e}$, we write $\abs{e} \in
\textsc{sig}$, if $\abs{e}$'s rep events are enumerated by the fields
of the signature $\textsc{sig}$.  If $\abs{e} \in \textsc{sig}$, where
$\textsc{sig}$ contains the field $a$, we write $\rep{\mathit{ea}}$
for the projection $\abs{e}.\rep{a}$. Similarly, if $\vrt{e}$ is a
\vrt{virtual} event.
We shall further assume the following properties of structures
satisfying event signatures:
\begin{enumerate}[label=(ES.\arabic*)]
\item Let $e$ be an event in a signature with field $a$. If $\abs{e}$
  is an \abs{abs} event, then $\rep{\mathit{ea}}.\evParent = \abs{e}$,
  and thus also $\rep{\mathit{ea}} \subev \abs{e}$, by our assumption
  on event structure in Section~\ref{subsub:eventstruct}. If $\vrt{e}$
  is a \vrt{virtual} event, we do not insist on
  $\rep{\mathit{ea}}.\evParent = \vrt{e}$. Virtual event $\vrt{e}$ is
  a custom collection of rep events, whose parents remain the abs
  events that invoked them, not the collection that is
  $\vrt{e}$. Nevertheless, we still require $\rep{\mathit{ea}} \subev
  \vrt{e}$.\label[prop]{ax:evsig-subev}

\item Any event in an event signature instance is unique to the
  instance, i.e., if $e, e' \in \textsc{sig}$, and $\textsc{sig}$
  contains the field $a$, and
  $\rep{\mathit{ea}} = \rep{\mathit{e'a}}$, then $e = e'$. 
  \label[prop]{ax:evsig-uniq}
\item Instances of event signatures may have some (or all) of their
  subevents undefined, unless the instance is a terminated event, in
  which case all its subevents must be defined. For example, if
  $\abs{e}$ is non-terminated, $\rep{\mathit{ea}}$ may be undefined.
  This corresponds to ongoing abs events having only a subset of its
  rep events executed so far. \label[prop]{ax:evsig-partial}
\end{enumerate}

We now present \cref{fig:struct} which consists of the resources and
event signatures necessary for the $\fwdSig$ signature.
First, we have the array $\resA$ which corresponds to the main memory
that a forwarding snapshot algorithm operates over. \cref{alg:jay1}
clearly has such an array.
Next, we have that every abs write $\abs{w_i}$ has a unique rep event
$\rep{\wra{w_i}}$ signifying the physical act of writing into
$\resA$. This is clearly true of \cref{alg:jay1}, as we have seen. 
For a virtual scan $\vrt{\vsc}$, the figure postulates an array
$\resB$, as well as three rep events $\rep{\scr{\vsc}{i}}$,
$\rep{\sca{\vsc}{i}}$, and $\rep{\scb{\vsc}{i}}$, corresponding to
resetting $\resB[i]$, reading from $\resA[i]$, and reading from
$\resB[i]$, respectively. In \cref{alg:jay1}, we have already
described these events under the names $\rep{\scr{s}{i}}$,
$\rep{\sca{s}{i}}$, and $\rep{\scb{s}{i}}$, and similar events will
exist for \cref{alg:jay2,alg:jay3} as well. The notation implies that
the array $\resB$ and the subevents are local fields of
$\vrt{\vsc}$. We will make use of the locality of $\resB$ in the
multi-scanner case, where different virtual scans have different
forwarding arrays. In the case of single-scanner algorithms, the
virtual and abs scans coincide, and the forwarding array of each
virtual scan is instantiated with the global forwarding array.

\subsubsection{Forwarding Signature Formally}

\newcommand{\fwdRobsDef}[1]{%
  \abs{w_i \robs \vrt{\vsc}} #1%
  \wideDefeq (\rep{\wra{w_i} \robs \sca{\vsc}{i}} \land%
  \rep{\scr{\vsc}{i} \robs \scb{\vsc}{i}}) \lor%
  \abs{w_i \fobs \vrt{\vsc}}%
}
\newcommand{\fwdWobsDef}[1]{%
  \abs{w_i \wobs w'_i} #1%
  \wideDefeq \rep{\wra{w_i} \hb \wra{w'_i}}%
}

\newcommand{\fwdAxVrtInScan}[1]{%
  \ensuremath{\forall \abs{s} \in \mathsf{dom}(\SMap{[-]}).} #1%
  {\ensuremath{\SMap{\abs{s}} \subev \abs{s}}}%
}
\newcommand{\fwdAxIO}[1]{%
  \ensuremath{\forall i, \abs{s} \in \term(\abs{S}).} #1%
  {\ensuremath{%
    \exists \abs{w_i}.\ \abs{w_i \robs \SMap{\abs{s}}}%
    \land \abs{w_i}.\evIn = \abs{s}.\evOut[i]}}%
}
\newcommand{\fwdAxWraUniq}[1]{%
  \ensuremath{\forall \rep{e_r} \in \rep{\resA[i].W}.} #1%
  {\ensuremath{\exists \abs{w_i}.\ \rep{\wra{w_i}} = \rep{e_r}}}%
}
\newcommand{\fwdAxScrUniq}[3]{%
  \ensuremath{\forall \vrt{\vsc}, \rep{e_r} \in \rep{\vrt{\vsc}.\resB[i].#1}.} #2%
  \ensuremath{\rep{e_r}.\evIn = \bot} #3%
  \ensuremath{\iff \exists \vrt{\vsc'}.\ \rep{\scr{\vsc'}{i}} = \rep{e_r}}%
}
\newcommand{\fwdAxScTotal}[2]{%
  \ensuremath{\forall \vrt{\vsc}, \vrt{\vsc'} \in \vrt{\Vsc}.} #1%
  \ensuremath{\vrt{\vsc \neq \vsc'}} #2%
  \ensuremath{\implies \vrt{\vsc \rb \vsc'} \lor \vrt{\vsc' \rb \vsc}}%
}
\newcommand{\fwdAxScStruct}[1]{%
  \ensuremath{\forall \vrt{\vsc} \in \vrt{\Vsc}.} #1%
  {\ensuremath{\rep{\scr{\vsc}{i} \rb \sca{\vsc}{i} \rb \scb{\vsc}{i}}}}%
}
\newcommand{\fwdAxFobsUniq}[2]{%
  \ensuremath{\forall \abs{w_i}, \abs{w'_i}, \vrt{\vsc}.} #1%
  \ensuremath{\abs{w_i \fobs \vrt{\vsc}} \land \abs{w'_i \fobs \vrt{\vsc}}} #2%
  \ensuremath{\implies \abs{w_i = w'_i}}%
}
\newcommand{\fwdAxFobsBot}[2]{%
  \ensuremath{\forall \vrt{\vsc}.} #1%
  \ensuremath{\term(\rep{\scb{\vsc}{i}}) \land \neg(\rep{\scr{\vsc}{i} \robs \scb{\vsc}{i}})} #2%
  \ensuremath{\implies \exists \abs{w_i}.\ \abs{w_i \fobs \vrt{\vsc}}}%
}
\newcommand{\fwdAxFobsHb}[2]{%
  \ensuremath{\forall \abs{w_i}, \vrt{\vsc}.} #1%
  \ensuremath{\abs{w_i \fobs \vrt{\vsc}}} #2%
  \ensuremath{\implies \rep{\wra{w_i} \hb \scb{\vsc}{i}} \land \neg(\rep{\scr{\vsc}{i} \robs \scb{\vsc}{i}})}
}
\newcommand{\fwdAxFwdA}[2]{%
  \ensuremath{\forall \abs{w_i}, \vrt{\vsc}.} #1%
  \ensuremath{\rep{\scr{\vsc}{i} \robs \scb{\vsc}{i}} \land \abs{w_i \rbhbeqrobs \vrt{\vsc}}} #2%
  \ensuremath{\implies \rep{\wra{w_i} \hb \sca{\vsc}{i}}}
}
\newcommand{\fwdAxFwdBa}[2]{%
  \ensuremath{\forall \abs{w_i}, \abs{w'_i}, \vrt{\vsc}.} #1%
  \ensuremath{\abs{w_i \fobs \vrt{\vsc}} \land \rep{\wra{w'_i} \hb \scr{\vsc}{i}}} #2%
  \ensuremath{\implies \abs{w_i \not\hb w'_i}}%
}
\newcommand{\fwdAxFwdBb}[2]{%
  \ensuremath{\forall \abs{w_i}, \abs{w'_i}, \vrt{\vsc}.} #1%
  \ensuremath{\abs{w_i \fobs \vrt{\vsc}} \land \abs{w'_i \rbhbeqrobs \vrt{\vsc}}} #2%
  \ensuremath{\implies \abs{w_i \not\hb w'_i}}
}

\begin{figure}
  \centering
  \renewcommand{\arraystretch}{\sigSpacing}%

\axiomset{F}
\signatureHeader{$\fwdSig$}{Snapshot data-structure with forwarding}

\begin{letbox}{l}
  & $\abs{W_i}$ & Set of all writes of cell $i$ of the snapshot
  in history with each $\abs{w_i} \in \writeProc_i$\\
  & $\abs{S}$ & Set of all scans of the snapshot in history\\
\end{letbox}

\begin{tabular*}{0.9\textwidth}{l@{\hskip 1em}A@{\extracolsep{\fill}}r}
  $\bigSum$
  & \vrt{\Vsc} && Set of virtual scans with each $\vrt{\vsc} \in \vscanProc$\\
  & \SMap{[-]} &\wideColon \abs{S} \parfun \vrt{\Vsc} & \ScVscDesc{}\\
  & \abs{\fobs} &\subseteq \bigcup_i \abs{W_i} \times \vrt{\Vsc} & Forwarding visibility
\end{tabular*}\vspace{\sepAfterVar}

\begin{letbox}{A}
  & \fwdRobsDef{&} & Reads-from visibility\\
  & \fwdWobsDef{&} & Writing visibility\\
  & \abs{\obs} &\wideDefeq \abs{\robs} \cup \abs{\wobs} & Visibility relation\\
  & \abs{\hb}  &\wideDefeq (\abs{\rb} \cup \abs{\obs})^+ & Happens-before order\\
\end{letbox}

\begin{tabular*}{0.9\textwidth}{l@{\hskip 0em}A@{\extracolsep{\fill}}r}
  \fwdAxVrtInScan{&\multicolumn{2}{c}}
  &\insertEq\label[prop]{ax:fwd-vrtinscan}\\
  \fwdAxIO{&\multicolumn{2}{c}}
  &\insertEq\label[prop]{ax:fwd-io}\\
  \newsubeqblock
  \fwdAxWraUniq{&\multicolumn{2}{c}}
  &\insertSubeq\label[prop]{ax:fwd-wrauniq}\\
  \fwdAxScrUniq{W}{&}{&}
  &\insertSubeq\label[prop]{ax:fwd-scruniq}\\
  \newsubeqblock
  \fwdAxScTotal{&}{&\ \,}
  &\insertSubeq\label[prop]{ax:fwd-sctotal}\\
  \fwdAxScStruct{&\multicolumn{2}{c}}
  &\insertSubeq\label[prop]{ax:fwd-scstruct}\\[\sepInnerProp]
  \newsubeqblock
  \fwdAxFobsUniq{&}{&\ \,}
  &\insertSubeq\label[prop]{ax:fobsuniq}\\
  \fwdAxFobsBot{&}{&\ \,}
  &\insertSubeq\label[prop]{ax:fobsbot}\\
  \fwdAxFobsHb{&}{&\ \,}
  &\insertSubeq\label[prop]{ax:fobshb}\\[\sepInnerProp]
  \fwdAxFwdA{&}{&\ \,}
  &\insertEq\label[prop]{ax:forward1}\\
  \newsubeqblock
  \fwdAxFwdBa{&}{&\ \,}
  &\insertSubeq\label[prop]{ax:forward2a}\\
  \fwdAxFwdBb{&}{&\ \,}
  &\insertSubeq\label[prop]{ax:forward2b}\\[1mm]
  \bottomrule\\[-0.5em]
\end{tabular*}\vspace{-4mm}

\caption{ %
  Signature representation of the properties of a snapshot
  data-structure that uses forwarding.\vspace{-2mm}
  }
  \label{fig:forward}
\end{figure}

Our axiomatization of forwarding principles is given in
\cref{fig:forward}.
The signature, which we refer to as $\fwdSig$, declares
as inputs the sets of write events $\abs{W_i}$ (for each $i$), and the
set of scan events $\abs{S}$. It then postulates the existence of a
set of virtual scans $\vrt{\Vsc}$, with a mapping $\SMap{[-]}$ from
abs to virtual scans, and a relation $\abs{\fobs}$ for visibility by
forwarding that captures abstractly when a virtual scan observes a
writer by forwarding. As before, these are abstract concepts, known to
satisfy only the properties listed in the scope of $\Sigma$.
We will shortly provide the intuition for the axioms
in~\cref{fig:forward}, but first, let us enumerate how the
axiomatization abstracts from \cref{alg:jay1}, so that the same
signature applies to all three algorithms.
\begin{itemize}
\item The axiomatization does not assume that we have rep write events
  for forwarding ($\rep{\wrb{w_i}}$). Instead, it makes forwarding
  abstract by tying it to the relation $\abs{\fobs}$, allowing for
  multiple different forwarding methods.
\item The axiomatization links virtual scans to abs scans by the
  function $\SMap{[-]}$, to support multiple scanners. Physical scans
  can overlap in time, as long as the virtual scans do not.
\item The axiomatization does not assume any rep events operating with
  memory cell $\resX$, or for that matter, that such a cell even
  exists. Since $\resX$ is solely used to communicate to writers if and
  when to forward, we instead encode the conditions for forwarding in
  the axioms. This is necessary because the three algorithms decide
  differently on whether to forward or not.
\end{itemize}

\subsubsection{Intuition Behind the Forwarding Signature}
\label{subsec:overview-principles}

Out of the forwarding visibility relation $\abs{\fobs}$, we define
two additional visibility relations over abstract and virtual
events: $\abs{\robs}$ and $\abs{\wobs}$ as shown in
\cref{fig:forward}. These relations define what we will use as an
abstract notion of a virtual scan observing a writer directly or by
forwarding ($\abs{\robs}$), and a writer observing a prior write
($\abs{\wobs}$), respectively. In each algorithm, the $\abs{\fobs}$
relation will be defined differently, but the definitions of
$\abs{\robs}$ and $\abs{\wobs}$ are constant in terms
of $\abs{\fobs}$. We next explain each of them.
\[
  \fwdRobsDef{}
\]
Read visibility ($\abs{\robs}$) captures that a virtual scan
$\vrt{\vsc}$ either reads the value of $\abs{w_i}$ by forwarding
(disjunct $\abs{w_i \fobs \vrt{\vsc}}$) or directly. The direct read
requires that the scanner finds the value in $\resA[i]$ (conjunct
$\rep{wra_i \robs \sca{\vsc}{i}}$) and that the scanner finds $\bot$
for the forwarded value in $\resB[i]$ (conjunct
$\rep{\scr{\vsc}{i} \robs \scb{\vsc}{i}}$).
\[
  \fwdWobsDef{}
\]
Write visibility ($\abs{\wobs}$) orders the abs writers of a common
memory cell, and is defined by the ordering of underlying rep writers
$\rep{\wra{w_i}}$. Ordering abs writers this way works because the rep
write is the unique point where the abs writer communicates its value
to the data structure. This holds for all three algorithms including
the multi-writer variants.

We next explain the forwarding properties. These are divided into
three groups, as shown in \cref{fig:forward}. The first group consists
of the structural properties related to writers and scanners.
\[
  \renewcommand{\arraystretch}{\eqSpacing}%
  \begin{tabular*}{\textwidth}{@{\hskip 2em}l@{\hskip 5em}c@{\extracolsep{\fill}}r}
    \fwdAxVrtInScan{&} &(\ref{ax:fwd-vrtinscan} revisited)\\
    \fwdAxIO{&} &(\ref{ax:fwd-io} revisited)
  \end{tabular*}
\]
\cref{ax:fwd-vrtinscan} states that each virtual scan must be inside
the interval of the abs scan it represents.  \cref{ax:fwd-io} states
that a terminated abs scan $\abs{s}$ must have a virtual scan
representative that, for each index $i$, observed some write
$\abs{w_i}$ such that the value written by $\abs{w_i}$ is the one the
scan $\abs{s}$ returned for $i$.  \cref{ax:fwd-io} is similar to
\cref{ax:ss-io}, and essentially says that the observing scan correctly
reads the array. On the other hand, \cref{ax:fwd-vrtinscan} is similar
to a property of linearization points, whereby a linearization point
must reside within the interval of the considered event. Here instead,
we have a whole virtual scan $\SMap{\abs{s}}$ representing when
$\abs{s}$ logically occurred.
\[
  \renewcommand{\arraystretch}{\eqSpacing}%
  \begin{tabular*}{\textwidth}{@{\hskip 2em}l@{\hskip 9em}A@{\extracolsep{\fill}}r}
    \fwdAxWraUniq{&\multicolumn{2}{c}} &(\ref{ax:fwd-wrauniq} revisited)\\
  \end{tabular*}
\]
\cref{ax:fwd-wrauniq} ensures that the only rep event that can write
into $\resA[i]$ is $\rep{\wra{w_i}}$ of some abs write $\abs{w_i}$.
This means that only abs writes into cell $i$ can have a rep event for
writing into $\resA[i]$; scans and writes into a cell different from
$i$ cannot, since if they did, they would have to be equal to some
$\abs{w_i}$ as per Property~\ref{ax:evsig-uniq} of event signatures.
Also, an abs write can write into $\resA[i]$ at most once, since any
other write must equal $\rep{\wra{w_i}}$.
\[
  \renewcommand{\arraystretch}{\eqSpacing}%
  \begin{tabular*}{\textwidth}{@{\hskip 2em}l@{\hskip 5em}A@{\extracolsep{\fill}}r}
    \fwdAxScrUniq{W}{&}{&} &(\ref{ax:fwd-scruniq} revisited) 
  \end{tabular*}
\]
The easiest way to explain~\cref{ax:fwd-scruniq} is to consider its
instantiation for \cref{alg:jay1}, where virtual scanners are abs
scanners, which, moreover, all share the same array $\resB$.
\mytag{F.3b'}%
\begin{equation}
  \renewcommand{\arraystretch}{\eqSpacing}%
  \hfill%
  \begin{tabular*}{\textwidth}{@{\hskip 2em}l@{\hskip 7em}A@{\extracolsep{\fill}}r}
    $\forall \rep{e_r} \in \rep{\resB[i].W}.$
    & \rep{e_r}.\evIn = \bot
    & \iff \exists \abs{\vsc'}.\ \rep{\scr{\vsc'}{i}} = \rep{e_r}
    &\insertEq\label[prop]{ax:fwd-scruniq2}
  \end{tabular*}\notag
\end{equation}
The simplified~\cref{ax:fwd-scruniq2} says that the only rep event
that can write $\bot$ into $\resB[i]$ is $\rep{\scr{\vsc'}{i}}$ of
some \abs{abs} scan $\abs{\vsc'}$. Thus, no event except scanners can
write $\bot$ into $\resB[i]$, and the writing can be done at most
once. Additionally, the rep event $\rep{\scr{\vsc}{i}}$ can
\emph{only} write $\bot$ into $\resB[i]$.  \cref{ax:fwd-scruniq}
generalizes \cref{ax:fwd-scruniq2} by allowing for virtual scanners,
each of which has their own $\resB$ array.  
\[
\begin{tabular*}{\textwidth}{@{\hskip 2em}l@{\hskip 7em}A@{\extracolsep{\fill}}r}
  \fwdAxScTotal{&}{&} &(\ref{ax:fwd-sctotal} revisited)
\end{tabular*}
\]
\cref{ax:fwd-sctotal} captures that virtual scans never overlap, and
are thus totally ordered. This property holds for single-scanner
algorithms by assumption, but must be proved in the multi-scanner
case.
\[
  \begin{tabular*}{\textwidth}{@{\hskip 2em}l@{\hskip 11em}c@{\extracolsep{\fill}}r}
    \fwdAxScStruct{&} &(\ref{ax:fwd-scstruct} revisited)
  \end{tabular*}
\]
\cref{ax:fwd-scstruct} simply reflects that the rep events of initializing
the scan for element $i$, reading $\resA[i]$, and then closing the
scanning for element $i$ by reading $\resB[i]$, are invoked
sequentially in the code of $\scanProc$. This is readily visible
in~\cref{alg:jay1} and remains true in \cref{alg:jay2,alg:jay3}.

The second group are the structural properties of the forwarding
visibility.
\[
  \renewcommand{\arraystretch}{\eqSpacing}%
  \begin{tabular*}{\textwidth}{@{\hskip 2em}l@{\hskip 2em}A@{\extracolsep{\fill}}r}
    \fwdAxFobsUniq{&}{&} &(\ref{ax:fobsuniq} revisited)\\
    \fwdAxFobsBot{&}{&} &(\ref{ax:fobsbot} revisited)\\
    \fwdAxFobsHb{&}{&} &(\ref{ax:fobshb} revisited)
  \end{tabular*}
\]
\cref{ax:fobsuniq} captures that a scan can observe at most
one forwarded write for a given index $i$.
\cref{ax:fobsbot} says that a forwarding of some write of cell $i$
will reach the virtual scan $\vrt{\vsc}$ if $\rep{\scb{\vsc}{i}}$ did
not observe the writing of $\bot$ in $\resB[i]$ by
$\rep{\scr{\vsc}{i}}$. Contrapositively, 
$\rep{\scr{\vsc}{i} \robs \scb{\vsc}{i}}$ holds if no write of cell $i$ was
forwarded to $\vrt{\vsc}$.
\cref{ax:fobshb} states the dual implication direction of
\cref{ax:fobsbot} and that a forwarded write $\abs{w_i}$ must have
written into $\resA[i]$ by $\rep{\wra{w_i}}$ before the scanner
$\vrt{\vsc}$ read the forwarded value by $\rep{\scb{\vsc}{i}}$. This
property captures that writers must first write their value directly
before attempting to forward.

The third group are the properties that capture the forwarding
principles of Jayanti. We have originally discovered these properties
by extracting the common patterns from our proofs of the snapshot
signature for the three algorithms, and only afterwards discovered
that they actually correspond quite closely to Jayanti's forwarding
principles. To describe how our axioms capture the forwarding
principles, we state the principles below in English, verbatim as Jayanti does,
but using our notation. We also use virtual scans instead of abstract
scans to make the connection to multi-scanner algorithm direct;
Jayanti only stated the principles in terms of \cref{alg:jay1}. As we
shall see, our axiomatization modifies the principles slightly to
encompass \cref{alg:jay2,alg:jay3}.
We also highlight and number subsentences so that we can relate them
to our axioms in discussion.
\begin{enumerate}
\item Suppose that \texttag{1}{a scan operation $\vrt{\vsc}$ misses a
    write operation $\abs{w_i}$ writing $v$ because $\vrt{\vsc}$ reads
    $\resA[i]$ before $\abs{w_i}$ writes in $\resA[i]$}. If
    \texttag{2}{$\abs{w_i}$ completes before $\vrt{\vsc}$ performs
    $\rep{\scoff{s}}$}, then \texttag{3}{$\abs{w_i}$ will have surely
    informed $\vrt{\vsc}$ of $v$ by writing $v$ in $\resB[i]$}.

\item Suppose that \texttag{4}{a scan operation $\vrt{\vsc}$ reads at
    $\rep{\scb{\vsc}{i}}$ a non-$\bot$ value in $\resB[i]$ by some
    write operation $\abs{w_i}$}. Then (a) \texttag{5}{$\abs{w_i}$ is
    concurrent with $\vrt{\vsc}$}, and (b) If \texttag{6}{$\abs{w'_i}$
    is any write operation that is executed after $\abs{w_i}$}, then
    \texttag{7}{$\abs{w'_i}$ completes only after $\vrt{\vsc}$ performs
    $\rep{\scoff{s}}$}.
\end{enumerate}

We start with \cref{ax:forward1}, which relates to Principle~(1) as
follows.
\[
  \begin{tabular*}{\textwidth}{@{\hskip 2em}l@{\hskip 5em}A@{\extracolsep{\fill}}r}
    \fwdAxFwdA{&}{&} &(\ref{ax:forward1} revisited)
  \end{tabular*}
\]
Here, $\rep{\scr{\vsc}{i} \robs \scb{\vsc}{i}}$ corresponds to the
negation of Statement~(\rom{3}), $\abs{w_i \rbhbeqrobs \vrt{\vsc}}$
corresponds to Statement~(\rom{2}), and $\rep{\wra{w_i} \hb
  \sca{\vsc}{i}}$ corresponds to the negation of Statement~(\rom{1}),
altogether combining into an equivalent, by contraposition, of
Principle~(1). We now explain each correspondence individually.
\begin{itemize}
\item $\rep{\scr{\vsc}{i} \robs \scb{\vsc}{i}}$ states that nothing
  was forwarded to $\vrt{\vsc}$,
 because the event $\rep{\scb{\vsc}{i}}$ of reading $\resB[i]$
 observes the event $\rep{\scr{\vsc}{i}}$ of clearing $\resB[i]$.
 This is a slightly stronger statement than the negation of
 Statement~(\rom{3}), namely, that the write $\abs{w_i}$ was not
 forwarded to $\vrt{\vsc}$. The difference between these statements
 will be covered by \cref{ax:forward2b}, which will handle the
 scenario where another write, not $\abs{w_i}$, is forwarded to
 $\vrt{\vsc}$.
\item $\abs{w_i \rbhbeqrobs \vrt{\vsc}}$, where
  $\abs{\rb \hbeq \robs}$ is the relational composition of $\abs{\rb}$
  and $\abs{\hbeq}$ and $\abs{\robs}$, roughly translates to
  ``$\abs{w_i}$ terminates before some other writer (possibly writing
  to a different memory cell) that in turn was observed by
  $\vrt{\vsc}$''. We know such other writer exists, since the domain
  of $\abs{\robs}$, the last relation in the composition, consists
  only of writers. 
  While this looks nothing like Statement~(\rom{2}), both statements
  imply that $\abs{w_i}$ must have been forwarded if the writing to
  $\resA[i]$ was missed.
  \ifappendix%
  We prove this with different proofs for each algorithm, with
  Lemma~\apndxLemJayILin{} for \cref{alg:jay1} and
  Lemma~\apndxLemJayIILin{} for \cref{alg:jay2,alg:jay3}.
  \else%
  We prove this with different proofs for each algorithm, which we
  give in the appendix \cite{extended}, with Lemma~\apndxLemJayILin{}
  for \cref{alg:jay1} and Lemma~\apndxLemJayIILin{} for
  \cref{alg:jay2,alg:jay3}.
  \fi%
  We use the statement $\abs{w_i \rbhbeqrobs \vrt{\vsc}}$, instead
  of potentially another statement involving $\rep{\scoff{s}}$ (as
  Jayanti's original statement does), because it makes the forwarding
  signature export fewer rep events, thus making it a bit more
  parsimonious.
\item $\rep{\wra{w_i} \hb \sca{\vsc}{i}}$ states that
  $\rep{\wra{w_i}}$ wrote to $\resA[i]$ before $\rep{\sca{\vsc}{i}}$
  read from it, which directly corresponds to the negation of
  Statement~(\rom{1}).
\end{itemize}

For Principle~(2a), we need to further deviate from Jayanti's original
formulation, since the latter is too specific to \cref{alg:jay1}, and
does not scale as-is to the multi-writer algorithms. The point of this
principle is to ensure that a writer may only forward values that are,
intuitively, ``current''. 
The principle itself ensures this property, but only in the
single-writer case of \cref{alg:jay1}. Indeed, in \cref{alg:jay1}, a
writer only forwards its own value (and other writes to the same
pointer are prohibited by assumption); therefore, if a write is
concurrent to the scan, then it can only forward a value that is
current. But in a multi-writer case, as we shall see
in~\cref{sec:multiwriter}, we want to allow a writer to forward values
of other writers. In particular, a writer that writes and then stalls
indefinitely could still be forwarded by another write to some
concurrent scan. But we still need to ensure that the forwarded write
is not arbitrarily old. We do so by insisting that a write $\abs{w_i}$
forwarded to $\vrt{\vsc}$ cannot occur before a write $\abs{w'_i}$
which wrote to $\resA[i]$ with $\rep{\wra{w_i}}$ before $\vrt{\vsc}$
started, i.e., executed $\rep{\scr{\vsc}{i}}$. In other words, we
require that the forwarded write $\abs{w_i}$ does not ``happen
before'' any write $\abs{w'_i}$ that started before the start of the
scan. A write $\abs{w'_i}$ has a \emph{newer} value, so its existence
should prevent $\abs{w_i}$ from being forwarded. We formalize this as
\cref{ax:forward2a}.
\[
  \begin{tabular*}{\textwidth}{@{\hskip 2em}l@{\hskip 5em}A@{\extracolsep{\fill}}r}
    \fwdAxFwdBa{&}{&} &(\ref{ax:forward2a} revisited)
  \end{tabular*}
\]
This captures that writers must forward to scan $\vrt{\vsc}$ either
the latest write that wrote to $\resA[i]$ before
$\rep{\scr{\vsc}{i}}$, or some write that occurred after
$\rep{\scr{\vsc}{i}}$, thus ensuring that the forwarded write is
current.

Lastly, \cref{ax:forward2b} relates to the contrapositive of
Principle~(2b).
\[
  \begin{tabular*}{\textwidth}{@{\hskip 2em}l@{\hskip 5em}A@{\extracolsep{\fill}}r}
    \fwdAxFwdBb{&}{&} &(\ref{ax:forward2b} revisited)
  \end{tabular*}
\]
Again, $\abs{w_i \fobs \vrt{\vsc}}$ directly corresponds to
Statement~(\rom{4}). Formal statement $\abs{w'_i \rbhbeqrobs
  \vrt{\vsc}}$ corresponds to the negation of Statement~(\rom{7}), for
the same reason as in \cref{ax:forward1}, since (\rom{2}) is the
negation of (\rom{7}). Finally, $\abs{w_i \not\hb w'_i}$ directly
corresponds to the negation of~(\rom{6}), with the exception that we
use $\hb$ instead of Jayanti's ``executed after'' to relate the writes
$\abs{w_i}$ and $\abs{w'_i}$. The two correspond in the single-writer
case, but our statement extends to the multi-writer case as well.

\cref{alg:jay1} satisfies this set of properties of the $\fwdSig$
signature, and they are sufficient to prove that an algorithm is a
snapshot algorithm according to the $\snapshotSig$ signature, which is
in turn sufficient to prove linearizability. In
\cref{sec:proofsketches} we sketch why these claims holds, and in
\cref{sec:multiwriter} we sketch how \cref{alg:jay2,alg:jay3} satisfy
the $\fwdSig$ signature.
\ifappendix
The complete proofs are available in the appendix.
\else
The complete proofs are available in the appendix of the extended
version \cite{extended}.
\fi

\section{Proof Sketches for Algorithm \ref{alg:jay1}}\label{sec:proofsketches}

We now sketch the proofs of our signatures implying linearizability,
to illustrate the gist of reasoning by visibility.
We begin by showing how histories that satisfy the $\snapshotSig$
signature are linearizable. This proof is based on a similar proof by
\citet{henzinger:lmcs15} for queues.

\begin{genthm}{\cref{lem:ss-lin}}
  Histories satisfying $\snapshotSig$~signature
  (\cref{fig:snapshot}) are linearizable.
\end{genthm}
\begin{proof}[Proof sketch]
  Following \cref{def:lin} of linearizability, we start by choosing a
  visibility relation $\abs{\obs}$ that constructs the set $\abs{E_c}
  = \{ \abs{e} \mid \exists \abs{e'} \in \term(\abs{E}).\ \abs{e
    \obs^* e'} \}$, and then extend it to a linearization
  $\abs{<}$. Let $\abs{\obs}$ be the visibility relation obtained by
  restricting the visibility relation postulated by $\snapshotSig$ to
  $\term(\abs{S}) \cup \bigcup_{i} \abs{\WrEff_i}$. Events outside of
  the restriction are non-terminated scans or writes that have not
  executed their effect yet; they do not modify the abstract state,
  and are hence not necessary for linearization.
Additionally, by \cref{ax:ss-wrterm} we know that every terminated
write is effectful. Thus going forward, we consider $\abs{E_c} =
\term(\abs{S}) \cup \bigcup_{i} \abs{\WrEff_i}$.

  We next extend the happens-before order $\abs{\hb}$ on the selected
  set of events as follows. The idea is to keep extending this order
  with more relations between the events of the selected set until we
  reach a total order, which will be the desired linearization order.
  We first induce a helper order
  \[
    \abs{\whb} \wideDefeq (\abs{\hb}\ \cup\ \abs{\whbI})^+
  \]
  which ranges over writes (into different cells), where
  $\abs{w_i \whbI w'_j}$ holds if $i \neq j$ and there exists a write
  $\abs{w_j}$ and a scan $\abs{s}$ such that $\abs{w_i,w_j \robs s}$
  and $\abs{w_j \hb w'_j}$. Or, in English, we order $\abs{w_i}$
  before $\abs{w'_j}$, if $\abs{w_i}$ is observed by some scan along
  with $\abs{w_j}$, and $\abs{w_j}$ is ordered before $\abs{w'_j}$.
  The $\abs{\whb}$ order is a (strict) partial order by
  \cref{ax:ss-io,ax:ss-nowrbetween,ax:ss-robsuniq,ax:ss-wrtotal,ax:ss-mono}
  (see Lemma~\apndxLemWhbIrrefl{} \ifNotAppendix{from~\cite{extended}}
  for the complete proof).

  Next, we select an arbitrary total order $\abs{\wlt}$ over writes
  that extends $\abs{\whb}$. Since the set of selected events is
  finite, such a total order always exists by Zorn's Lemma. We use the
  order $\abs{\wlt}$ to determine which event we will consider as
  being last in the eventually constructed linearization order. We say
  that a write is latest-in-time if it is the greatest in
  $\abs{\wlt}$, while a scan is latest-in-time if it is maximal in
  $\abs{\hb}$ and all the writes that it observes are the greatest in
  $\abs{<_w}$ for their respective memory cell. As long as the set of
  events is non-empty, there will exist a latest-in-time event under
  this definition by
  \cref{ax:ss-io,ax:ss-nowrbetween,ax:ss-robsuniq,ax:ss-wrtotal} (see
  Lemma~\apndxLemMaxcand{} \ifNotAppendix{from~\cite{extended}} for
  the complete proof). By inductively selecting the latest-in-time
  event, we construct a total order of events that is consistent with
  $\abs{\hb}$ and snapshot semantics, and is thus a linearization
  order. More concretely, we add the latest-in-time event to the end
  of the linearization order, forming a chain with previously added
  events, and repeat this step with the set of events excluding the
  last selected.
\end{proof}

Next, we show that histories satisfying the $\fwdSig$ signature also
satisfy the $\snapshotSig$ signature.
Because some $\fwdSig$ properties use rep events, to be able to
efficiently use them, we first need a lemma that turns a $\abs{\hb}$
relation between virtual or abs events into a $\rep{\hb}$ relation
between rep events. The lemma will also be useful in showing that
\cref{alg:jay1} (and \cref{alg:jay2,alg:jay3} in
\cref{sec:multiwriter}) satisfy the $\fwdSig$ signature, once
\cref{ax:fwd-scstruct,ax:fobshb} have been proven.

\begin{lemma}[Happens-before of subevents]\label{lem:fwd-hbabsrep}
  Assume \cref{ax:fwd-scstruct,ax:fobshb}. For events
  $\abs{e}, \abs{e'} \in \bigcup_i \abs{W_i} \cup \vrt{\Vsc}$ of the
  $\fwdSig$ signature (\cref{fig:forward}) where $\abs{e'}$ is
  populated with at least one rep event from its event structure
  (\cref{fig:struct}), if $\abs{e \hb e'}$ holds (where $\abs{\hb}$
  follows the definition from \cref{fig:forward}) then there exists
  $\rep{e_r}$ and $\rep{e'_r}$ where $\rep{e_r \hb e'_r}$ holds, with
  $\rep{e_r}$ and $\rep{e'_r}$ belonging to $\abs{e}$ and $\abs{e'}$
  respectively. 
\end{lemma}

\begin{proof}%
  We can view $\abs{e \hb e'}$ as a chain of $\abs{\hb_1}$ relations
  connecting $\abs{e}$ and $\abs{e'}$; that is,
  $\abs{e \hb_1 \cdots \hb_1 e'}$, where by definition
  $\abs{\hb_1} = (\abs{\rb} \cup \abs{\obs}) = (\abs{\rb} \cup
  \abs{\wobs} \cup \abs{\robs})$. It is easy to see that each abs
  event in this chain is populated with a rep event. Indeed, if an abs
  event's relation $\abs{\hb_1}$ to its successor is realized by
  $\abs{\rb}$, then the event is terminated and all of its rep events
  are executed (and by our axioms, each abs event must have at least
  one rep event). Alternatively, if the abs event's relation
  $\abs{\hb_1}$ to its successor is realized by $\abs{\obs}$, then the
  event must be populated as per the definitions of reads-from and
  writing visibility in \cref{fig:forward}, and in the $\abs{\fobs}$
  case by \cref{ax:fobshb}. This leaves out $\abs{e'}$, which has no
  successor, but $\abs{e'}$ is populated by assumption. 
  
  Now the proof is by induction on the length of the above chain. The
  base case is when $\abs{e \hb_1 e'}$.
  If $\abs{e \rb e'}$, the proof is by \cref{ax:rb-absrep} giving
  us $\rep{e_r \rb e'_r}$ for arbitrary $\rep{e_r}$ and $\rep{e'_r}$
  belonging to $\abs{e}$ and $\abs{e'}$ respectively. If
  $\abs{e \wobs e'}$, the proof is by the definition of
  $\abs{\wobs}$ in \cref{fig:forward}, giving us
  $\rep{\wra{w_i} \hb \wra{w'_i}}$ with $\abs{e} = \abs{w_i}$ and
  $\abs{e'} = \abs{w'_i}$. For the last case, $\abs{e \robs e'}$, we
  have $\abs{e } = \abs{w_i}$ and $\vrt{\vsc'} = \abs{e'}$ with either
  $\rep{\wra{w_i} \robs \sca{\vsc'}{i}}$ or
  $\abs{w_i \fobs \vrt{\vsc'}}$. The former case directly exhibits rep
  events that satisfy our goal. The latter does so as well, since by
  \cref{ax:fobshb}, it must be $\rep{\wra{w_i} \hb \scb{\vsc'}{i}}$.
  The inductive step is similar, with the addition that we need
  transitivity of $\rep{\hb}$ and \cref{ax:fwd-scstruct} to join
  events. 
\end{proof}

\begin{lemma}\label{lem:fwd-ss}
  Histories satisfying $\fwdSig$ signature (\cref{fig:forward})
  also satisfy $\snapshotSig$ signature. %
\end{lemma}
\begin{proof}[Proof sketch]
  We start by determining the instantiations of $\abs{\WrEff_i}$,
  $\abs{\robs}$, and $\abs{\obs}$ for the $\snapshotSig$ signature.
  Since the effect of $\abs{w_i}$ is observable only after the
  execution of $\rep{\wra{w_i}}$, we let
  $\abs{w_i} \in \abs{\WrEff_i}$ iff $\rep{\wra{w_i}}$ is defined. For
  the other notions, we first define a helper visibility over scans
  $\abs{\sobs}$, which we use to define $\abs{\obs}$. The highlighted
  relations $\abs{\hlrobs}$ and $\abs{\hlwobs}$ originate from the
  $\fwdSig$ signature.
  \begin{align*}
    \abs{w_i \robs s} &\wideDefeq \abs{w_i \hlrobs \SMap{\abs{s}}} &
    \abs{s \sobs s'} &\wideDefeq \vrt{\SMap{\abs{s}} \rb \SMap{\abs{s'}}} &
    \abs{\obs} &\wideDefeq \abs{\robs} \cup \abs{\hlwobs} \cup \abs{\sobs}
  \end{align*}

  \cref{ax:wfobs} holds since if $\abs{e \obs^+ e'}$ and
  $\abs{e' \rbeq e}$, then we can construct a $\rep{\hb}$-cycle of rep
  events with \cref{lem:fwd-hbabsrep}, contradicting irreflexivity of
  $\rep{\hb}$ (\cref{lem:hb-irrefl}). \cref{ax:ss-io} holds directly
  by \cref{ax:fwd-io}. \cref{ax:ss-robsuniq} holds because by the
  definition of $\abs{w_i \hlrobs \vrt{\vsc}}$ in $\fwdSig$ signature,
  it is either $\rep{\wra{w_i} \robs \sca{\vsc}{i}}$ for which
  \cref{ax:mem-robsuniq} ensures uniqueness, or
  $\abs{w_i \fobs \vrt{\vsc}}$ where \cref{ax:fobsuniq} ensures
  uniqueness. \cref{ax:ss-wrtotal} of total ordering on effectful
  writes follows because in $\fwdSig$ signature,
  $\abs{w_i \hlwobs w'_i}$ is defined as
  $\rep{\wra{w_i} \hb \wra{w'_i}}$, and the rep writes are totally
  ordered by \cref{ax:mem-wrtotal}. \cref{ax:ss-wrterm} that all
  terminated writes are effectful holds because for every terminated
  write $\abs{w_i}$ we must have $\rep{\wra{w_i}}$ defined.

  We next prove \cref{ax:ss-nowrbetween}. We sketch the proof by
  assuming there is a write $\abs{w'_i}$ such that
  $\abs{w_i \hb w'_i \hb s}$ where $\abs{w_i \robs s}$, and then
  deriving a contradiction. Unfolding the definition of $\abs{\robs}$
  in $\abs{w_i \robs s}$, we get $\abs{w_i \hlrobs \vrt{\vsc}}$ where
  $\vrt{\vsc} = \SMap{\abs{s}}$. Also, from $\abs{w'_i \hb s}$, using
  \cref{ax:fwd-vrtinscan}, it turns out that
  $\abs{w'_i \hb \vrt{\vsc}}$, and moreover, the latter splits into
  two possible cases: $\abs{w'_i \rbhbeqrobs \vrt{\vsc}}$ and
  $\rep{\wra{w'_i} \hb \scr{\vsc}{i}}$. The full proof of how the
  cases follow from the premise and that they exhaust the
  possibilities is given in Appendix~\apndxFwd{}\ifNotAppendix{
    from~\cite{extended}}. Note that the proposition
  $\abs{w'_i \rbhbeqrobs \vrt{\vsc}}$ is the key part of
  \cref{ax:forward1,ax:forward2b}, while
  $\rep{\wra{w'_i} \hb \scr{\vsc}{i}}$ is the key part of
  \cref{ax:forward2a}; in fact, this exhaustion of cases of
  $\abs{w'_i \hb \vrt{\vsc}}$ is what let us rediscover and
  mathematically formulate the forwarding principles.
  The proof now proceeds by analyzing $\abs{w_i \hlrobs \vrt{\vsc}}$.
  From the definition of reads-from visibility in \cref{fig:forward},
  we also get two cases:
  $(\rep{\wra{w_i} \robs \sca{\vsc}{i}} \land \rep{\scr{\vsc}{i} \robs
    \scb{\vsc}{i}})$ and $\abs{w_i \fobs \vrt{\vsc}}$, for a total of
  four cases when combined with the above. In each case, contradiction
  follows easily, relying in different cases on a different subset of
  \cref{ax:forward1,ax:forward2a,ax:forward2b,ax:mem-nowrbetween}.

   Finally, we prove \cref{ax:ss-mono};
   that is, $\abs{w_i, w_j \robs s}$ and $\abs{w'_i, w'_j \robs s'}$
   with $\abs{w_i \hb w'_i}$ and $\abs{w'_j \hb w_j}$ lead to
   contradiction. As before, we first transform the assumptions
   $\abs{w_i, w_j \robs s}$ and $\abs{w'_i, w'_j \robs s'}$ into
   $\abs{w_i, w_j \robs \vrt{\vsc}}$ and
   $\abs{w'_i, w'_j \robs \vrt{\vsc'}}$, where
   $\vrt{\vsc} = \SMap{\abs{s}}$ and $\vrt{\vsc'} = \SMap{\abs{s'}}$.
   Next, by property \cref{ax:fwd-sctotal}, the virtual scans are
   totally ordered by $\vrt{\rb}$, thus either $\vrt{\vsc \rb \vsc'}$
   or $\vrt{\vsc' \rb \vsc}$ (the case when $\vrt{\vsc} = \vrt{\vsc'}$
   contradicts reads-from uniqueness \cref{ax:ss-robsuniq}). In the first
   case (the second is symmetric), we have
   $\abs{w'_j \hb w_j \robs \vrt{\vsc \rb \vsc'}}$, and thus also
   $\abs{w'_j \hb w_j \hb \vrt{\vsc'}}$. In other words, we have a
   scan $\vrt{\vsc'}$ observing a write $\abs{w'_j}$ with another
   intervening write $\abs{w_j}$ showing up in between. But this
   contradicts the argument we carried out for
   \cref{ax:ss-nowrbetween} above.
\end{proof}

\begin{lemma}\label{lem:jay1-fwd}
  Every execution of \cref{alg:jay1} satisfies the $\fwdSig$ signature
  (\cref{fig:forward}).
\end{lemma}
\begin{proof}[Proof sketch]
  We instantiate the set of virtual scans to be the same as the set of
  abs scans (since \cref{alg:jay1} is single-scanner), and define
  $\SMap{[-]}$ by $\SMap{\abs{s}} = \abs{s}$. \cref{ax:fwd-vrtinscan},
  requiring $\SMap{\abs{s}} \in \abs{s}$, follows immediately because
  every event $\abs{e}$ satisfies $\abs{e} \subev \abs{e}$. Since
  \cref{alg:jay1} is a single-scanner algorithm, the total order of
  virtual scans \cref{ax:fwd-sctotal} naturally follows.
  
  We instantiate forwarding visibility
  $\abs{w_i \fobs s} \defeq \rep{\wrb{w_i} \robs \scb{s}{i}}$, to
  directly capture how forwarding is defined for \cref{alg:jay1}.
  Properties~\eqref{ax:fwd-io}, \eqref{ax:fwd-wrauniq},
  \eqref{ax:fwd-scruniq} and \eqref{ax:fwd-scstruct}
  hold by the structure of the algorithm, while
  \cref{ax:fobsuniq,ax:fobsbot,ax:fobshb} follows from the
  $\abs{\fobs}$ instantiation (proved in
  Appendix~\apndxJayI{}\ifNotAppendix{ of \cite{extended}}).
  For the remaining Properties~\eqref{ax:forward1} to
  \eqref{ax:forward2b}, we need two helper lemmas (also proved in
  Appendix~\apndxJayI{}):
  \begin{itemize}
  \item Lemma~\apndxLemJayIObsX{}: $\rep{\scon{s} \robs \wrx{w_i}}$ iff we
    have $\rep{\scon{s} \hb \wrx{w_i}}$ and
    $\rep{\scoff{s} \not\hb \wrx{w_i}}$.
  \item Lemma~\apndxLemJayILin{}: If $\abs{w'_i \rbhbeqrobs s}$
    then $\rep{\scoff{s} \not\hb \wrx{w'_i}}$ and if
    $\rep{\wrb{w'_i}}$ was executed, then
    $\rep{\wrb{w'_i} \hb \scb{s}{i}}$.
  \end{itemize}
  We focus here on proving \cref{ax:forward1} which involves showing if
  $\abs{w'_i \rbhbeqrobs s}$ and $\rep{\scr{s}{i} \robs \scb{s}{i}}$
  then $\rep{\wra{w'_i} \hb \sca{s}{i}}$.
By Lemma~\apndxLemJayILin{}, we have
  $\rep{\scoff{s} \not\hb \wrb{w'_i}}$ and
  $\rep{\wrb{w'_i} \hb \scb{s}{i}}$ if $\rep{\wrb{w'_i}}$ was
  executed. Consider if we have $\rep{\scon{s} \hb \wrx{w'_i}}$, then
  by Lemma~\apndxLemJayIObsX{} we have $\rep{\scon{s} \robs \wrx{w'_i}}$
  which in turn implies that $\rep{\wrb{w'_i}}$ must have executed,
  however that contradicts \cref{ax:mem-nowrbetween} by
  $\rep{\wrb{w'_i}}$ occurring in between
  $\rep{\scr{s}{i} \robs \scb{s}{i}}$, thus we must have
  $\rep{\scon{s} \not\hb \wrx{w'_i}}$. By \cref{ax:interval} over
  $\rep{\wra{w'_i} \rb \wrx{w'_i}}$ and
  $\rep{\scon{s} \rb \sca{s}{i}}$, since
  $\rep{\scon{s} \not\hb \wrx{w'_i}}$ contradicts
  $\rep{\scon{s} \rb \wrx{w'_i}}$, we must have
  $\rep{\wra{w'_i} \rb \sca{s}{i}}$, which is our goal.
  \end{proof}

\section{Multi-Writer Algorithms}\label{sec:multiwriter}

To allow for correct non-blocking algorithms with multiple concurrent
writers, as well as for multiple concurrent scanners in the case of
\cref{alg:jay3}, Jayanti uses $\LLins$ (load-link), $\SCins$
(store-conditional) and $\VLins$ (validate) operations
\cite{jensen:1987} in addition to simple memory mutation and
dereference. We thus need to extend the register signature to account
for the new operations. Additionally, \cref{alg:jay2,alg:jay3}, both
being multi-writer algorithms share more structure than what we
capture in the $\fwdSig$ signature of \cref{fig:forward}. We thus
introduce a new $\mwFwdSig$ signature to capture the commonality of
the multi-writer algorithms. We then show that histories satisfying
$\mwFwdSig$ also satisfy $\fwdSig$ and that \cref{alg:jay2,alg:jay3}
satisfy $\mwFwdSig$.

\subsection{LL/SC Registers}

For a memory cell $\resX$, $\LLins(\resX)$, or load-link, reads from,
and returns the value of $\resX$. It also records the time of the
read, for use in future $\SCins$ and $\VLins$ operations of the same
thread. 
$\SCins(\resX, v)$, or store-conditional, writes $v$ into $\resX$ and
returns $\true$ if no other write into $\resX$ occurred since the most
recent $\LLins(\resX)$ from the same thread. Otherwise, $\SCins$
returns $\false$, keeping $\resX$ unchanged. Finally, $\VLins(\resX)$,
or validate, returns truth values identically to $\SCins$, but does
not mutate $\resX$.

\begin{figure}
  \centering
  \renewcommand{\arraystretch}{\sigSpacing}%

\axiomset{M\textsuperscript{+}\hspace{-2pt}}
\signatureHeader{$\llregSig$}{$\LLins$/$\SCins$ Register}

\begin{letbox}{A}
  & \multicolumn{2}{l}{$\rep{W}, \rep{R}, \rep{\LL}, \rep{\SC}, \rep{\VL}$} & Set of all
  writes, reads, $\LLop$, $\SCop$, and $\VLop$ operations\\
  & \rep{\sucSet} &\subseteq \rep{\SC} \cup \rep{\VL} & Set of all
  successful $\SCop$ and $\VLop$ operations\\
  & \rep{W_c} &\wideEq \rep{W} \cup (\rep{\SC} \cap \rep{\sucSet}) &
  Set of all write-like events of the register\\
  & \rep{R_c} &\wideEq \rep{R} \cup \rep{\LL} \cup \rep{\SC} \cup
  \rep{\VL} & Set of all read-like events of the register\\
\end{letbox}

\begin{tabular*}{0.9\textwidth}{l@{\hskip 1em}A@{\extracolsep{\fill}}r}
  $\bigSum$
     & \rep{\obs} &\subseteq (\rep{W} \cup \rep{R_c})^2 & Visibility relation\\
     & \rep{\robs} &\subseteq \mathord{\rep{\obs}} \cap \rep{W_c} \times \rep{R_c} & Reads-from visibility\\
     & \rep{\llobs} &\subseteq \mathord{\rep{\obs}} \cap \rep{\LL} \times (\rep{\SC} \cup \rep{\VL}) & $\LLop$ visibility
\end{tabular*}\vspace{\sepAfterVar}

\begin{letbox}{A}
  & \rep{\hb} &\wideDefeq (\rep{\rb} \cup \rep{\obs})^+ & Happens-before order\\
\end{letbox}

\begin{tabular*}{0.9\textwidth}{l@{\hskip 0em}A@{\extracolsep{\fill}}r}
  $\forall \rep{e}, \rep{e'}.$
  &\rep{e \obs^+ e'}
  &\implies \rep{e' \not\rbeq e}
  &\eqref{ax:wfobs}\\
  \newsubeqblock
  $\forall \rep{r} \in \term(\rep{R_c}).$
  &\multicolumn{2}{c}{$\exists \rep{w} \in \rep{W_c}.\ \rep{w \robs r}$}
  &\insertSubeq\label[prop]{ax:mem2-robspop}\\
  $\forall \rep{w} \in \rep{W_c}, \rep{r} \in \term(\rep{R} \cup \rep{\LL}).$
  &\rep{w \robs r}
  &\implies \rep{w}.\evIn = \rep{r}.\evOut
  &\insertSubeq\label[prop]{ax:mem2-io}\\
  $\forall \rep{c} \in \term(\rep{\SC} \cup \rep{\VL}).$
  &\multicolumn{2}{c}{$\rep{c}.\evOut = \mathsf{booleanOf}(\rep{c} \in \rep{\sucSet})$}
  &\insertSubeq\label[prop]{ax:vl-io}\\
  $\forall \rep{w} \in \rep{W_c}, \rep{r} \in \rep{R_c}.$
  &\rep{w \robs r}
  &\implies \nexists \rep{w'} \in \rep{W_c}.\ \rep{w \hb w' \hb r}
  &\insertEq\label[prop]{ax:mem2-nowrbetween}\\
  $\forall \rep{w}, \rep{w'}, \rep{r}.$
  &\rep{w \robs r} \land \rep{w' \robs r} &\implies \rep{w = w'}
  &\insertEq\label[prop]{ax:mem2-robsuniq}\\
  $\forall \rep{w}, \rep{w'} \in \rep{W_c}.$
  &\rep{w \neq w'} &\implies \rep{w \hb w'} \lor \rep{w' \hb w}
  &\insertEq\label[prop]{ax:mem2-wrtotal}\\
  $\forall \rep{c} \in \rep{\SC} \cup \rep{\VL}.$
  &\multicolumn{2}{c}{$\exists \rep{l} \in \term(\rep{\LL}).\ \rep{l \llobs c}$}
  &\insertEq\label[prop]{ax:mem2-llobspop}\\
  $\forall \rep{l}, \rep{c}.$
  &\multicolumn{2}{c}{\hspace{-1cm}$\rep{l \llobs c}
  \implies \rep{l}.\evParent = \rep{c}.\evParent \land \nexists \rep{l'} \in \rep{\LL}.\ \rep{l \hb l' \hb c}$}
  &\insertEq\label[prop]{ax:mem2-llobsparent}\\
  $\forall \rep{w}, \rep{w'}, \rep{l}, \rep{c}.$
  &\hspace{-0.9cm}\rep{l \llobs c} \land \rep{w \robs l} \land \rep{w' \robs c}
  &\implies (\rep{w = w'}  \iff \rep{c} \in \rep{\sucSet})
  &\insertEq\label[prop]{ax:llsc-success}\\[\sepLastProp]
  \bottomrule
  \end{tabular*}\vspace{-2mm}

  \caption{Signature for atomic
    register with $\LLins$/$\SCins$/$\VLins$ operations.\vspace{-2mm}}
  \label{fig:llsc-memory}
\end{figure}

To capture the properties of memory with the $\LLins$, $\SCins$ and
$\VLins$ operations, we introduce a new signature in
\cref{fig:llsc-memory}, which we call $\LLins$/$\SCins$ registers, or
$\llregSig$ for short. It extends $\aregSig$ from \cref{fig:memory}.
We will treat each execution of $\LLins$, $\SCins$ and $\VLins$ as a
read-like event, i.e., they will observe some write by
$\rep{\robs}$. Additionally, we denote the $\LLop$ visibility
$\rep{\llobs}$ to link $\LLins$ events with their corresponding
$\SCins$/$\VLins$ events. To capture the concept of $\SCins$/$\VLins$
operations being successful, we introduce the set of successful events
$\rep{\sucSet}$. If an $\SCins$ event is in $\rep{\sucSet}$, it was
successful and must have written to memory, therefore we treat it as a
write-like event. These sets are formally defined in the top part of
\cref{fig:llsc-memory}.

Further, referring to \cref{fig:llsc-memory}, the
Properties~\eqref{ax:mem2-robspop} to \eqref{ax:mem2-wrtotal} are
almost identical to \eqref{ax:mem-io} to \eqref{ax:mem-wrtotal} from
$\aregSig$ from \cref{fig:memory}. The main difference is that we
substitute quantification over writes and reads with quantification
over write-like events ($\rep{W_c} = \rep{W} \cup (\rep{\SC} \cap
\rep{\sucSet})$) and read-like events ($\rep{R_c} = \rep{R} \cup
\rep{\LL} \cup \rep{\SC} \cup \rep{\VL}$), respectively. We also split
\cref{ax:mem-io} into three: \eqref{ax:mem2-robspop} says that each
terminated read-like event has some observed write-like event;
\eqref{ax:mem2-io} says that the output of reads and $\LLins$ is the
same as the input of the observed write-like event, and
\eqref{ax:vl-io} says that the output of $\SCins$ and $\VLins$ depends
on its success.
Additionally, \cref{ax:mem2-llobspop} states that each $\SCins$ and
$\VLins$ event has a terminated $\LLins$ event it is related to in
$\rep{\llobs}$. \cref{ax:mem2-llobsparent} says that each
$\rep{\llobs}$ pair has the same abs parent, capturing that the two
events of the pair were executed in the same thread (as per our
assumption in Section~\ref{subsub:eventstruct} that each event is
single-threaded), and that no other $\LLins$ event occurs in between
the pair.  
Finally,
\cref{ax:llsc-success} captures that successful events observe the
same write-like event as their $\rep{\llobs}$-related $\LLins$ event,
meaning no other modifications occurred between the $\LLins$ event and
the successful event.

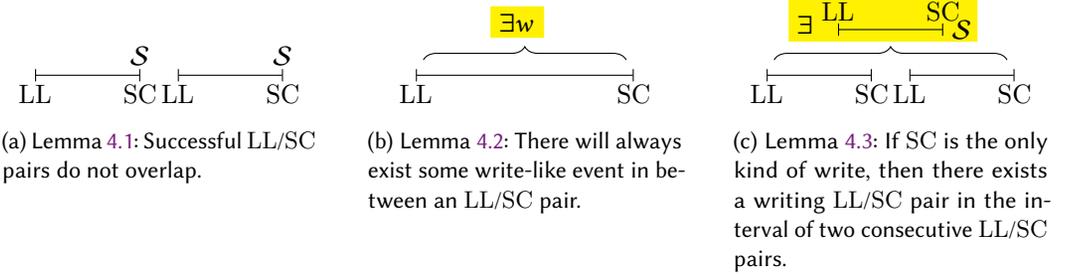
\begin{figure}
  \centering
  \begin{subfigure}[t]{0.3\textwidth}
    \centering
    \begin{tikzpicture}[baseline]
      \draw[|-|] (0,0) node[below] {$\LLins$} -- (1.4,0) node[below] {$\SCins$} node[above] {$\sucSet$};
      \draw[|-|] (1.9,0) node[below] {$\LLins$} -- (3.3,0) node[below] {$\SCins$} node[above] {$\sucSet$};
    \end{tikzpicture}
    \caption{\cref{lem:llsc-seq}: Successful $\LLins$/$\SCins$ pairs
      do not overlap.}
    \label{fig:llsc-propA}
  \end{subfigure}
  \hfill
  \begin{subfigure}[t]{0.3\textwidth}
    \centering
    \begin{tikzpicture}[baseline]
      \draw[fill=yellow,yellow] (1.0,0.5) rectangle ++(0.7,0.4);
      \draw[|-|] (0,0) node[below] {$\LLins$} -- (2.9,0) node[below] {$\SCins$};
      \draw[decorate, decoration={brace, amplitude=5pt}] ([yshift=0.2cm]0.1,0)-- node[above=0.25cm,xshift=-3pt]
      {$\exists w$}([yshift=0.2cm]2.8,0);
    \end{tikzpicture}
    \caption{\cref{lem:llsc-interference}: There will always exist
      some write-like event in between an $\LLins$/$\SCins$
      pair.}
    \label{fig:llsc-propB}
  \end{subfigure}
  \hfill
  \begin{subfigure}[t]{0.3\textwidth}
    \centering
    \begin{tikzpicture}[baseline]
      \draw[fill=yellow,yellow] (0.3,0.45) rectangle ++(2.5,0.6);
      \draw[|-|] (0,0) node[below] {$\LLins$} -- (1.4,0) node[below] {$\SCins$};
      \draw[|-|] (1.9,0) node[below] {$\LLins$} -- (3.3,0) node[below] {$\SCins$};
      \draw[decorate, decoration={brace, amplitude=5pt}] ([yshift=0.2cm]0.1,0)-- ([yshift=0.2cm]3.2,0);
      \draw[|-|] (0.95,0.6) node[left,xshift=-6pt,yshift=2pt] {$\exists$} node[above] {$\LLins$} -- (2.35,0.6) node[above] {$\SCins$} node[right] {$\sucSet$};
    \end{tikzpicture}
    \caption{\cref{lem:llsc-double}: If $\SCins$ is the only kind of
      write, then there exists a writing $\LLins$/$\SCins$ pair in the
      interval of two consecutive $\LLins$/$\SCins$ pairs.}
    \label{fig:llsc-propC}
  \end{subfigure}\vspace{-2mm}

  \caption{Illustration of key properties of $\LLins$/$\SCins$.\vspace{-2mm}}
  \label{fig:llsc-prop}
\end{figure}

There are several important properties of $\LLins/\SCins/\VLins$ that
the algorithms relies upon for synchronization. We codify them in the
following three lemmas and illustrate in \cref{fig:llsc-prop}. For
each of the lemmas, the events $l$, $c$, $w$, and variants, operate
over the same memory cell $\resX : \llregSig$.

\begin{lemma}\label{lem:llsc-seq}
  Let $\rep{l \llobs c}$ and $\rep{l' \llobs c'}$ with $\rep{c},
  \rep{c'} \in \term(\rep{\SC} \cap \rep{\sucSet})$. If $\rep{c \hb
    c'}$ then $\rep{c \hb l'}$.
\end{lemma}

\begin{proof}
  By \cref{ax:mem2-llobspop}, we know that $\rep{l}$ and $\rep{l'}$
  are terminated, and by \cref{ax:mem2-robspop}, we have that each of
  $\rep{l}$, $\rep{l'}$, $\rep{c}$ and $\rep{c'}$ observes some
  write-like event. By \cref{ax:llsc-success} with
  $\rep{c}, \rep{c'} \in \rep{\sucSet}$, then $\rep{l}$ and $\rep{c}$
  observe the same write-like event, and dually for $\rep{l'}$ and
  $\rep{c'}$, i.e., there exists $\rep{w}$ and $\rep{w'}$ such that
  $\rep{w \robs l}$ and $\rep{w \robs c}$ and $\rep{w' \robs l'}$ and
  $\rep{w' \robs c'}$. By \cref{ax:mem-wrtotal}, we either have
  $\rep{w' \hb c}$ or $\rep{c \hbeq w'}$. In the former case, we have
  $\rep{w' \hb c \hb c'}$, which contradicts \cref{ax:mem-nowrbetween}
  by there being a $\rep{c}$ in between $\rep{w' \robs c'}$. Thus, we
  must have the latter case, giving us $\rep{c \hbeq w'}$, letting us
  construct $\rep{c \hbeq w' \robs l'}$ which implies our goal
  $\rep{c \hb l'}$.
\end{proof}

\begin{lemma}\label{lem:llsc-interference}
  Let $\rep{l \llobs c}$ with either $\rep{c} \in \term(\rep{\SC})$ or
  $\rep{c} \in \term(\rep{\VL} \setminus \rep{\sucSet})$. Then there
  exists some $\rep{w} \in \rep{W_c}$ such that $\rep{w \not\hb l}$
  and $\rep{w \hbeq c}$.
\end{lemma}

\begin{proof}
  If $\rep{c} \in \term(\rep{\SC} \cap \rep{\sucSet})$, then $\rep{c}$
  must have successfully written; thus $\rep{c} = \rep{w}$ trivially
  satisfies $\rep{w \hbeq c}$ and $\rep{w \not\hb l}$. The latter
  holds, because otherwise we get $\rep{c \hb l}$ and
  $\rep{l \hbeq c}$ and thus $\rep{c \hb c}$ which contradicts
  irreflexivity of $\rep{\hb}$. Next, consider
  $\rep{c} \notin \rep{\sucSet}$, where
  $\rep{c} \in \term(\rep{\SC}\cup \rep{\VL})$.

  By \cref{ax:mem2-llobspop}, $\rep{l}$ is terminated, and by
  \cref{ax:mem2-robspop}, for $\rep{l}$ and $\rep{c}$ there exist
  write-like events $\rep{w}, \rep{w'} \in \rep{W_c}$ observed by
  $\rep{l}$ and $\rep{c}$, respectively. By \cref{ax:llsc-success},
  these write-like events must be distinct, because otherwise
  $\rep{c}$ would have been successful. In other words,
  $\rep{w \robs l}$ and $\rep{w' \robs c}$ where
  $\rep{w} \neq \rep{w'}$. Taking $\rep{w'}$ to be the required write
  of the lemma, we trivially have $\rep{w' \hbeq c}$. To show that
  also $\rep{w' \not\hb l}$, we assume $\rep{w' \hb l}$ and derive a
  contradiction. By \cref{ax:mem-wrtotal}, we either have
  $\rep{w \hb w'}$ or $\rep{w' \hb w}$. In the first case, we
  contradict \cref{ax:mem2-nowrbetween} for $\rep{w \robs l}$ by
  $\rep{w \hb w' \hb l}$. In the second case, we contradict
  \cref{ax:mem2-nowrbetween} for $\rep{w' \robs c}$ by
  $\rep{w' \hb w \robs l \llobs c}$.
\end{proof}

\begin{lemma}\label{lem:llsc-double}
  Let $\rep{l \llobs c \rb l' \llobs c'}$ where $\rep{c}, \rep{c'} \in
  \term(\rep{\SC})$. If within the time frame of the events $\rep{l}$,
  $\rep{l'}$, $\rep{c}$, and $\rep{c'}$ there are no mutations to
  register $\resX$ except by $\SCins$, then there exist some
  $\rep{l''}$ and $\rep{c''}$ (over $\resX$) with $\rep{l'' \llobs
    c''}$ such that $\rep{l'' \not\hb l}$ and $\rep{c'' \hbeq c'}$ and
  $\rep{c''} \in \rep{\SC} \cap \rep{\sucSet}$.
\end{lemma}

\begin{proof}
  From the assumption, by \cref{lem:llsc-interference}, there exists
  some $\rep{c_0}$ and $\rep{c'_0}$ such that $\rep{c_0 \not\hb l}$
  and $\rep{c_0 \hbeq c}$ and $\rep{c'_0 \not\hb l'}$ and $\rep{c'_0
    \hbeq c'}$. Since there are no other writes, we can only have
  $\rep{c_0}, \rep{c'_0} \in \rep{\SC} \cap \rep{\sucSet}$, by
  \cref{ax:mem2-llobspop} there has to be some $\rep{l'_0}$ such that
  $\rep{l'_0 \llobs c'_0}$. Let $\rep{l'_0}$ and $\rep{c'_0}$ be the
  existentials, we trivially have $\rep{c'_0 \hbeq c'}$ and
  $\rep{c'_0} \in \rep{\SC} \cap \rep{\sucSet}$. For $\rep{l'_0
    \not\hb l}$, assume we have $\rep{l'_0 \hb l}$, we derive a
  contradiction. We must have $\rep{c_0 \hb c'_0}$, since we otherwise
  we contradict $\rep{c'_0 \not\hb l'}$ by $\rep{c'_0 \hb c_0 \hbeq c
    \rb l'}$, meaning we also must have $\rep{c_0 \hb l'_0}$ for
  $\rep{c'_0}$ to be successful, however this contradicts $\rep{c_0
    \not\hb l}$ by $\rep{c_0 \hb l'_0 \hb l}$.
\end{proof}

\cref{lem:llsc-seq} says that if we have two successful
$\LLins$/$\SCins$ pairs, then the intervals of the two pairs do not
overlap. To see the intuition behind the other two lemmas, consider
the special case when the order $\rep{\hb}$ is total (e.g., it is a
linearization order), so that $\rep{\not\hb}$ is equivalent to
$\rep{\haeq}$.
In that case, \cref{lem:llsc-interference} says that an
$\LLins$/$\SCins$ interval, or an $\LLins$/$\VLins$ interval with a
failing $\VLins$, must contain some succesful write-like event.
Similarly, \cref{lem:llsc-double} says there exists a successful
$\LLins/\SCins$ interval contained within the interval defined by two
sequential $\LLins/\SCins$ intervals, even if the latter two intervals
themselves correspond to failing $\SCins$'s.
The last property will be used by multi-writer algorithms as follows.
Referring to \cref{alg:jay2}, a write procedure will try to execute
forwarding by invoking two sequential intervals of $\LLins/\SCins$
pairs into the cell $\resB[i]$. Even if both $\SCins$'s fail,
\cref{lem:llsc-double} guarantees that another concurrent write will
have invoked a successful $\LLins/\SCins$ interval, thus forwarding
the same value on behalf of the original write.

\subsection{Multi-Writer Forwarding Signature}

\renewcommand{\AComment}[1]{\Comment{\makebox[8mm][l]{\abs{#1}}}}
\renewcommand{\RComment}[1]{\Comment{\makebox[8mm][l]{\rep{#1}}}}

\begin{algorithm}[t]
  \setlength\multicolsep{0pt}
  \begin{multicols}{2}
  \begin{algorithmic}[1]
    \Resource $\resA : \arrayType{n}{\valType}$
    \Resource $\resB : \arrayType{n}{\valType \cup \left\{\bot\right\}}$
    \Resource $\resX : \boolType \set \false$
    \Statex

    \Procedure{\writeProc}{$i : \natType, v : \valType$}{} \AComment{$w_i$}
    \miniskip
      \State $\resA[i] \set v$ \RComment{$\wra{w_i}$}
      \State $x \gets \LLins(\resX)$ \RComment{$\wrx{w_i}$}
      \If{$x$}
        \State \Call{\forwardProc}{$i,x$} \RComment{$\vrt{\wrfA{w_i}}$}
        \State \Call{\forwardProc}{$i,x$} \RComment{$\vrt{\wrfB{w_i}}$}
      \EndIf
    \EndProcedure
    \Statex

    \vspace{-2mm}
    \Procedure{\forwardProc}{$i : \natType$}{} \AComment{$\vrt{f}$}
    \miniskip
      \State $\LLins(\resB[i])$ \RComment{$\wrfbA{f}$}
      \State $v \gets \resA[i]$ \RComment{$\wrfa{f}$}
      \State $x' \gets \VLins(\resX)$ \RComment{$\wrfxB{f}$}
      \If{$x'$}
        $\SCins(\resB[i],v)$ \RComment{$\wrfbB{f}$}
      \EndIf
    \EndProcedure
    \columnbreak

    \Procedure{\scanProc}{}{$\arrayType{n}{\valType}$} \AComment{$s$}
    \miniskip
      \For{$i \in \left\{ 0 \dots n-1 \right\}$}
        \State $\resB[i] \set \bot$ \RComment{$\scr{s}{i}$}
      \EndFor
      \State $\resX \set \true$ \RComment{$\scon{s}$}
      \For{$i \in \left\{ 0 \dots n-1 \right\}$}
        \State $a \gets \resA[i]$ \RComment{$\sca{s}{i}$}
        \State $V[i] \set a$
      \EndFor
      \State $\resX \set \false$ \RComment{$\scoff{s}$}
      \For{$i \in \left\{ 0 \dots n-1 \right\}$}
        \State $b \gets \resB[i]$ \RComment{$\scb{s}{i}$}
        \If{$b \neq \bot$} $V[i] \set b$ \EndIf
      \EndFor
      \State \textbf{return} $V$
    \EndProcedure
  \end{algorithmic}
  \end{multicols}
  \caption{\label{alg:jay2} Jayanti's multi-writer, single-scanner snapshot algorithm.\vspace{-4mm}}
\end{algorithm}

\begin{figure}
  \centering \renewcommand{\arraystretch}{\sigSpacing}%
  \setlength\multicolsep{0pt}
  \begin{multicols}{2}
    \begin{tabular*}{0.4\textwidth}{l@{\extracolsep{\fill}}r}
      \toprule
      \textbf{resource}\ \ $\rep{\resA} \wideColon \arrayType{n}{\aregSig}$\\
      \textbf{resource}\ \ $\rep{\resX} \wideColon \llregSig$\\
      \bottomrule
    \end{tabular*}
    
    \medskip

    \begin{tabular*}{0.4\textwidth}{@{\hskip 2em}A@{\extracolsep{\fill}}r}
      \toprule
      \multicolumn{2}{l}{\rlap{\evsignature{$\vscanProc$+}}}\\
      \midrule
      \rep{\resB} &\wideColon $\rlap{$\arrayType{n}{\llregSig}$}$\\
      \rep{r_i} &\wideColon \rep{\resB[i].W}\\
      \rep{a_i} &\wideColon \rep{\resA[i].R}\\
      \rep{b_i} &\wideColon \rep{\resB[i].R}\\
      \rep{\mathit{on}} &\wideColon \rep{\resX.W_c}\\
      \rep{\underline{\mathit{on}}} &\wideColon \rep{\resX.W_c} \,\mathrm{or}\, \rep{\resX.R_c}\\
      \rep{\mathit{off}} &\wideColon \rep{\resX.W_c}\\
      \rep{\underline{\mathit{off}}} &\wideColon \rep{\resX.W_c} \,\mathrm{or}\, \rep{\resX.R_c}\\
      \bottomrule
    \end{tabular*}

    \begin{tabular*}{0.4\textwidth}{@{\hskip 2em}A@{\extracolsep{\fill}}r}
      \toprule
      \multicolumn{2}{l}{\evsignature{$\writeProc_i$+}}\\
      \midrule
      \rep{a} &\wideColon \rep{\resA[i].W}\\
      \rep{x} &\wideColon \rep{\resX.\LL}\\
      \vrt{f_1}, \vrt{f_2} &\in \forwardProc_i\\
      \bottomrule
    \end{tabular*} 

    \medskip

    \begin{tabular*}{0.4\textwidth}{@{\hskip 2em}A@{\extracolsep{\fill}}r}
      \toprule
      \multicolumn{2}{l}{\rlap{\evsignature{$\forwardProc_i$}}}\\
      \midrule
      \rep{\resBcell} &\wideColon \llregSig\\
      \rep{b} &\wideColon \rep{\resBcell.\LL}\\
      \rep{\underline{a}} &\wideColon \rep{\resA[i].R}\\
      \rep{\underline{x}} &\wideColon \rep{\resX.\VL}\\
      \rep{\underline{b}} &\wideColon \rep{\resBcell.\SC}\\
      \bottomrule
    \end{tabular*}
  \end{multicols}
  \caption{%
    Resources and event signatures for abs writes, virtual
    scans and virtual forwarding corresponding to $\mwFwdSig$. The rep events correspond to
    the rep events with the same suffix in \cref{alg:jay2}, with the
    exception of $\rep{\underline{\mathit{on}}}$ and
    $\rep{\underline{\mathit{off}}}$ which are new and correspond to a
    potential observer of $\rep{\mathit{on}}$ and $\rep{\mathit{off}}$
    respectively.\vspace{-3mm}
  }
\label{fig:struct2}
\end{figure}

\begin{figure}
  \centering
  \renewcommand{\arraystretch}{\sigSpacing}%

\axiomset{F\textsuperscript{+}\hspace{-2pt}}
\signatureHeader{$\mwFwdSig$}{Multi-Writer snapshot with forwarding}

\begin{letbox}{l}
  & $\abs{W_i}$ & Set of all writes of cell $i$ of the snapshot
  in history with each $\abs{w_i} \in \writeProc_i$+\\
  & $\abs{S}$ & Set of all scans of the snapshot in history\\
  & $\vrt{F_i}$ & Set of all forwarding events in history with each $\vrt{f} \in \forwardProc_i$\\
\end{letbox}

\begin{tabular*}{0.9\textwidth}{l@{\hskip 1em}A@{\extracolsep{\fill}}r}
  $\bigSum$
  & \vrt{\Vsc} && Set of virtual scans with each $\vrt{\vsc} \in \vscanProc$+\\
  & \SMap{[-]} &\wideColon \abs{S} \parfun \vrt{\Vsc} & \ScVscDesc{}\\
  & \abs{[\vrt{-}]_w} &\wideColon \vrt{F_i} \to \abs{W_i} & Write which executed the forwarding event\\
  & \vrt{[-]_\vsc} &\wideColon \vrt{F_i} \to \vrt{\Vsc} & Virtual scan observed by forwarding event
\end{tabular*}\vspace{\sepAfterVar}

\begin{letbox}{A}
  &\abs{w_i \fobs \vrt{\vsc}}
  &\wideDefeq \exists \vrt{f} \in \vrt{F_i} .\ \rep{\wra{w_i} \robs \wrfa{f}}
    \land \rep{\wrfbB{f} \robs \scb{\vsc}{i}}
  & Forwarding visibility\\
  &\fwdRobsDef{&} & Reading visibility\\
  &\fwdWobsDef{&} & Writing visibility\\
  &\abs{\obs} &\wideDefeq \abs{\robs} \cup \abs{\wobs} & Visibility relation\\
  &\abs{\hb}  &\wideDefeq (\abs{\rb} \cup \abs{\obs})^+ & Happens-before order\\
\end{letbox}

\begin{tabular*}{0.9\textwidth}{l@{\hskip -2em}A@{\extracolsep{\fill}}r}
  \fwdAxVrtInScan{&\multicolumn{2}{c}}
  &\stepcounter{equation}\eqref{ax:fwd-vrtinscan}\\
  \fwdAxIO{&\multicolumn{2}{c}}
  &\stepcounter{equation}\eqref{ax:fwd-io}\\
  \newsubeqblock
  \fwdAxWraUniq{&\multicolumn{2}{c}}
  &\stepcounter{subeq}\eqref{ax:fwd-wrauniq}\\
  \fwdAxScrUniq{W_c}{&}{&}
  &\stepcounter{subeq}\eqref{ax:fwd-scruniq}\\
  $\forall \vrt{\vsc}, \rep{e_r} \in \rep{\vrt{\vsc}.\resB[i].W_c}.$
  &\rep{e_r}.\evIn \neq \bot
  &\iff \exists \vrt{f}.\ \rep{\wrfbB{f}} = \rep{e_r}
  &\insertSubeq\label[prop]{ax:fwd2-fbBuniq}\\
  $\forall \vrt{\vsc}, \rep{e_r} \in \rep{\resX.R_c}.$
  &\rep{\scon{\vsc} \robs e_r}
  &\iff \rep{\scon{\vsc} \hb e_r} \land \rep{\scoff{\vsc} \not\hb e_r}
  &\insertSubeq\label[prop]{ax:fwd2-sconuniq}\\[\sepInnerProp]
  \newsubeqblock
  \fwdAxScTotal{&}{&\ \,}
  &\stepcounter{subeq}\eqref{ax:fwd-sctotal}\\
  $\forall \vrt{\vsc} \in \vrt{\Vsc}.$
  &\multicolumn{2}{c}{\hspace{-1em}$\rep{\scr{\vsc}{i} \rb \scon{\vsc} \robseq \sconobs{\vsc} \rb \sca{\vsc}{i} \rb \scoff{\vsc} \robseq \scoffobs{\vsc} \rb \scb{\vsc}{i}}$}
  &\insertSubeq\label[prop]{ax:fwd2-scstruct}\\
  $\forall \abs{w_i} \in \abs{W_i}.$
  &\multicolumn{2}{c}{\hspace{-4em}$\rep{\wra{w_i} \rb \wrx{w_i} \rb \vrt{\wrfA{w_i} \rb \wrfB{w_i}}}$}
  &\insertSubeq\label[prop]{ax:fwd2-wrstruct}\\
  $\forall \vrt{f} \in \vrt{F_i}.$
  &\multicolumn{2}{c}{\hspace{-4em}$\rep{\wrfbA{f} \rb \wrfa{f} \rb \wrfxB{f} \rb \wrfbB{f}} \land \rep{\wrfbA{f} \llobs \wrfbB{f}}$}
  &\insertSubeq\label[prop]{ax:fwd2-fwdstruct}\\[\sepInnerProp]
  \newsubeqblock
  $\forall \vrt{f}, \abs{w_i}, \vrt{\vsc}.$
  &\abs{w_i} = \abs{[ \vrt{f} ]_w} \land \vrt{\vsc} = \vrt{[ f ]_\vsc}
  &\ \,\implies \rep{\scon{\vsc} \robs \wrx{w_i}} \rb \vrt{f} \land \vrt{f}.\resBcell = \vrt{\vsc}.\resB[i]
  &\insertSubeq\label[prop]{ax:fwd2-fwdprecond}\\
  $\forall \vrt{f} \in \vrt{F_i}.$
  &\mathsf{def}(\rep{\wrfbB{f}})
  &\ \,\implies \exists \vrt{\vsc}.\ \vrt{\vsc} = \vrt{[ f ]_\vsc} \land \rep{\scon{\vsc} \robs \wrfxB{f}}
  &\insertSubeq\label[prop]{ax:fwd2-fwdsccond}\\[\sepLastProp]
  \bottomrule
\end{tabular*}\vspace{-2mm}

\caption{Signatures for multi-writer
  snapshot data-structure with $\LLins$/$\SCins$ forwarding
  (\cref{alg:jay2,alg:jay3}).\vspace{-3mm}}
  \label{fig:forward2}
\end{figure}

\subsubsection{Description of Algorithm \ref{alg:jay2}}
We first discuss \cref{alg:jay2}, which will motivate the definition
of $\mwFwdSig$ signature. 
The common aspect with \cref{alg:jay1} is that the writer communicates
its value to the data structure by writing into $\resA[i]$. Whereas
\cref{alg:jay1} may try to communicate the same value \emph{again} by
forwarding to $\resB[i]$, in \cref{alg:jay2}, the forwarding procedure
reads whichever value is currently present in $\resA[i]$ and attempts
to forward it to $\resB[i]$ by means of $\LLins$ and $\SCins$. Because
multiple writes may be racing on $\resA[i]$, the forwarding procedure
may read and forward a different value from $\resA[i]$ than the one
the writer initially wrote. The forwarding procedure may even fail to
forward anything. Nevertheless, \cref{lem:llsc-double} provides a
guarantee that there will exist \emph{some} successful forwarding
among the concurrent processes, forwarding a value that is current
(i.e., written by another overlapping write). This is the key property
facilitating linearizability.

Other than that, the $\scanProc$ procedure is mostly the same in
\cref{alg:jay2} compared to \cref{alg:jay1}; the main difference is
that the order between $\rep{\scon{s}}$ and each $\rep{\scr{s}{i}}$ is
swapped, which is done in preparation for \cref{alg:jay3}.
In \cref{alg:jay1} this swap would result in a bug, causing
interference between forwarding and initialization of $\resB$; the
writes may fail to forward their value because the forwarding flag is
turned on too late, or scans may erase forwarded values by the
initializations. But this is safe to do in \cref{alg:jay2} because the
use of $\LLins$, $\SCins$ and $\VLins$ over $\resX$ and $\resB$
ensures that no forwarding (intended for a scan $\abs{s'}$ prior to
$\abs{s}$) can succeed in-between $\rep{\scr{s}{i}}$ and
$\rep{\scon{s}}$ of $\abs{s}$. This is substantiated by the following
lemma.

\begin{lemma}\label{lem:jay2nofwd-at-reset}
  Let $\abs{s'}$ be a scan prior to $\abs{s}$, i.e., $\abs{s' \rb
    s}$. If $\rep{\wrfxB{f}}$ succeeds observing the same
  $\rep{\scon{s'}}$ that the $\rep{\wrx{w_i}}$ preceding
  $\rep{\wrfxB{f}}$ observed, i.e., $\rep{\scon{s'} \robs \wrfxB{f}}$
  and $\rep{\scon{s'} \robs \wrx{w_i}}$, and $\rep{\scr{s}{i} \hb
    \wrfbB{f}}$, then $\rep{\wrfbB{f}}$ must fail to write.
\end{lemma}

\begin{proof}
  We assume that $\rep{\wrfbB{f}}$ is successful, and derive
  contradiction. By \cref{ax:mem2-io}, there exists some write-like
  event $\rep{w}$ such that $\rep{w \robs \wrfbA{f}}$, and by
  $\rep{\wrfbB{f}}$ being successful, \cref{ax:llsc-success} implies
  that $\rep{\wrfbB{f}}$ observes the same write-like event as
  $\rep{\wrfbA{f}}$, i.e.\ $\rep{w \robs \wrfbB{f}}$. By
  \cref{ax:mem2-wrtotal}, we either have $\rep{\scr{s}{i} \hb w}$ or
  $\rep{w \hbeq \scr{s}{i}}$, in the latter case we have
  $\rep{w \hb \scr{s}{i} \hb \wrfbB{f}}$, contradicting
  \cref{ax:mem2-nowrbetween} by $\rep{\scr{s}{i}}$ being in between
  $\rep{w \robs \wrfbB{f}}$, thus we can only have
  $\rep{\scr{s}{i} \hbeq w}$. Since \cref{alg:jay2} is single-scanner,
  it follows that $\abs{s}$ only starts after $\abs{s'}$ has finished,
  thus we can derive the following:
  \[
    \rep{\scon{s'} \rb \scoff{s'} \rb \scr{s}{i} \hbeq w \hb \wrfbA{f} \hb \wrfxB{f}}
  \]
  which implies that we have $\rep{\scoff{s'}}$ in between
  $\rep{\scon{s'} \robs \wrfxB{f}}$, contradicting \cref{ax:mem2-nowrbetween}.
\end{proof}

\subsubsection{The $\mwFwdSig$ Signature}

The rep-event signatures for multi-writer algorithms in
\cref{fig:struct2} correspond closely to the rep events of
\cref{alg:jay2}.
\cref{fig:struct2} extends \cref{fig:struct}, adding the
$\LLins/\SCins$ register $\resX$ and extending the event signatures.
$\writeProc_i$+ extends $\writeProc_i$ by events $\rep{\wrx{w_i}}$,
$\rep{\wrfA{w_i}}$, and $\rep{\wrfB{w_i}}$ from \cref{alg:jay2}, along
with an event signature $\forwardProc_i$ corresponding to the
procedure $\forwardProc$ from \cref{alg:jay2}.
$\vscanProc$+ extends $\vscanProc$ by updating $\resB$ to be an array
of $\LLins/\SCins$ registers, and adding four additional rep events:
$\rep{\scon{\vsc}}$ and $\rep{\scoff{\vsc}}$ corresponding to
$\rep{\scon{s}}$ and $\rep{\scoff{s}}$ from \cref{alg:jay2}, and
$\rep{\sconobs{\vsc}}$ and $\rep{\scoffobs{\vsc}}$ which are providing
support for multi-scanner \cref{alg:jay3}.
As we shall see in Section~\ref{sec:jay3}, in \cref{alg:jay3},
$\rep{\scon{\vsc}}$ might not return before $\rep{\sca{\vsc}{i}}$, and
similarly for $\rep{\scoff{\vsc}}$ and $\rep{\scb{\vsc}{i}}$. Thus, we
need events that observe $\rep{\scon{\vsc}}$ and $\rep{\scoff{\vsc}}$
returning before $\rep{\sca{\vsc}{i}}$ and $\rep{\scb{\vsc}{i}}$
respectively. For \cref{alg:jay2}, we simply take
$\rep{\sconobs{\vsc}} = \rep{\scon{\vsc}}$ and $\rep{\scoffobs{\vsc}}
= \rep{\scoff{\vsc}}$.  

Signature $\mwFwdSig$ in \cref{fig:forward2} captures the common
structure of \cref{alg:jay2,alg:jay3}. Some aspects are preserved from
\cref{alg:jay1}; for instance
\cref{ax:fwd-vrtinscan,ax:fwd-io,ax:fwd-wrauniq,ax:fwd-scruniq,ax:fwd-sctotal}
are the same as in $\fwdSig$ signature of \cref{fig:forward}.
One difference is that we now have sets of forwarding events
$\vrt{F_i}$ that are instances of the forwarding procedure. The new
signature introduces operations over such forwardings:
$\abs{[\vrt{-}]_w}$ maps a forwarding to the write that invoked it,
and $\vrt{[\vrt{-}]_{\vsc}}$ maps a forwarding to the virtual scan it
is forwarding to. Forwarding visibility $\abs{\fobs}$ is also updated:
it now says that a forwarding reads a value of a write from
$\rep{\wra{w_i}}$ and relays this value by $\rep{\wrfbB{f_i}}$.

Among the new properties, \eqref{ax:fwd2-fbBuniq} encodes that only
$\rep{\wrfbB{f}}$ performs forwarding writes into
$\vrt{\vsc}.\resB[i]$. \cref{ax:fwd2-sconuniq} ensures that nothing
wrote to $\resX$ other than $\rep{\scoff{\vsc}}$ directly after
$\rep{\scon{\vsc}}$.
\cref{ax:fwd2-scstruct,ax:fwd2-wrstruct,ax:fwd2-fwdstruct} corresponds
to the rep event order of scans, writes and forwarding respectively.
Lastly, \cref{ax:fwd2-fwdprecond} ensures that every forwarding has a
virtual scan to forward to, while \cref{ax:fwd2-fwdsccond} states that
$\rep{\wrfbB{f}}$ can only be performed if $\rep{\wrfxB{f}}$ observed
$\rep{\scon{\vsc}}$.

\begin{lemma}\label{lem:fwd2-fwd}
  Histories satisfying $\mwFwdSig$ signature also satisfy $\fwdSig$
  signature.
\end{lemma}
\begin{proof}[Proof sketch]
  We instantiate $\vrt{\Vsc}$ and $\SMap{[-]}$ of $\fwdSig$ to be the
  same as the corresponding instantiations of $\mwFwdSig$.
  \cref{ax:fwd-vrtinscan,ax:fwd-io,ax:fwd-wrauniq,ax:fwd-scruniq,ax:fwd-sctotal}
  are shared between the signatures, therefore they trivially hold.
  
  We next focus on proving \cref{ax:forward2b},
  \ifappendix%
  leaving the remaining properties for Appendix~\apndxMWFwd{}.
  \else%
  we prove the remaining properties in Appendix~\apndxMWFwd{}
  of~\cite{extended}.
  \fi%
  Proving \cref{ax:forward2b} requires showing that
  $\abs{w_i \fobs \vrt{\vsc}}$ and $\abs{w'_i \rbhbeqrobs \vrt{\vsc}}$
  and $\abs{w_i \hb w'_i}$ derive a contradiction. To show this, we
  need a helper lemma, proved in Appendix~\apndxMWFwd{}\ifNotAppendix{
    of~\emph{loc.~cit.}} (Lemma~\apndxLemJayIILin{}):
  If $\abs{w_i \rbhbeqrobs \vrt{\vsc}}$ then surely $\abs{w_i}$ will
  be able to forward to $\vrt{\vsc}$ in time, i.e., we have
  $\rep{\scoff{\vsc} \not\hb \wrx{w_i}}$ and
  $\rep{\scoff{\vsc} \not\hb \wrfxB{f}}$ and
  $\rep{\wrfbB{f} \hb \scb{\vsc}{i}}$ for every $\vrt{f}$ executed by
  $\abs{w_i}$.
  From $\abs{w_i \fobs \vrt{\vsc}}$, by the definition of
  $\abs{\fobs}$, we have $\rep{\wra{w_i} \robs \wrfa{f}}$ and
  $\rep{\wrfbB{f} \robs \scb{\vsc}{i}}$ for some $\vrt{f}$. By
  \cref{lem:fwd-hbabsrep}, $\abs{w_i \hb w'_i}$ implies
  $\rep{\wra{w_i} \hb \wra{w'_i}}$.
  We can also infer (full proof in
  Appendix~\apndxMWFwd{}\ifNotAppendix{ of~\emph{loc.~cit.}}) that
  $\rep{\scon{\vsc} \hb \wrx{w'_i}}$.

  By Lemma~\apndxLemJayIILin{}, we have
  $\rep{\scoff{\vsc} \not\hb \wrx{w_i}}$, thus by
  \cref{ax:fwd2-sconuniq} we derive
  $\rep{\scon{\vsc} \robs \wrx{w_i}}$, implying $\vrt{\wrfA{w'_i}}$
  and $\vrt{\wrfB{w'_i}}$ will be executed, let
  $\vrt{f'} = \vrt{\wrfA{w'_i}}$ and $\vrt{f''} = \vrt{\wrfB{w'_i}}$.
  By Lemma~\apndxLemJayIILin{}, we have
  $\rep{\scoff{\vsc} \not\hb \wrfxB{f'}}$ and
  $\rep{\scoff{\vsc} \not\hb \wrfxB{f''}}$ and
  $\rep{\wrfbB{f''} \hb \scb{\vsc}{i}}$, similarly to above with
  \cref{ax:fwd2-sconuniq} we derive
  $\rep{\scon{\vsc} \robs \wrfxB{f'}}$ and
  $\rep{\scon{\vsc} \robs \wrfxB{f''}}$, meaning $\rep{\wrfbB{f'}}$
  and $\rep{\wrfbB{f''}}$ will be executed, meaning we have two
  consecutive $\LLins$/$\SCins$ pairs of $\vrt{f'}$ and $\vrt{f''}$.
  Since we have $\rep{\scon{\vsc} \hb \wrx{w'_i}}$ and
  $\rep{\wrfbB{f''} \hb \scb{\vsc}{i}}$, it is not possible for any
  $\rep{\scr{\vsc'}{i}}$ to occur in the intervals of the
  $\LLins$/$\SCins$ pairs, therefore we can apply
  \cref{lem:llsc-double}, thus there exists some $\vrt{f'''}$ such
  that $\rep{\wrfbA{f'''} \not\hb \wrfbA{f'}}$ and
  $\rep{\wrfbB{f'''} \hbeq \wrfbB{f''}}$.
   
  By \cref{ax:mem2-wrtotal}, we either have
  $\rep{\wrfbB{f} \hb \wrfbB{f'''}}$ or
  $\rep{\wrfbB{f'''} \hbeq \wrfbB{f}}$. In the first case, we have
  \[
    \rep{\wrfbB{f} \hb \wrfbB{f'''} \hbeq \wrfbB{f''} \hb
      \scb{\vsc}{i}}
  \]
  which contradicts \cref{ax:mem2-nowrbetween} by $\rep{\wrfbB{f'''}}$
  occurring in between $\rep{\wrfbB{f} \robs \scb{\vsc}{i}}$. In the
  second case, if we have $\rep{\wrfbB{f'''} \hb \wrfbB{f}}$ then by
  \cref{lem:llsc-seq}, and $\rep{\wrfbB{f}}$ and $\rep{\wrfbB{f'''}}$
  being successful, we have $\rep{\wrfbB{f'''} \hb \wrfbA{f}}$, and
  this also follows in the case of $\rep{\wrfbB{f'''} = \wrfbB{f}}$.
  By \cref{ax:interval} over $\rep{\wra{w'_i} \rb \wrfbA{f'}}$ and
  $\rep{\wrfbA{f'''} \rb \wrfbB{f'''}}$ we have
  $\rep{\wra{w'_i} \rb \wrfbB{f'''}}$ since we cannot have
  $\rep{\wrfbA{f'''} \rb \wrfbA{f'}}$ by
  $\rep{\wrfbA{f'''} \not\hb \wrfbA{f'}}$. Thus, we have
  \[
    \rep{\wra{w_i} \hb \wra{w'_i} \rb \wrfbB{f'''} \hb \wrfbA{f} \rb \wrfa{f}}
  \]
  which contradicts \cref{ax:mem-nowrbetween} by $\rep{\wra{w'_i}}$
  occurring in between $\rep{\wra{w_i} \robs \wrfa{f}}$.
\end{proof}

\begin{lemma}\label{lem:jay2-fwd2}
  Every execution of \cref{alg:jay2} satisfies the $\mwFwdSig$ signature
  (\cref{fig:forward2}).
\end{lemma}
\begin{proof}
  Since this algorithm is single scanner, we can simply define the set
  of virtual scans to be the same as the set of abs scans and map each
  abs scan to itself, i.e., $\vrt{\Vsc} = \abs{S}$ and
  $\SMap{\abs{s}} = \abs{s}$. Each rep event directly corresponds to
  their equivalent variant in \cref{alg:jay2}, except for
  $\rep{\sconobs{\vsc}}$ and $\rep{\scoffobs{\vsc}}$, which are set as
  $\rep{\sconobs{\vsc}} = \rep{\scon{\vsc}} = \rep{\scon{s}}$ and
  $\rep{\scoffobs{\vsc}} = \rep{\scoff{\vsc}} = \rep{\scoff{s}}$.
  \cref{ax:fwd-vrtinscan} holds by each event being a
  subevent of itself ($\abs{e \subev e}$) and \cref{ax:fwd-sctotal}
  holds since the algorithm is single-scanner.
  Properties~\eqref{ax:fwd-io} to \eqref{ax:fwd2-fbBuniq} and
  \eqref{ax:fwd2-scstruct} to \eqref{ax:fwd2-fwdsccond} holds directly
  by the structure of the algorithm and \cref{ax:fwd2-sconuniq} can be
  proven by Lemma~\apndxLemJayIObsX{}\ifNotAppendix{
    from~\cite{extended}}, the same lemma that established this
  property for \cref{alg:jay1}, since the relative structure of
  $\rep{\scon{s}}$ and $\rep{\scoff{s}}$ is the same.
\end{proof}

\subsection{Multi-Writer, Multi-Scanner}\label{sec:jay3}

\renewcommand{\AComment}[1]{\Comment{\makebox[8mm][l]{\abs{#1}}}}
\renewcommand{\RComment}[1]{\Comment{\makebox[8mm][l]{\rep{#1}}}}

\begin{algorithm}[t]
  \setlength\multicolsep{0pt}
  \begin{multicols}{2}
  \begin{algorithmic}[1]
    \Resource $\resA : \arrayType{n}{\valType}$
    \Resource $\resA_p : \arrayType{n}{\valType}$
    \Resource $\resB_p : \arrayType{n}{\valType \cup \left\{\bot\right\}}$
    \Resource $\resX :\scaleleftright[1.75ex]{<}
      {\,\begin{aligned}
        \phase &: \left\{1,2,3\right\} \set 1\\
        \procA\ \procB &: \PID\\
        \toggle &: \boolType \set \false
      \end{aligned}\,}{>}$
    \Resource $\resSS :\scaleleftright[1.75ex]{<}
    {\,\begin{aligned}
        \snapshot &: \arrayType{n}{\valType}\\
        \toggle &: \boolType \set \false
      \end{aligned}\,}{>}$
    \Statex

    \Procedure{\writeProc}{$i : \natType, v : \valType$}{} \AComment{$w_i$}
    \miniskip
      \State $\resA[i] \set v$ \RComment{$\wra{w_i}$}
      \State $x \gets \LLins(\resX)$ \RComment{$\wrx{w_i}$}
      \If{$\xPhase = 2$}
        \State \Call{\forwardProc}{$i,\xB$} \RComment{$\vrt{\wrfA{w_i}}$}
        \State \Call{\forwardProc}{$i,\xB$} \RComment{$\vrt{\wrfB{w_i}}$}
      \EndIf
    \EndProcedure
    \Statex

    \Procedure{\forwardProc}{$i : \natType, p : \PID$}{} \AComment{$\vrt{f}$}
    \miniskip
      \State $\LLins(\resB_p[i])$ \RComment{$\wrfbA{f}$}
      \State $v \gets \resA[i]$ \RComment{$\wrfa{f}$}
      \State $x' \gets \VLins(\resX)$ \RComment{$\wrfxB{f}$}
      \If{$x'$}
        $\SCins(\resB_p[i],v)$ \RComment{$\wrfbB{f}$}
      \EndIf
    \EndProcedure
    \Statex

    \Procedure{\scanProc}{}{$\arrayType{n}{\valType}$} \AComment{$s$}
    \miniskip
      \State $p \gets \getpid$
      \State \Call{\pushLSProc}{$p$} \RComment{$\vrt{\scvA}$}
      \State \Call{\pushLSProc}{$p$} \RComment{$\vrt{\scvB}$}
      \State $\mathit{ss} \gets \resSS$ \RComment{$\scx$}
      \State \textbf{return} $\ssSnapshot$
    \EndProcedure
    \columnbreak
    
    \Procedure{\pushLSProc}{$p : \PID$}{} \AComment{$\vrt{\scv}$}
    \miniskip
      \State $x \gets \LLins(\resX)$ \RComment{$\scvx$}
    
      \If{$\xPhase = 1$}
        \For{$i \in \left\{ 0 \dots n-1 \right\}$}
          \State $\resB_p[i] \set \bot$ \RComment{$\scvr_i$}
        \EndFor
        \State $\SCins( \resX,\left\langle 2,\xA, p,\xToggle\right\rangle )$ \RComment{$\scvon$}
        \State $x \gets \LLins(\resX)$ \RComment{$\scvx$}
      \EndIf

      \If{$\xPhase = 2$}
        \For{$i \in \left\{ 0 \dots n-1 \right\}$}
          \State $a \gets \resA[i]$ \RComment{$\scva_i$}
          \State $\resA_p[i] \set a$ \RComment{$\scvaA_i$}
        \EndFor
        \State $\SCins( \resX,\left\langle 3,p,\xB,\xToggle\right\rangle )$ \RComment{$\scvoff$}
        \State $x \gets \LLins(\resX)$ \RComment{$\scvx$}
    \EndIf

      \If{$\xPhase = 3$}
        \For{$i \in \left\{ 0 \dots n-1 \right\}$}
          \State $b \gets \resB_{\xB}[i]$ \RComment{$\scvb_i$}
          \If{$b \neq \bot$}
            \State $V[i] \set b$
          \Else
            \State $a \gets \resA_{\xA}[i]$ \RComment{$\scvaB_i$}
            \State $V[i] \set a$
          \EndIf
        \EndFor
        \State $\mathit{ss} \gets \LLins(\resSS)$ \RComment{$\scvssA$}
        \State $x' \gets \VLins(\resX)$ \RComment{$\scvxB$}
        \If{$\ssToggle = \xToggle \land x'$}
          \State $\SCins(\resSS, \left\langle V, \neg\ssToggle \right\rangle )$ \RComment{$\scvss$}
          \EndIf
        \State $\SCins( \resX,\left\langle 1,\xA,\xB, \neg\xToggle \right\rangle )$ \RComment{$\scvend$}
      \EndIf
    \EndProcedure
  \end{algorithmic}
  \end{multicols}
  \caption{\label{alg:jay3} Jayanti's multi-writer, multi-scanner snapshot algorithm.\vspace{-4mm}}
\end{algorithm}

\cref{alg:jay3} is the final of the Jayanti's snapshot algorithms,
allowing unconstrained write and scan operations. It differs from
\cref{alg:jay2} in the implementation of \scanProc{}, while
\writeProc{} is mostly unchanged.
While \cref{alg:jay3} allows for multiple concurrently running
physical scanners, under the hood, at most one \emph{virtual} scanner
can be logically viewed as running at any point in time. 
The physical scans maintain their own local snapshot and forwarding
arrays, but the shared resource $\resX$, via the fields $\resX.\procA$
and $\resX.\procB$, keeps track which local snapshot array and local
forwarding array, respectively, should be considered as the current
data of the virtual scan. The physical scanners race to modify these
fields in order to promote their own local arrays as the arrays of the
virtual scan. However, $\resX$ can be modified only by $\SCins$; thus,
by \cref{lem:llsc-seq}, mutations to $\resX$ occur sequentially, and
by \cref{lem:llsc-interference}, at least one mutation exists. The
sequential nature of these mutations justifies viewing them as
belonging to a single, albeit virtual, scan. Once the virtual scan
terminates, a snapshot is copied into $\resSS$ ($\rep{\scvss}$), which
scanners read ($\rep{\scx}$) and return. We follow Jayanti in assuming
that $\resSS$, even though physically a compound object, can still be
considered an $\LLins$/$\SCins$ register; this is not a loss of
generality, as multi-word $\LLins$/$\SCins$ registers
exist~\cite{jayanti:multi-llsc}.

The \pushLSProc{} procedure (short for ``push virtual scan'')
propels the virtual scanner to its completion, producing a snapshot in
$\resSS$.
The procedure is divided into three phases, tracked by the
field $\resX.\phase$, corresponding to the three phases of the
previous algorithms: the first phase sets each memory cell in $\resB$
to $\bot$ with $\rep{\scvr_i}$; the second phase reads the main memory
$\resA$ with $\rep{\scva_i}$; and the third phase reads the forwarding
memory $\resB$ with $\rep{\scvb_i}$. However, unlike before when
$\resB$ was a global array, here we have one $\resB_p$ for each scan
process $p$. A scanner that successfully writes with $\rep{\scvon}$ in
the first phase will have its $\resB_p$ used for forwarding in the
current virtual scan. For the second phase, the algorithm uses
$\resA_p$ to store the results of reading $\resA$, and similarly the
process $p$ which successfully writes with $\rep{\scvoff}$ will have
its $\resA_p$ used for the next phase. The third phase uses $\resA_p$
and $\resB_q$ from the processes stored in $\resX$ to construct the
snapshot, which is then stored in $\resSS$ with $\rep{\scvss}$.

Each phase starts with an $\LLins(\resX)$ (three lines labeled
$\rep{\scvx}$) and ends with an $\SCins(\resX)$. This ensures that if
multiple physical scanners run a phase simultaneously, only a single
$\SCins$ operation succeeds writing to $\resX$, committing the results
of the phase. Thus, considered across all physical scanners,
successful phases cannot overlap, by \cref{lem:llsc-seq}. Also, phases
must run consecutively, because a phase in some scanner can only start
if the previous phase, maybe executed by some other scanner, has been
terminated, incrementing $\resX.\phase$ to enable the next phase
($\rep{\scvon}$, $\rep{\scvoff}$, and $\rep{\scvend}$). 
Because the phases cannot overlap, must be executed in order, and
correspond to the scanner phases in the previous algorithms, a trace of
a virtual scan in \cref{alg:jay3} has essentially the same structure
as a trace of a physical scan in the \cref{alg:jay1,alg:jay2}.

To ensure that $\resSS$ is only mutated once per virtual scan, we have
two toggle bits $\resX.\toggle$ and $\resSS.\toggle$ as part of
$\resX$ and $\resSS$ respectively. These function such that every time
$\rep{\scvend}$ successfully writes, $\resX.\toggle$ is negated.
$\resSS$ can only be written to by $\rep{\scvss}$, which can only be
executed if $\resX.\toggle$ and $\resSS.\toggle$ are
equal\footnote{Jayanti's original presentation is dual, as the sync
  bits need to be distinct for writing to $\resSS$, but this is
  equivalent.} and $\resX$ did not change during the phase. Once
$\rep{\scvss}$ successfully writes, $\resSS.\toggle$ will be negated,
meaning no more writes into $\resSS$ can occur in this phase since
$\resX.\toggle$ and $\resSS.\toggle$ are now distinct. We also know
that there has to be exactly one successful write with $\rep{\scvss}$
before $\rep{\scvend}$, since $\rep{\scvss}$ only fails if some other
$\rep{\scvss}$ succeeded, or some other $\rep{\scvend}$ ended the
phase.

It is also vital for linearizability that the virtual scan is in the
time interval of the abs \scanProc{} returning it. To ensure the
virtual scan ends before the abs scanner ends, the latest virtual scan
is finished before \pushLSProc{} terminates. Dually, to ensure that
the virtual scan was not started before the abs scanner starts,
\pushLSProc{} is executed twice; once to finish the currently running
virtual scan, and once to generate a virtual scan that can be used by
the abs scanner. Executing \pushLSProc{} twice is similar to the idea of
executing \forwardProc{} twice in \cref{alg:jay2}.

\begin{lemma}\label{lem:jay3-fwd2}
  Every execution of \cref{alg:jay3} satisfies the $\mwFwdSig$ signature
  (\cref{fig:forward2}).
\end{lemma}
\begin{proof}[Proof sketch]
  We instantiate the write and forward rep events for the event
  structures of \cref{fig:struct2} with rep events of the same name.
  The structure of \writeProc{} and \forwardProc{} ensures that
  Properties~\eqref{ax:fwd-wrauniq}, \eqref{ax:fwd2-fbBuniq} and
  \eqref{ax:fwd2-wrstruct} to \eqref{ax:fwd2-fwdsccond} are satisfied.

  We instantiate virtual scans as a logical object corresponding to
  the rep events executed by \pushLSProc{} constructing a complete
  virtual scan. The interval of a virtual scan $\vrt{\vsc}$ is
  instantiated as the smallest interval containing all its rep events.
  More formally, if we have $\vrt{\scv}$, $\vrt{\scv'}$ and
  $\vrt{\scv''}$ such that $\rep{\scvon \robs \scvx'}$ and
  $\rep{\scvoff' \robs \scvx''}$ and $\rep{\scvss''}$ wrote to
  $\resSS$, then there exists a virtual scan $\vrt{\vsc}$. Its rep
  events are instantiated using the rep events of $\vrt{\scv}$,
  $\vrt{\scv'}$ and $\vrt{\scv''}$ as follows and instantiate
  $\SMap{\abs{s}} = \vrt{\vsc}$ iff $\rep{\scvss'' \robs \scx}$.
  \begin{align*}
    \rep{\scr{\vsc}{i}} &\defeq \rep{\scvr_i}
    &\rep{\scon{\vsc}} &\defeq \rep{\scvon}
    &\rep{\sconobs{\vsc}} &\defeq \rep{\scvx'}
    &\rep{\sca{\vsc}{i}} &\defeq \rep{\scva'_i}
    &\rep{\scoff{\vsc}} &\defeq \rep{\scvoff'}
    &\rep{\scoffobs{\vsc}} &\defeq \rep{\scvx''}
    &\rep{\scb{\vsc}{i}} &\defeq \rep{\scvb''_i}
  \end{align*}
  We also let $\rep{\scxinit{\vsc} \defeq \rep{\scvx}}$ and
  $\rep{\scss{\vsc} \defeq \rep{\scvss''}}$ and
  $\rep{\scend{\vsc} \defeq \rep{\scvend'''}}$ where
  $\rep{\scvend'''}$ was successful and
  $\rep{\scvoff' \robs \scvend'''}$, however these do not need to be
  in the interval of $\vrt{\vsc}$.
  \cref{ax:fwd-io,ax:fwd-scruniq,ax:fwd2-sconuniq,ax:fwd2-scstruct}
  follows from this structure, which we show in the
  Appendix~\apndxJayIII{}\ifNotAppendix{ of~\cite{extended}}.

  The most involved part is proving
  \cref{ax:fwd-vrtinscan,ax:fwd-sctotal}, corresponding to
  $\SMap{\abs{s}} \subev \abs{s}$ and
  ($\vrt{\vsc \rb \vsc'} \lor \vrt{\vsc' \rb \vsc}$) respectively,
  which require us to establish the following helper lemmas, presented
  in Appendix~\apndxJayIII{}\ifNotAppendix{ of~\cite{extended}}: By
  Lemma~\apndxLemJayIIISSEnd{}, we have exactly one
  $\rep{\scss{\vsc}}$ per $\vrt{\vsc}$, satisfying
  $\rep{\scss{\vsc} \hb \scend{\vsc}}$, and by
  Lemma~\apndxLemJayIIILSEnd{}, for any terminated $\vrt{\scv}$ and
  $\vrt{\scv'}$ of \pushLSProc{} with $\vrt{\scv \rb \scv'}$, there
  exists $\vrt{\vsc}$ and $\vrt{\vsc'}$ such that
  $\rep{\scend{\vsc} \hb \scend{\vsc'}}$ and
  $\rep{\scvx \not\hb \scend{\vsc}}$.
  With Lemma~\apndxLemJayIIISSEnd{}, we can that show for every $i$ we either have
  $\rep{\scb{\vsc}{i} \rb \scss{\vsc} \hb \scend{\vsc} \hb
    \scxinit{\vsc'} \rb \scr{\vsc'}{i}}$ or
  $\rep{\scb{\vsc'}{i} \rb \scss{\vsc'} \hb \scend{\vsc'} \hb
    \scxinit{\vsc} \rb \scr{\vsc}{i}}$, implying
  \cref{ax:fwd-sctotal}.
  With Lemma~\apndxLemJayIIILSEnd{}, we can show for $\abs{s}$ that
  there exists some $\vrt{\vsc'}$ and $\vrt{\vsc''}$ with
  $\rep{\scend{\vsc'} \hb \scend{\vsc''}}$ and
  $\rep{\scvx \not\hb \scend{\vsc'}}$ where
  $\rep{\scvx} \subev \vrt{\scvA}$. This implies that if we have
  $\vrt{\vsc} = \SMap{\abs{s}}$, then
  $\rep{\scend{\vsc'} \hb \scxinit{\vsc}}$, since some new virtual
  scan must have started after $\rep{\scend{\vsc'}}$, and by the
  instantiation of $\SMap{[-]}$ we have
  $\rep{\scss{\vsc} \robs \scx}$, thus for all $i$ we have
  $ \rep{\scend{\vsc'} \hb \scxinit{\vsc} \rb \scr{\vsc}{i} \rb
    \scb{\vsc}{i} \rb \scss{\vsc} \robs \scx} $, implying
  $\vrt{\vsc} \subev \abs{s}$ and \cref{ax:fwd-vrtinscan}, the last
  step is further explained in Appendix~\apndxJayIII{}\ifNotAppendix{ of~\cite{extended}}.
\end{proof}

\section{Applying to Other Snapshot Algorithms}\label{sec:afek}

\axiomset{A}

To show the generality of our approach, we next apply it to the
snapshot algorithm of \citet{afek:acm93}, presented as
\cref{alg:afek}. It is a single-writer, multi-scanner algorithm.
At its core, $\scanProc$ of \cref{alg:afek} collects snapshots by
reading the memory array $\resA$ twice. We denote the first read of
index $i$ by $\rep{\sca{s_k}{i}}$, and the second by
$\rep{\scb{s_k}{i}}$. The algorithm looks for changes in the array. If
a change is detected, the reading is restarted in the next
iteration. The index $k$ on the reading events identifies the
iteration in which the event occurs. If no change to the array is
detected between the first and second read events, then what was read
is a valid snapshot, reflecting what was in the array at the moment
the last index was first read in the current iteration.
To ensure that changes to the array are properly recognized, each
memory cell holds a version number ($\resA[i].\seqField$) of the
latest write, which is a number that is incremented each time a new
value is written to the cell.

\renewcommand{\AComment}[1]{\Comment{\makebox[8mm][l]{\abs{#1}}}}
\renewcommand{\RComment}[1]{\Comment{\makebox[8mm][l]{\rep{#1}}}}

\begin{algorithm}[t]
  \setlength\multicolsep{0pt}
  \begin{multicols}{2}
  \begin{algorithmic}[1]
    \Resource $\resA : \arrayType{n}{}$
    \Statex \hspace{6em}$\scaleleftright[1.75ex]{<}
      {\,\begin{aligned}
        \dataField &: \valType\\
        \seqField &: \natType\\
        \viewField &: \arrayType{n}{\valType}
      \end{aligned}\,}{>}$
    \columnbreak
    \Procedure{\writeProc}{$i : \natType, v : \valType$}{} \AComment{$w_i$}
    \miniskip
      \State $s \set \scanProc()$ \RComment{$\wrs{w_i}$}
      \State $\resA[i] \set \left\langle v, \resA[i].\seqField + 1, s \right\rangle$ \RComment{$\wra{w_i}$}
        \label{alg:afek-wra}
    \EndProcedure
  \end{algorithmic}
  \end{multicols}
  \noindent\rule{0.94\textwidth}{0.4pt}
  \begin{multicols}{2}
  \begin{algorithmic}[1]
    \setcounterref{ALG@line}{alg:afek-wra}
    \Procedure{\scanProc}{}{$\arrayType{n}{\valType}$} \AComment{$s$}
    \miniskip
    \For{$i \in \left\{ 0 \dots n-1 \right\}$}
    \State $\movedVar[i] \set \false$
    \EndFor
      \While{$\true$} \hspace{3.6em} {\color{gray}$k$-th iteration of loop}
        \For{$i \in \left\{ 0 \dots n-1 \right\}$}
        \State $a[i] \set \resA[i]$ \RComment{$\scanth{s}{i}{k}$}
        \EndFor
        \For{$i \in \left\{ 0 \dots n-1 \right\}$}
        \State $b[i] \set \resA[i]$ \RComment{$\scbnth{s}{i}{k}$}
        \EndFor
        \State $\changedVar \set \false$
        \For{$i \in \left\{ 0 \dots n-1 \right\}$}
          \If{$a[i].\seqField \neq b[i].\seqField$}
            \If{$\movedVar[i]$}
              \State \textbf{return} $b[i].\viewField$ \label{alg:afek-ret2}
            \Else
              \State $\changedVar \set \true$
              \State $\movedVar[i] \set \true$  
            \EndIf
          \EndIf
        \EndFor
        \If{$\neg \changedVar$}
          \State \textbf{return} $(b[0].\dataField, \dots, b[n-1].\dataField)$ \label{alg:afek-ret1}
        \EndIf
      \EndWhile
      \State \textbf{end while}
    \EndProcedure
  \end{algorithmic}
  \end{multicols}
  \caption{\label{alg:afek} Single-writer, multi-scanner snapshot algorithm of
    \citet{afek:acm93}.\vspace{-4mm}
  }
\end{algorithm}

If $\scanProc$ detects that some index $i$ has changed twice during
scanning, it can immediately terminate, returning the \emph{view} of
the latest write into $i$ ($\resA[i].\viewField$). The view itself is
a snapshot each $\writeProc$ collects by calling $\scanProc$ before
writing its value. Keeping views ensures that the snapshot methods are
wait-free. $\scanProc$ tracks how many times index $i$ has changed by
updating the local variable $\movedVar[i]$ each time it detects a
change. Each $\writeProc$ commits its value simultaneously with
incrementing the version number and updating the view.

We next use visibility relations to show that \cref{alg:afek}
satisfies the $\snapshotSig$ signature, and is thus linearizable by
\cref{lem:ss-lin}.

\begin{lemma}
  Every execution of \cref{alg:afek} satisfies the $\snapshotSig$
  signature (\cref{fig:snapshot}).
\end{lemma}

\begin{proof}[Proof sketch]
  To prove that \cref{alg:afek} satisfies the $\snapshotSig$
  signature, we employ virtual scans as a simplification mechanism
  that allows us to elide $\writeProc$ views from consideration when
  establishing $\snapshotSig$.
  A virtual scan will denote a scan that detected no change in the
  array, i.e., it returned at line~\ref{alg:afek-ret1}. To each abs
  scan, we can associate a virtual scan as follows. If the abs scan
  terminated because it detected no change, than it immediately is a
  virtual scan. If the abs scan terminated by returning a view from a
  $\writeProc$, i.e., returned at line~\ref{alg:afek-ret2}, then that
  view itself is an abs scan which can, recursively, be associated
  with a virtual scan. Importantly, given abs scan $\abs{s}$, the
  associated virtual scan $\SMap{\abs{s}}$ satisfies $\SMap{\abs{s}}
  \subev \abs{s}$, from which we shall derive the $\snapshotSig$
  axioms.

  Formally, a virtual scan $\vrt{\vsc}$ consists of rep events
  $\rep{\sca{\vsc}{i}}$ and $\rep{\scb{\vsc}{i}}$ for each $i$, where
  both $\rep{\sca{\vsc}{i}}$ and $\rep{\scb{\vsc}{i}}$ observe the
  same write (i.e., no change detected), and every
  $\rep{\sca{\vsc}{i}}$ occurs before any $\rep{\scb{\vsc}{j}}$:

  \smallskip
  {
    \renewcommand{\arraystretch}{\eqSpacing}%
    \hfill\begin{tabular*}{\textwidth}{@{\hskip 1em}l@{\hskip 10em}c@{\extracolsep{\fill}}r@{\hskip 1.5em}}
      $\forall \vrt{\vsc}, i.$ &$\exists \abs{w_i}.\ \rep{\wra{w_i} \robs
       \sca{\vsc}{i}} \land \rep{\wra{w_i} \robs \scb{\vsc}{i}}$ &\insertEq\label[prop]{ax:afek-success}\\
      $\forall \vrt{\vsc}, i, j.$ &$\rep{\sca{\vsc}{i} \rb \scb{\vsc}{j}}$ &\insertEq\label[prop]{ax:afek-scrb}
    \end{tabular*}
  }

  \noindent
  If scan $\abs{s}$ returned at line~\ref{alg:afek-ret1}, we simply
  define the virtual scan $\SMap{\abs{s}}$ as the set of rep events
  $\rep{\scanth{s}{i}{k}}$ and $\rep{\scbnth{s}{i}{k}}$, where $k$ is
  the last reading iteration of $\abs{s}$. If scan $\abs{s}$ returned
  at line~\ref{alg:afek-ret2}, then $\abs{s}$ returns the view of some
  $\abs{w_i}$, and we define $\SMap{\abs{s}} =
  \SMap{(\rep{\wrs{w_i}})}$. This is a well-founded recursive
  definition, because the ending time of the scan $\rep{\wrs{w_i}}$ on
  the right is smaller than the ending time of $\abs{s}$ on the left,
  and the ending times are bounded from below by $0$. To see that the
  ending time of $\rep{\wrs{w_i}}$ is below that of $\abs{s}$, suppose
  otherwise. Then it must be $\abs{s} \rb \rep{\wra{w_i}}$, because
  from the code of $\writeProc$ we have that $\rep{\wrs{w_i} \rb
    \wra{w_i}}$. In particular, $\rep{\scbnth{s}{i}{k}} \rb
  \rep{\wra{w_i}}$ for every $k$. But, because $\rep{\wra{w_i}}$ is
  observed by some rep read in $\abs{s}$, we also have $\rep{\wra{w_i}
    \robs \scbnth{s}{i}{k}}$ for some $k$.  Thus, we have an event
  $\rep{\scbnth{s}{i}{k}}$ that terminated before the event
  $\rep{\wra{w_i}}$ that it observes, which
  contradicts~\cref{ax:wfobs}.

  Next, we show that $\SMap{\abs{s}} \subev \abs{s}$.
  If $\abs{s}$ terminated with no changes detected, this is trivial,
  since $\SMap{\abs{s}}$ consists of the selected rep events of
  $\abs{s}$. Otherwise, $\SMap{\abs{s}} = \SMap{(\rep{\wrs{w_i}})}$,
  for some write view $\rep{\wrs{w_i}}$. By recursion on the
  definition of $\SMap{[-]}$ it must be
  $\SMap{\abs{s}} = \SMap{\rep{\wrs{w_i}}} \subev \rep{\wrs{w_i}}$, so
  it suffices to show $\rep{\wrs{w_i}} \subev \abs{s}$.
  This holds because, to reach the return at line~\ref{alg:afek-ret2},
  two changes of $\resA[i]$ must have occurred, thus there are two
  different writes to $\resA[i]$ before $\abs{w_i}$ that where
  observed by some rep event in $\abs{s}$, implying $\abs{s}$ starts
  before $\abs{w_i}$ and thus also before $\rep{\wrs{w_i}}$.
  Additionally, since some $\rep{\scbnth{s}{i}{k}}$ must have read
  $\rep{\wra{w_i}}$, i.e., $\rep{\wra{w_i} \robs \scbnth{s}{i}{k}}$,
  it must be that $\rep{\wrs{w_i}}$ ends before $\abs{s}$, as
  otherwise $\rep{\scbnth{s}{i}{k} \rb \wra{w_i}}$ contradicts
  \cref{ax:wfobs}.

  We next proceed to establish the $\snapshotSig$ signature. First, we
  instantiate $\abs{\WrEff_i}$ to be the set of all writes $\abs{w_i}
  \in \abs{W_i}$ where $\rep{\wra{w_i}}$ is defined; these are the
  writes that executed their effect. Second, we instantiate the
  visibility relations $\abs{\obs}$ and $\abs{\robs}$ as follows,
  using the helper relation $\abs{\hlrobs}$.
  \begin{align*}
    \abs{w_i \hlrobs \vrt{\vsc}} &\wideDefeq \rep{\wra{w_i} \robs \sca{\vsc}{i}} &
    \abs{w_i \robs s} &\wideDefeq \abs{w_i \hlrobs \SMap{\abs{s}}} &
    \abs{\obs} &\wideDefeq \abs{\robs}
  \end{align*}
  In English, the scan $\abs{s}$ reads from, and also observes,
  $\abs{w_i}$ iff there is an appropriate rep event in the virtual
  scan $\SMap{\abs{s}}$ that reads from $\rep{\wra{w_i}}$ at the level
  of rep events.
  We now argue that the definitions of the visibility relations
  satisfy the $\snapshotSig$ properties. \cref{ax:wfobs} holds since
  if $\abs{e \obs^+ e'}$ and $\abs{e' \rbeq e}$ then $\abs{e \obs^+
    e'}$ can only hold if $\abs{e \robs e'}$, i.e., $\rep{\wra{w_i}
    \robs \sca{\vsc}{i}}$. But then, by \cref{ax:rb-absrep}, from
  $\abs{e' \rbeq e}$ we derive $\rep{\sca{\vsc}{i} \rb \wra{w_i}}$,
  which contradicts \cref{ax:wfobs} for registers. \cref{ax:ss-io}
  holds by the structure of the algorithm. \cref{ax:ss-robsuniq} holds
  since if $\abs{w_i \robs s}$ and $\abs{w'_i \robs s}$, then
  $\rep{\wra{w_i}}$ and $\rep{\wra{w'_i}}$ are observed by the same
  read. Thus, by \cref{ax:mem-robsuniq}, $\rep{\wra{w_i} =
    \wra{w'_i}}$, and since each write executes $\rep{\wra{w_i}}$ only
  once by Property~\ref{ax:evsig-uniq}, it must be $\abs{w_i =
    w'_i}$. \cref{ax:ss-wrtotal} holds because the algorithm is
  assumed to be single-writer. \cref{ax:ss-wrterm} holds since each
  terminated write must have $\rep{\wra{w_i}}$ defined, which is how
  we define the set $\abs{\WrEff_i}$.

  Next, \cref{ax:ss-nowrbetween} says that given $\abs{w_i \robs s}$,
  there exists no $\abs{w'_i}$ such that $\abs{w_i \hb w'_i \hb s}$.
  To prove it, we assume that $\abs{w'_i}$ exists, and derive a
  contradiction. By assumption, we have $\abs{w_i \hlrobs \vrt{\vsc}}$
  for $\vrt{\vsc} = \SMap{\abs{s}}$. From
  $\vrt{\vsc} = \SMap{\abs{s}} \subev \abs{s}$, and
  $\abs{w_i \hb w'_i \hb s}$, we can derive
  $\abs{w_i \hb w'_i \hb \vrt{\vsc}}$, and then by
  \cref{ax:afek-scrb}, also
  $\rep{\wra{w_i} \hb \wra{w'_i} \hb \scb{\vsc}{i}}$. We elide the
  proof of the last derivation; it is similar to
  \cref{lem:fwd-hbabsrep} and is in
  Appendix~\apndxAfek{}\ifNotAppendix{ of~\cite{extended}}. By the
  definition of $\abs{\hlrobs}$ and \cref{ax:afek-success},
  $\abs{w_i \hlrobs \vrt{\vsc}}$ implies
  $\rep{\wra{w_i} \robs \sca{\vsc}{i}, \scb{\vsc}{i}}$. But then
  $\rep{\wra{w'_i}}$ occurrs between $\rep{\wra{w_i}}$ and
  $\rep{\scb{\vsc}{i}}$, contradicting \cref{ax:mem-nowrbetween} and
  $\rep{\wra{w_i} \robs \scb{\vsc}{i}}$.

  Lastly, we prove \cref{ax:ss-mono}: $\abs{w_i, w_j \robs s}$ and
  $\abs{w'_i, w'_j \robs s'}$ with $\abs{w_i \hb w'_i}$ and $\abs{w'_j
    \hb w_j}$ lead to contradiction. We unfold the assumptions
  $\abs{w_i, w_j \robs s}$ and $\abs{w'_i, w'_j \robs s'}$ into
  $\abs{w_i, w_j \hlrobs \vrt{\vsc}}$ and $\abs{w'_i, w'_j \hlrobs
    \vrt{\vsc'}}$, where $\vrt{\vsc} = \SMap{\abs{s}}$ and
  $\vrt{\vsc'} = \SMap{\abs{s'}}$.  By definition of
  $\abs{\hlrobs}$ and \cref{ax:afek-success}, from $\abs{w_j \hlrobs
    \vrt{\vsc}}$ we have $\rep{\wra{w_j} \robs \sca{\vsc}{j},
    \scb{\vsc}{j}}$, and similarly for $\abs{w'_j \hlrobs
    \vrt{\vsc'}}$. By \cref{ax:afek-scrb} we have $\rep{\sca{\vsc}{j}
    \rb \scb{\vsc}{i}}$ and $\rep{\sca{\vsc'}{i} \rb \scb{\vsc'}{j}}$, and
  by \cref{ax:interval}, it is either $\rep{\sca{\vsc}{j} \rb
    \scb{\vsc'}{j}}$ or $\rep{\sca{\vsc'}{i} \rb \scb{\vsc}{i}}$. In
  the first case (the second is symmetric) we can construct
  $\rep{\wra{w'_j} \hb \wra{w_j} \robs \sca{\vsc}{j} \rb
    \scb{\vsc'}{j}}$, meaning that we have $\rep{\wra{w'_j} \hb
    \wra{w_j} \hb \scb{\vsc'}{j}}$ contradicting
  \cref{ax:mem-nowrbetween} for $\rep{\wra{w'_j} \robs
    \scb{\vsc'}{j}}$ by $\rep{\wra{w_j}}$ occurring between
  $\rep{\wra{w'_j}}$ and $\rep{\scb{\vsc'}{j}}$.
\end{proof}

We conclude this section by observing that the inverse of
\cref{lem:ss-lin} also holds; that is, $\snapshotSig$ signature is
actually satisfiable by every linearizable snapshot algorithm. This
follows by instantiating $\abs{\obs}$ with the linearization order,
defining $\abs{\robs}$ to relate each read with the latest write
before that read in the linearization, and defining $\abs{\WrEff_i}$
as the set of all writes in the linearization. The axioms of the
signature then simply state straightforward properties of the
linearization order and of sequential execution in that order.
That said, we have established $\snapshotSig$ signature for
\cref{alg:afek} \emph{without} assuming linearizability, so that we
can derive linearizability by \cref{lem:ss-lin}.

\section{Related Work}

The idea to use visibility relations for proving linearizability has
first been proposed by \citet{henzinger:concur13} and applied to
concurrent queue algorithms. The explicit motivation of this approach
was to modularize the linearizability proofs. 
The work on time-stamped stack algorithm of~\citet{dodds:popl15} builds
on this by introducing a hybrid approach, combining the visibility
method together with the linearization points method. Visibility has
also been applied on algorithms over weak memory, where it is
typically called communication order~\cite{raad:popl19}. Another
related approach of visibility by \citet{emm+ene:popl19,enea:esop20}
formulates weakened specifications for linearizability, where all
operations do not have to linearly occur one after another in a
simulation, but occur in partial order, allowing for certain
operations to miss another.

Concerning Jayanti's algorithms specifically, \citet{jayanti:stoc05}
sketched linearizability of the first algorithm by describing its
linearization points. He also argued informally, based on forwarding
principles, that the changes between the algorithms preserve
linearizability.
\citet{delbianco:ecoop17} gave a proof formalized in Coq of the first
algorithm, by constructing a linearization order along the execution
of the algorithm. The proof is specific to the implementation, making
it is unclear how to lift it to the other two algorithms. Based on
that development, \citet{jacobs:jayanti} developed a mechanized proof,
also of the first algorithm, in VeriFast using prophecy variables.
\citet{petrank:disc13} and \citet{timnat:thesis} present an algorithm
based on the forwarding idea of Jayanti, which implements a set
interface with insert, remove and contains operations, along with an
iterator which is a generalized form of scanner. The set operations
can be forwarded to the scanner, generalizing Jayanti's forwarding
which applies only to write operations. In the future, we will
consider how to generalize our proof to apply to this algorithm as
well. This would involve axiomatizing the set data-structure, but also
what it means abstractly to be a forwarding structure, so that both
Jayanti's snapshot and the set structure of Petrank and Timnat are
instances.

\citet{afek:acm93} proved their algorithm linearizable by using the
linearization point method. While the use of the linearization point
method is fairly straightforward for this algorithm, we showed that we
can reuse our snapshot axiomatization for its proof.

\section{Conclusion and Future Work}

We presented proofs of linearizability of Jayanti's three snapshot
algorithms developed in a modular fashion, using the method of
visibility relations. More concretely, the linearizability proofs are
decomposed into proof modules, with many of the modules being shared
between the three algorithms; they are developed once and reused three
times to reduce proof complexity.

Importantly, the module interfaces are signatures consisting of
relations and axioms on them that encode the key idea of ``forwarding
principles'' underpinning Jayanti's design. We thus show that a
formalism based on visibility relations is powerful enough to
mathematically capture these principles; previously, Jayanti has only
presented them informally in English.

In the future, we plan to apply the visibility methods to other
algorithms, and potentially also use the developed modules and
signatures as guides in designing new and more efficient algorithms,
much like Jayanti's three algorithms start with the
single-writer/single-scanner variant and build to the ultimately
desired multi-writer/multi-scanner variant. We shall also study how
the visibility method applies to non-linearizable algorithms, and
how to mechanize our proof in a proof assistant.

\begin{acks}
  We thank the anonymous reviewers from the POPL’22 PC for their
  feedback. This research was partially supported by the Spanish
  MICINN projects BOSCO (PGC2018-102210-B-I00) and the European
  Research Council project Mathador (ERC2016-COG-724464). 
\end{acks}

\bibliographystyle{ACM-Reference-Format}
\bibliography{references}

\ifappendix
\newpage
\appendix

\section{Snapshot Signature Implies Linearizability}
\label{apx:ss2lin}

Here we prove that every algorithm with histories satisfying the
$\snapshotSig$ signature is linearizable. This proof is based on a
similar proof for queues by~\citet{henzinger:lmcs15}.

Following \cref{def:lin} of linearizability, we start by choosing a
visibility relation $\abs{\obs}$ that constructs the set
$\abs{E_c} = \{ \abs{e} \mid \exists \abs{e'} \in \term(\abs{E}).\
\abs{e \obs^* e'} \}$, and then extend it to a linearization
$\abs{<}$. Let $\abs{\obs}$ be the visibility relation obtained by
restricting the visibility relation postulated by $\snapshotSig$ to
$\term(\abs{S}) \cup \bigcup_{i} \abs{\WrEff_i}$. Events outside of
the restriction are non-terminated scans or writes that have not
executed their effect yet; they do not modify the abstract state, and
are hence not necessary for linearization.
Additionally, by \cref{ax:ss-wrterm} we know that every terminated
write is effectful. Thus going forward, we consider $\abs{E_c} =
\term(\abs{S}) \cup \bigcup_{i} \abs{\WrEff_i}$.

To move towards the construction of a linearization, we will first
construct a partial order $\abs{\whb}$ over all write events that
contains all logically implied orderings. We will in steps expand this
order, until we have a total order.

\begin{definition}
  Let $\abs{w_i \whbI w'_j}$ be defined if $i\neq j$ and there exists
  a write $\abs{w_j}$ and a scan $\abs{s}$ such that
  $\abs{w_i,w_j \robs s}$ and $\abs{w_j \hb w'_j}$. Let
  $\abs{\whb} \defeq (\abs{\hb}\ \cup\ \abs{\whbI})^+$.
\end{definition}

\begin{lemma}
  \label{lem:whb-irrefl}
  The relation $\abs{\whb}$ is a strict partial order.
\end{lemma}
\begin{proof}
  For a relation to be a strict partial order, it needs to be
  transitive and irreflexive. The relation $\abs{\whb}$ is transitive
  by definition since it is a transitive closure. We prove
  irreflexivity by proving acyclicity of
  $\abs{\whbS} = \abs{\hb}\ \cup\ \abs{\whbI}$, i.e., we prove there
  exists no cycle of the following form:
  \[
    \abs{w_1 \whbS w_2 \whbS \cdots \whbS w_n \whbS w_1}
  \]
  We show this by showing that there exists no cycle of size $n = 1$
  and that we can shrink any cycle of size $n \ge 2$, meaning that we
  can recursively shrink a cycle until it is of size $n = 1$. For a
  cycle with $n = 1$, we have $\abs{w_1 \whbS w_1}$, which can only be
  the case $\abs{w_1 \hb w_1}$, since two writes joined by
  $\abs{\whbI}$ must come from distinct memory cells. However, we
  cannot have $\abs{w_1 \hb w_1}$ either, since it contradicts
  $\abs{\hb}$ being irreflexive (\cref{lem:hb-irrefl}).
  We show that we can shrink any cycle of $n \ge 2$ by doing a case
  analysis on the first two instances of $\abs{\whbS}$ in the cycle:
  \begin{enumerate}
  \item We have $\abs{w_1 \hb w_2 \hb w_3 \whbS^* w_1}$. By
    transitivity of $\abs{\hb}$ we can shrink the cycle to the smaller
    cycle $\abs{w_1 \hb w_3 \whbS^* w_1}$.
  \item We have $\abs{w_1 \whbI w_2 \hb w_3 \whbS^* w_1}$, i.e.,
    $\abs{w_1} = \abs{w_i}$ and $\abs{w_2} = \abs{w'_j}$
    with $i \neq j$ and there exists $\abs{w_j}$ such that
    $\abs{w_i, w_j \robs s}$ and $\abs{w_j \hb w'_j}$ and
    $\abs{w_2 \hb w_3 \whbS^* w_1}$. Let $\abs{w_3} = \abs{w''_k}$
    where $k$ can either be equal to $i$ or $j$ or some other distinct
    index. Since each scan is terminated and thus observes the entire
    memory by \cref{ax:ss-io}, there must exist some $\abs{w_k}$ such
    that $\abs{w_k \robs s}$. By comparing $\abs{w_k}$ and
    $\abs{w''_k}$ with \cref{ax:ss-wrtotal}, we have either
    $\abs{w''_k \hbeq w_k}$ or $\abs{w_k \hb w''_k}$. If we have
    $\abs{w''_k \hbeq w_k}$ then we also have $\abs{w'_j \hb s}$ by
    $\abs{w'_j \hb w''_k \hbeq w_k \robs s}$, which implies that
    $\abs{w'_j}$ happens between $\abs{w_j}$ and $\abs{s}$,
    contradicting \cref{ax:ss-nowrbetween}. Thus we must have
    $\abs{w_k \hb w''_k}$, which lets us construct
    $\abs{w_1 \whbI w_3 \whbS w_1}$, giving us a shorter cycle.
  \item We have $\abs{w_1 \whbI w_2 \whbI w_3 \whbS^* w_1}$, i.e.,
    $\abs{w_1} = \abs{w_i}$ and $\abs{w_2} = \abs{w'_j}$ and
    $\abs{w_3} = \abs{w''_k}$ with $\abs{w_i, w_j \robs s}$ and
    $\abs{w'_j, w'_k \robs s'}$ and $\abs{w_j \hb w'_j}$ and
    $\abs{w'_k \hb w''_k}$ and $\abs{w_3 \whbS^* w_1}$ where
    $i \neq j$ and $j \neq k$. Similar to the previous case, by
    \cref{ax:ss-io}, there must exist a $\abs{w_k}$ such that
    $\abs{w_k \robs s}$. By \cref{ax:ss-mono} with
    $\abs{w_j, w_k \robs s}$ and $\abs{w'_j, w'_k \robs s'}$ and
    $\abs{w_j \hb w'_j}$, it must follow that
    $\abs{w'_k \not\hb w_k}$, which by \cref{ax:ss-wrtotal} implies
    $\abs{w_k \hbeq w'_k}$, giving us $\abs{w_k \hb w''_k}$ by
    $\abs{w_k \hbeq w'_k \hb w''_k}$. This lets us construct
    $\abs{w_1 \whbI w_3 \whbS^* w_1}$ giving us a shorter cycle.
  \item We have $\abs{w_1 \hb w_2 \whbI w_3 \whbS^* w_1}$. For this
    case we simply \emph{rotate} the cycle to move the first instance
    to the end of the cycle, i.e.,
    $\abs{w_2 \whbI w_3 \whbS^* w_1 \hb w_2}$, and case match on the
    next two $\abs{\whbS}$ instances to create a scenario where we
    either use case (2) or case (3). \qedhere
  \end{enumerate}
\end{proof}

By Zorn's Lemma, we can arbitrarily choose a total order $\abs{\wlt}$
extending $\abs{\whb}$, for which we define the concept of maximal
candidate, which we use to construct the complete total order.

\begin{definition}
  An event $\abs{e}$ is called a maximal candidate of $\abs{\wlt}$
  when $\abs{e}$ is $\abs{\hb}$-maximal (i.e., there exists no greater
  event in $\abs{\hb}$, which may not necessarily be unique) and one
  of the following holds:
  \begin{enumerate}
  \item $\abs{e} = \abs{w_i}$ is a write with no scans observing it
    and it is the greatest write in $\abs{\wlt}$,
  \item $\abs{e} = \abs{s}$ is a scan with for all $i$, there exists
    $\abs{w_i}$ such that $\abs{w_i \robs s}$ and $\abs{w_i}$ is the
    greatest event of order $\abs{\wlt}$ in $\abs{\WrEff_i}$.
  \end{enumerate}
\end{definition}

\begin{lemma}\label{lem:maxcand-populated}
  For any non-empty history satisfying the $\snapshotSig$ signature,
  there exists a maximal candidate of $\abs{\wlt}$.
\end{lemma}
\begin{proof}
  Since $\abs{\wlt}$ is a total order over a finite set, there exists
  a unique greatest write $\abs{w_i}$, consider the following cases:
  \begin{enumerate}
  \item If nothing reads $\abs{w_i}$ and $\abs{w_i}$ is
    $\abs{\hb}$-maximal, then $\abs{w_i}$ is trivially a maximal
    candidate.
  \item If nothing reads $\abs{w_i}$ and $\abs{w_i}$ is not
    $\abs{\hb}$-maximal, then there exists an event $\abs{e}$ greater
    than $\abs{w_i}$, such that $\abs{w_i \hb e}$. The event cannot be
    a write event $\abs{e} = \abs{w'_j}$, since that would imply
    $\abs{w_i \wlt w'_j}$, contradicting $\abs{w_i}$ being
    $\abs{\wlt}$-maximal. Thus it must be a scan event
    $\abs{e} = \abs{s}$, by \cref{ax:ss-io} there exists $\abs{w'_i}$
    such that $\abs{w'_i \robs s}$ with $\abs{w_i} \neq \abs{w'_i}$,
    since nothing reads $\abs{w_i}$. Since we have
    $\abs{w'_i \robs s}$ and $\abs{w_i \hb s}$, it must be the case
    that $\abs{w_i}$ does not happen in-between $\abs{w'_i}$ and
    $\abs{s}$ as per \cref{ax:ss-nowrbetween}, i.e., we have
    $\abs{w'_i \not\hb w_i}$, thus by \cref{ax:ss-wrtotal} we have
    $\abs{w_i \hb w'_i}$. With $\abs{w_i}$ being
    $\abs{\wlt}$-maximal, we also have $\abs{w'_i \wlt w_i}$, however
    this contradicts $\abs{w_i \hb w'_i}$.
  \item If there exists $\abs{s}$ such that $\abs{w_i \robs s}$ with
    $\abs{s}$ being $\abs{\hb}$-maximal, then we show that $\abs{s}$
    is a maximal candidate For $\abs{s}$ to be a maximal candidate, we
    need for all $j$, there exists $\abs{w_j \robs s}$ such that
    $\abs{w_j}$ is $\abs{\wlt}$-maximal in $\abs{\WrEff_j}$. By
    \cref{ax:ss-io}, we know that there exists some $\abs{w_j}$ such
    that $\abs{w_j \robs s}$. Suppose for each $j \neq i$, $\abs{w_j}$
    is not $\abs{\wlt}$-maximal in $\abs{\WrEff_j}$, then there exists
    $\abs{w'_j}$ such that $\abs{w_j \wlt w'_j}$. This implies
    $\abs{w_j \hb w'_j}$ by \cref{ax:ss-wrtotal} and since
    $\abs{\wlt}$ extends $\abs{\hb}$, which by
    $\abs{w_i, w_j \robs s}$ and $\abs{w_j \hb w'_j}$ lets us
    construct $\abs{w_i \whbI w'_j}$, which implies
    $\abs{w_i \wlt w'_j}$, however this contradicts maximality of
    $\abs{w_i}$. Thus each $\abs{w_j}$ is $\abs{\wlt}$-maximal in
    $\abs{\WrEff_j}$, meaning $\abs{s}$ is a maximal candidate.
  \item If there exists $\abs{s}$ such that $\abs{w_i \robs s}$ with
    $\abs{s}$ not being $\abs{\hb}$-maximal, then there must exist
    some other event $\abs{e}$ such that $\abs{s \hb e}$. If there
    exists multiple scans observing $\abs{w_i}$ we let $\abs{s}$ be
    the $\abs{\hb}$-greatest among them. If $\abs{e}$ is a write event
    $\abs{e} = \abs{w'_j}$ then we have $\abs{w_i \hb w'_j}$ since
    $\abs{w_i \robs s \hb w'_j}$, which contradicts $\abs{w_i}$ being
    the greatest write in $\abs{\wlt}$. If $\abs{e}$ is a scan
    $\abs{e} = \abs{s'}$ then we either have $\abs{w_i \robs s'}$ or
    $\abs{w'_i \robs s'}$ for some write $\abs{w'_i} \neq \abs{w_i}$,
    however the former contradicts our selection of $\abs{s}$ since
    both $\abs{s}$ and $\abs{s'}$ observe $\abs{w_i}$, but
    $\abs{s \hb s'}$ contradicts that $\abs{s}$ be the
    $\abs{\hb}$-greatest among scans observing $\abs{w_i}$. We thus
    have that $\abs{w'_i \robs s'}$ and since we also have
    $\abs{w_i \hb s'}$ by $\abs{w_i \robs s \hb s'}$, if $\abs{w_i}$
    is in-between $\abs{w'_i}$ and $\abs{s'}$ we contradict
    \cref{ax:ss-nowrbetween}, thus we have $\abs{w'_i \not\hb w_i}$
    which by \cref{ax:ss-wrtotal} we must have $\abs{w_i \hb w'_i}$,
    which contradicts that $\abs{w_i}$ is $\abs{\wlt}$-maximal.
    \qedhere
\end{enumerate}
\end{proof}

To construct the total order, we recursively select a maximal
candidate, which we make the last event of the total order and remove
from consecutive steps.

\begin{theorem}
  Every history satisfying our snapshot specifications is
  linearizable, i.e., there exists a total order $\abs{<}$ over each
  event which forms a legal sequential snapshot behavior and extends
  $\abs{\wlt}$ and $\abs{\hb}$.
\end{theorem}
\begin{proof}
  By induction on the number of events in the history. The minimal
  history, the one that only consists of the initial writes, is
  clearly linearizable, no matter what order $\abs{\wlt}$ we have.

  Our inductive hypothesis is that for some history
  $\abs{E \subseteq E_c}$ satisfying the axioms under $\abs{\hb}$,
  there exists a total order $\abs{<}$ over $\abs{E}$ that forms a
  legal sequential snapshot behavior and extends $\abs{\wlt}$ and
  $\abs{\hb}$ when only considering events in $\abs{E}$. For the
  inductive step, we need to prove for a there exists such a total
  order $\abs{<'}$ of history $\abs{E'}$, which extends $\abs{E}$ by
  the single event $\abs{e}$, with $\abs{e}$ being an arbitrary
  maximal candidate of $\abs{\wlt}$ under history $\abs{E'}$. By
  \cref{lem:maxcand-populated}, since $\abs{E'}$ is non-empty, there
  must exist a maximal candidate. Consider the two following cases for
  $\abs{e}$.
  \begin{itemize}
  \item $\abs{e} = \abs{w_i}$ is a write. By the definition of being a
    maximal candidate, we know that there exists no scan in $\abs{E'}$
    observing $\abs{w_i}$. If we remove $\abs{w_i}$ from $\abs{E'}$
    then the new history $\abs{E}$ still satisfies the snapshot
    properties since it is safe to remove any event that is not read
    visible. By the inductive hypothesis we have a total order
    $\abs{<}$ over $\abs{E}$ that forms a legal snapshot behavior and
    extends $\abs{\wlt}$ and $\abs{\hb}$. To construct the total order
    $\abs{<'}$ over $\abs{E'}$, we extend $\abs{<}$ with write
    $\abs{w_i}$ being the last event in the order. Since any write
    event added to the end of $\abs{<'}$ cannot influence the result
    of previous scan events, it means that since $\abs{<}$ forms a
    legal snapshot behavior, so must $\abs{<'}$.
  \item $\abs{e} = \abs{s}$ is a scan. By the definition of being a
    maximal candidate, we know that all the writes $\abs{s}$ observe
    are maximal under $\abs{\wlt}$ in their respective memory cell. If
    we remove $\abs{s}$ from $\abs{E'}$ then the new history $\abs{E}$
    still satisfies the snapshot properties since no other events
    depend on scans. By the inductive hypothesis we have a total order
    $\abs{<}$ over $\abs{E}$ that forms legal snapshot behavior and
    extends $\abs{\wlt}$ and $\abs{\hb}$. To construct the total order
    $\abs{<'}$ over $\abs{E'}$, we extend $\abs{<}$ with scan
    $\abs{s}$ being the last event in the order. The $\abs{<'}$ order
    forms a legal snapshot behavior, since any write $\abs{w_i}$ that
    $\abs{s}$ observes is maximal in $\abs{\wlt}$, it is impossible
    for another write to be in-between $\abs{w_i}$ and $\abs{s}$, thus
    since $\abs{<}$ forms a legal snapshot behavior, so does
    $\abs{<'}$. \qedhere
  \end{itemize}
\end{proof}

\section{Forwarding Signature Implies Snapshot Signature}\label{apx:fwd2ss}

We here give the full proof that any history satisfying the $\fwdSig$
signature (\cref{fig:forward}) also satisfies the $\snapshotSig$
signature (\cref{fig:snapshot}). We start by determining the
instantiations of $\abs{\WrEff_i}$ and $\abs{\robs}$ and $\abs{\obs}$
for the $\snapshotSig$ signature. Let $\abs{w_i} \in \abs{\WrEff_i}$
iff $\abs{w_i}$ has executed $\rep{\wra{w_i}}$, since the effect of
$\abs{w_i}$ is only observable after this execution. We define a
helper notion of visibility over scans $\abs{\sobs}$, which we use to
define $\abs{\obs}$. The highlighted relations $\abs{\hlrobs}$ and
$\abs{\hlwobs}$ originate from the $\fwdSig$ signature.
\begin{align*}
  \abs{w_i \robs s} &\wideDefeq \abs{w_i \hlrobs \SMap{\abs{s}}}\\
  \abs{s \sobs s'} &\wideDefeq \vrt{\SMap{\abs{s}} \rb \SMap{\abs{s'}}}\\
  \abs{\obs} &\wideDefeq \abs{\robs} \cup \abs{\hlwobs} \cup \abs{\sobs}
\end{align*}
We prove the $\snapshotSig$ properties by unfolding the definition of
$\abs{\robs}$ in the properties so that we can apply the $\fwdSig$
properties. 

\begin{genthm}{\cref{ax:ss-io}}
  For every terminated scan $\abs{s}$ and index $i$, there exists some
  write $\abs{w_i}$ such that $\abs{w_i \robs s}$ and
  $\abs{w_i}.\evIn = \abs{s}.\evOut[i]$.
\end{genthm}

\begin{proof}
  By \cref{ax:fwd-io}, there exists $\abs{w_i}$ such as
  $\abs{w_i \hlrobs \SMap{\abs{s}}}$ and
  $\abs{w_i}.\evIn = \abs{s}.\evOut[i]$ for $\abs{s}$, where the
  latter immediately corresponds to one of our goals and the former
  corresponds to the remaining goal by the instantiation of
  $\abs{\robs}$ giving us
  $\abs{w_i \hlrobs \SMap{\abs{s}}} = \abs{w_i \robs s}$.
\end{proof}

\begin{genthm}{\cref{ax:ss-wrterm}}
  Every terminated write is effectful, i.e.,
  $\term(\abs{W_i}) \subseteq \abs{\WrEff_i}$.
\end{genthm}

\begin{proof}
  Since every terminated event must have executed all their rep
  events, it follows that a terminated write $\abs{w_i}$ must have
  executed $\rep{\wra{w_i}}$, which matches the instantiation of
  $\abs{\WrEff_i}$.
\end{proof}

\begin{genthm}{\cref{ax:ss-wrtotal}}
  For two distinct writes $\abs{w_i}$ and $\abs{w'_i}$, we have either
  $\abs{w_i \hb w'_i}$ or $\abs{w'_i \hb w_i}$.
\end{genthm}

\begin{proof}
  For each write, we either have $\rep{\wra{w_i} \hb \wra{w'_i}}$ or
  $\rep{\wra{w'_i} \hb \wra{w_i}}$ by \cref{ax:mem-wrtotal}. By the
  definition of $\abs{\hlwobs}$, if we have
  $\rep{\wra{w_i} \hb \wra{w'_i}}$ then we have $\abs{w_i \hlwobs w'_i}$
  (and thus also $\abs{w_i \hb w'_i}$) and if we have
  $\rep{\wra{w'_i} \hb \wra{w_i}}$ then we have $\abs{w'_i \hlwobs w_i}$
  (and thus also $\abs{w'_i \hb w_i}$).
\end{proof}

\begin{genthm}{\cref{ax:wfobs}}\label{lem:fwd-wfobs}
  If $\abs{e \obs^+ e'}$ then we cannot have
  $\abs{e' \rbeq e}$.
\end{genthm}

\begin{proof}
  Assume we have $\abs{e' \rbeq e}$, we derive a contradiction.
  Consider first that we have $\abs{e = e'}$ such that we have
  $\abs{e \obs^+ e}$. There is no $\abs{w_i \robs s}$ in
  $\abs{e \obs^+ e}$ since if there is, there is no way to relate from
  a scan back to a write, i.e., there is no observation from scan to
  write. Thus we must have either $\abs{e \hlwobs^+ e}$ or
  $\abs{e \sobs^+ e}$, which we cannot have by total ordering of
  $\rep{\wra{w_i}}$ events by \cref{ax:mem-wrtotal} or by virtual scan
  total order by \cref{ax:fwd-sctotal}. Thus $\abs{e}$ and $\abs{e'}$
  cannot be equal.
  
  Assuming instead we have $\abs{e' \rb e}$, we derive a
  contradiction. By \cref{lem:fwd-hbabsrep} from $\abs{e \obs^+ e'}$
  we have $\rep{e_r \hb e'_r}$ for some $\rep{e_r} \subev \abs{e}$ and
  $\rep{e'_r} \subev \abs{e'}$. By \cref{ax:rb-absrep} we can derive
  $\rep{e'_r \rb e_r}$ from $\abs{e' \rb e}$, we can construct the
  cycle $\rep{e_r \hb e'_r \rb e_r}$ which contradicts irreflexivity
  of $\rep{\hb}$ (\cref{lem:hb-irrefl}).
\end{proof}

By \cref{lem:hb-irrefl} and Virtual \cref{ax:wfobs}, we have that
$\abs{\hb}$ is irreflexive, and thus it is a strict partial order. We
will use this to prove following properties.

\begin{genthm}{\cref{ax:ss-robsuniq}}
  \label{lem:fwd-robsuniq}
  If $\abs{w_i \robs s}$ and $\abs{w'_i \robs s}$ then we must have
  $\abs{w_i = w'_i}$.
\end{genthm}

\begin{proof}
  Let $\SMap{\abs{s}} = \vrt{\vsc}$, by unfolding $\abs{w_i \robs s}$
  we have $\abs{w_i \hlrobs \vrt{\vsc}}$, and dually for
  $\abs{w'_i \robs s}$. By case splitting
  $\abs{w_i \hlrobs \vrt{\vsc}}$ and $\abs{w'_i \hlrobs \vrt{\vsc}}$.
  Remember that $\abs{w_i \hlrobs \vrt{\vsc}}$ splits into either
  $\rep{\wra{w_i} \robs \sca{\vsc}{i}}$ and
  $\rep{\scr{\vsc}{i} \robs \scb{\vsc}{i}}$, or
  $\abs{w_i \fobs \vrt{\vsc}}$ and similarly for
  $\abs{w'_i \hlrobs \vrt{\vsc}}$.
  \begin{itemize}
  \item If $\abs{w_i \hlrobs \vrt{\vsc}}$ and
    $\abs{w'_i \hlrobs \vrt{\vsc}}$ splits into different cases, e.g.,
    we have $\abs{w_i \fobs \vrt{\vsc}}$ and
    $\rep{\scr{\vsc}{i} \robs \scb{\vsc}{i}}$, then by
    \cref{ax:fobsbot} we derive a contradiction.
  \item If both split into the first case, i.e., we have
    $\rep{\wra{w_i} \robs \sca{\vsc}{i}}$ and
    $\rep{\wra{w'_i} \robs \sca{\vsc}{i}}$, then by
    \cref{ax:mem-robsuniq} we have $\rep{\wra{w_i} = \wra{w'_i}}$.
    Since each abstract write has a distinct representative write, we
    must have $\abs{w_i = w'_i}$.
  \item If both split into the second case, i.e., we have
    $\abs{w_i \fobs \vrt{\vsc}}$ and $\abs{w'_i \fobs \vrt{\vsc}}$ then by
    \cref{ax:fobsuniq} we have $\abs{w_i = w'_i}$. \qedhere
  \end{itemize}
\end{proof}

\begin{lemma}\label{lem:hb-ss2fwd}
  If $\abs{w_i \hb s}$ then $\abs{w_i \hb \SMap{\abs{s}}}$.
\end{lemma}

\begin{proof}
  By viewing $\abs{w_i \hb s}$ as a chain of $\abs{\hb_1}$, we have
  $\abs{w_i \hb_1 \dots \hb_1 s}$. By induction on the length of the
  chain. In our base case we have $\abs{w_i \hb_1 s}$, which splits
  into cases $\abs{w_i \rb s}$ and $\abs{w_i \obs s}$, which can only
  be $\abs{w_i \robs s}$. For $\abs{w_i \rb s}$, by
  \cref{ax:rb-absrep,ax:fwd-vrtinscan}, we have
  $\abs{w_i \rb \SMap{\abs{s}}}$. For $\abs{w_i \robs s}$, we have
  $\abs{w_i \hlrobs \SMap{\abs{s}}}$ directly from the definition of
  $\abs{\robs}$. The inductive step is similar, with the exception
  that we need to also consider for
  $\abs{w_i \hb_1 \dots \hb_1 s' \hb_1 s}$ the $\abs{s' \sobs s}$
  case, which by the definition of $\abs{\sobs}$ gives us
  $\vrt{\SMap{\abs{s'}} \rb \SMap{\abs{s}}}$, which implies
  $\abs{\SMap{\abs{s'}} \hb \SMap{\abs{s}}}$, allowing us to construct
  $\abs{w_i \hb \SMap{\abs{s}}}$.
\end{proof}

\begin{genthm}{\cref{ax:ss-nowrbetween}}
  \label{lem:fwd-nowrbetween}
  If $\abs{w_i \robs s}$ then there does not exist a
  write $\abs{w'_i}$ such that $\abs{w_i \hb w'_i \hb s}$ holds.
\end{genthm}

\begin{figure}[t]
\centering
\tikzset{
mystyle/.style={
  align=center
  }
}

\tikzset{
mystyleC/.style={
  minimum width=5em,
  align=center
  }
}

\tikzset{
mystyleD/.style={
  minimum width=4.5em,
  minimum height=1.8em,
  align=center
  }
}

\tikzset{
mystyleE/.style={
  minimum width=6em,
  align=center
  }
}

\tikzset{
mystyleF/.style={
  minimum width=6em,
  minimum height=1.8em,
  align=center
  }
}

\begin{tikzpicture}
  \node[mystyle] (A) at (0,0) {$\abs{w'_i \hb \vrt{\vsc}}$};
  \node[mystyle] (B) at (1.7,0) {$\abs{w'_i \hb_1\hbeq \vrt{\vsc}}$};

  \node[mystyleC] (C1) at (4,1) {$\abs{w'_i \rb\hbeq \vrt{\vsc}}$};
  \node[mystyleC,draw=black] (C2) at (4,0)  {$\abs{w'_i \hlwobs\hbeq \vrt{\vsc}}$};
  \node[mystyleC] (C3) at (4,-1) {$\abs{w'_i \hlrobs\hbeq \vrt{\vsc}}$};

  \node[mystyleD] (D1) at (6,1.5) {$\abs{w'_i \rb\hb \vrt{\vsc}}$};
  \node[mystyleD] (D2) at (6,0.5) {$\abs{w'_i \rb \vrt{\vsc}}$};

  \node[mystyleD,draw=black] (D3) at (6,-0.5) {$\abs{w'_i \hlrobs \vrt{\vsc}}$};
  \node[mystyleD] (D4) at (6,-1.5) {$\abs{w'_i \hlrobs\hb \vrt{\vsc}}$};

  \node[mystyleE] (E1) at (8,1.5) {$\abs{w'_i \rb\hbeq\hb_1 \vrt{\vsc}}$};
  \node[mystyleE] (E2) at (8,-1.5) {$\abs{w'_i \hlrobs\hbeq\hb_1 \vrt{\vsc}}$};

  \node[mystyleF,draw=black] (F1) at (10.5, 2) {$\abs{w'_i \rb\hbeq\hlrobs \vrt{\vsc}}$};
  \node[mystyleF] (F3) at (10.5, 1) {$\abs{w'_i \rb\hbeq\rb \vrt{\vsc}}$};

  \node[mystyleF] (F4) at (10.5, -1) {$\abs{w'_i \hlrobs\hbeq\rb \vrt{\vsc}}$};
  \node[mystyleF] (F6) at (10.5, -2) {$\abs{w'_i \hlrobs\hbeq\hlrobs \vrt{\vsc}}$};

  \node[mystyle,draw=black] (G) at (12.5,0) {$\rep{\wra{w'_i} \hb \scr{\vsc}{i}}$};

  \draw[-implies,double equal sign distance] (A.east) -- (B.west);
  
  \draw[->] (B.east) -- (C1.west);
  \draw[->] (B.east) -- (C2.west);
  \draw[->] (B.east) -- (C3.west);

  \draw[->] (C1.east) -- (D1.west);
  \draw[->] (C1.east) -- (D2.west);

  \draw[->] (C3.east) -- (D3.west);
  \draw[->] (C3.east) -- (D4.west);

  \draw[-implies,double equal sign distance] (D1.30) to [out=90,in=90] (E1.150);
  \draw[-implies,double equal sign distance] (D2.east) to [out=-20,in=180] (G.west);
  \draw[-implies,double equal sign distance] (D4.330) to [out=-90,in=-90] (E2.210);

  \draw[->] (E1.east) -- (F1.west);
  \draw[->] (E1.east) -- (F3.west);

  \draw[->] (E2.east) -- (F4.west);
  \draw[->] (E2.east) -- (F6.west);

  \draw[-implies,double equal sign distance] (F3.east) -- (G);
  \draw[-implies,double equal sign distance] (F4.east) -- (G);
  \draw[-implies,double equal sign distance] (F6.east) to [out=0,in=0] (F4.south east);
\end{tikzpicture} 
\caption{\label{fig:nowrbetween-cases} Directed acyclic graph
  representing the proof of splitting $\abs{w'_i \hb \vrt{\vsc}}$ into
  four final (boxed) cases. Forks in the graph with $\to$ arrows
  represents case splitting, splitting either $\abs{\hb_1}$ or
  $\abs{\hbeq}$. Implication arrows ($\Rightarrow$) represents an
  implication.}
\end{figure}
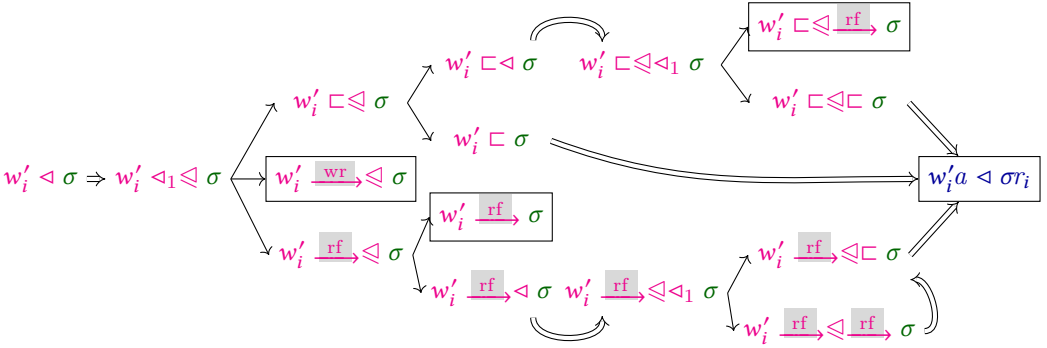

\begin{proof}
  Assume we have such a write $\abs{w'_i}$, we derive a contradiction.
  Let $\vrt{\vsc} = \SMap{\abs{s}}$, we unfold $\abs{w_i \robs s}$
  into $\abs{w_i \hlrobs \vrt{\vsc}}$. By \cref{lem:hb-ss2fwd} we have
  $\abs{w'_i \hb \vrt{\vsc}}$ from $\abs{w'_i \hb s}$. We first show
  that there exists a write $\abs{w''_i}$ with $\abs{w_i \hb w''_i}$
  such that either $\rep{\wra{w''_i} \hb \scr{\vsc}{i}}$ or
  $\abs{w''_i \rb\hbeq\hlrobs \vrt{\vsc}}$ hold. We split
  $\abs{w'_i \hb \vrt{\vsc}}$ into multiple cases, multiple times, as
  shown in \cref{fig:nowrbetween-cases}. We justify each of the
  implications:
  \begin{itemize}
  \item Any implication from something with $\abs{\hb}$ into either
    $\abs{\hb_1\hbeq}$ or $\abs{\hbeq\hb_1}$ simply holds by
    definition.
  \item The implication $\abs{w'_i \hlrobs\hbeq\hlrobs \vrt{\vsc}}$ into
    $\abs{w'_i \hlrobs\hbeq\rb \vrt{\vsc}}$ holds since there has to
    exist a write $\abs{w_j}$ and a scan $\vrt{\vsc'}$ such that
    $\abs{w'_i \hlrobs \vrt{\vsc'} \hb w_j \hlrobs \vrt{\vsc}}$. By
    \cref{ax:fwd-sctotal}, we either have $\vrt{\vsc \rbeq \vsc'}$
    or $\vrt{\vsc' \rb \vsc}$, where in the former case we
    have $\abs{\vrt{\vsc'} \hb w_j \hlrobs \vrt{\vsc \rbeq \vsc'}}$,
    contradicting irreflexivity of $\abs{\hb}$, thus we must have
    the latter. Thus we have $\abs{w'_i \hlrobs \vrt{\vsc' \rb \vsc}}$.
  \item All the implications going into $\rep{\wra{w'_i} \hb \scr{\vsc}{i}}$
    holds by \cref{lem:fwd-hbabsrep,ax:rb-absrep}.
  \end{itemize}
 
  We now consider each of the possible final cases of
  \cref{fig:nowrbetween-cases}:
  \begin{itemize}
  \item If we have $\abs{w'_i \hlwobs\hbeq \vrt{\vsc}}$, then we must
    have some $\abs{w''_i}$ such that
    $\abs{w'_i \hlwobs w''_i \hb \vrt{\vsc}}$. We recursively consider
    the cases of \cref{fig:nowrbetween-cases} with
    $\abs{w_i \hb w''_i \hb \vrt{\vsc}}$. We cannot have an infinite
    recursive chain of this case, since the set of writes is finite
    and thus, the only way to have such an infinite chain is by having
    a cycle in $\abs{\hlwobs}$, which is impossible since such a cycle
    implies we also have a cycle of rep events.
  \item If we have $\abs{w'_i \hlrobs \vrt{\vsc}}$ then we contradict
    the result of proving \cref{ax:ss-robsuniq} and
    $\abs{w_i \hlrobs \vrt{\vsc}}$ since $\abs{w_i}$ and $\abs{w'_i}$
    cannot be equal by $\abs{w_i \hb w'_i}$.
  \end{itemize}
  Thus the only remaining cases are
  $\rep{\wra{w''_i} \hb \scr{\vsc}{i}}$ and
  $\abs{w''_i \rb\hbeq\hlrobs \vrt{\vsc}}$. We case split
  $\abs{w_i \hlrobs \vrt{\vsc}}$ and consider all possible combined
  cases:
  \begin{itemize}
  \item If $\rep{\wra{w_i} \robs \sca{\vsc}{i}}$ and
    $\rep{\wra{w''_i} \hb \scr{\vsc}{i}}$, then $\rep{\wra{w''_i}}$
    happens in-between $\rep{\wra{w_i}}$ and $\rep{\sca{\vsc}{i}}$,
    thus we contradict \cref{ax:mem-nowrbetween}.
  \item If $\rep{\wra{w_i} \robs \sca{\vsc}{i}}$ and
    $\rep{\scr{\vsc}{i} \robs \scb{\vsc}{i}}$ and
    $\abs{w''_i \rb\hbeq\hlrobs \vrt{\vsc}}$, then from \cref{ax:forward1} we
    have $\rep{\wra{w''_i} \hb \sca{\vsc}{i}}$, implying
    $\rep{\wra{w''_i}}$ happened in-between $\rep{\wra{w_i}}$ and
    $\rep{\sca{\vsc}{i}}$, thus we contradict
    \cref{ax:mem-nowrbetween}.
  \item If $\abs{w_i \fobs \vrt{\vsc}}$ and
    $\rep{\wra{w''_i} \hb \scr{\vsc}{i}}$, then \cref{ax:forward2a}
    contradicts $\abs{w_i \hb w''_i}$.
  \item If $\abs{w_i \fobs \vrt{\vsc}}$ and
    $\abs{w''_i \rb\hbeq\hlrobs \vrt{\vsc}}$, then \cref{ax:forward2b}
    contradicts $\abs{w_i \hb w''_i}$. \qedhere
  \end{itemize}
\end{proof}

\begin{genthm}{\cref{ax:ss-mono}}\label{lem:fwd-mono}
  If $\abs{w_i, w_j \robs s}$ and $\abs{w'_i, w'_j \robs s'}$ with
  $\abs{w_i \hb w'_i}$ then we cannot have $\abs{w'_j \hb w_j}$.
\end{genthm}

\begin{proof}
  Assume we have $\abs{w'_j \hb w_j}$, we derive a contradiction. Let
  $\vrt{\vsc} = \SMap{\abs{s}}$ and $\vrt{\vsc'} = \SMap{\abs{s'}}$,
  we unfold $\abs{w_i, w_j \robs s}$ into
  $\abs{w_i, w_j \hlrobs \SMap{\abs{s}}}$ and similarly for
  $\abs{w'_i, w'_j \robs s'}$. By \cref{ax:fwd-sctotal}, we either
  have $\vrt{\vsc = \vsc'}$ or $\vrt{\vsc \rb \vsc'}$ or
  $\vrt{\vsc' \rb \vsc}$. If we have $\vrt{\vsc = \vsc'}$, we
  contradict the result of proving \cref{ax:ss-robsuniq}, since we
  would have $\abs{w_i \hlrobs \vrt{\vsc}}$ and
  $\abs{w'_i \hlrobs \vrt{\vsc}}$, but $\abs{w_i}$ and $\abs{w'_i}$
  cannot be equal since we have $\abs{w_i \hb w'_i}$. If we have
  $\vrt{\vsc \rb \vsc'}$, then we have
  $\abs{w'_j \hb w_j \hlrobs \vrt{\vsc \rb \vsc'}}$, contradicting the
  result of proving \cref{ax:ss-nowrbetween} by $\abs{w_j}$ occurring
  in between $\abs{w'_j \hlrobs \vrt{\vsc'}}$. Similarly, if we have
  $\vrt{\vsc' \rb \vsc}$, then we have
  $\abs{w_i \hb w'_i \hlrobs \vrt{\vsc' \rb \vsc}}$, contradicting the
  result of proving \cref{ax:ss-nowrbetween} by $\abs{w'_i}$ occurring
  in between $\abs{w_i \hlrobs \vrt{\vsc}}$.
\end{proof}

\section{Algorithm \ref{alg:jay1} Satisfies the Forwarding Signature}\label{apx:jay1}

We here present the complete proof for Jayanti's
single-writer/single-scanner algorithm (\cref{alg:jay1}) satisfying
the $\fwdSig$ signature (\cref{fig:forward}).

Since this algorithm is single-scanner, we instantiate the set of
virtual scans to be equal to the set of abs scans,
$\vrt{\Vsc} = \abs{S}$, and let $\SMap{[-]}$ be the identity mapping,
$\SMap{\abs{s}} = \abs{s}$. We let $\rep{\resA}$ correspond to the
array of the same name and $\rep{\resB}$ for each virtual scan will
only map to $\rep{\resB}$. The rep events of the event signatures for
\writeProc{} and \vscanProc{} (\cref{fig:struct}) corresponds to the
rep events of the algorithm of the same name, i.e., $\rep{\wra{w_i}}$
and $\rep{\scr{s}{i}}$, $\rep{\sca{s}{i}}$ and $\rep{\scb{s}{i}}$. We
instantiate forwarding visibility to the following:
\[
  \abs{w_i \fobs s} \wideDefeq \rep{\wrb{w_i} \robs \scb{s}{i}}
\]
I.e., a write $\abs{w_i}$ was forwarded to scan $\abs{s}$ if
$\rep{\scb{s}{i}}$ read the value written by $\rep{\wrb{w_i}}$.

By $\SMap{[-]}$ being instantiated as the identity mapping,
\cref{ax:fwd-vrtinscan} holds, since for all $\abs{e}$,
$\abs{e} \subev \abs{e}$. By the algorithm being single-scanner implies
that virtual scan total order \cref{ax:fwd-sctotal} holds. The
structure of the algorithm and the instantiation of $\abs{\fobs}$
implies that \cref{ax:fwd-io,ax:fwd-scstruct} hold. Additionally,
since the algorithm only writes to $\resA$ in writers and writes
$\bot$ to $\resB$ in scanners, \cref{ax:fwd-wrauniq,ax:fwd-scruniq}.
This leaves Properties~\eqref{ax:fobsuniq} to \eqref{ax:forward2b} to
prove.

\begin{genthm}{\cref{ax:fobsuniq}}
  If we have $\abs{w_i \fobs s}$ and $\abs{w'_i \fobs s}$ then
  $\abs{w_i}$ and $\abs{w'_i}$ must be equal.
\end{genthm}

\begin{proof}
  By our instantiation we have $\rep{\wrb{w_i} \robs \scb{s}{i}}$ and
  $\rep{\wrb{w'_i} \robs \scb{s}{i}}$, thus by \cref{ax:mem-robsuniq} we have
  $\rep{\wrb{w_i} = \wrb{w'_i}}$. Since all writes only performs their
  operations at most once, this means that the writes must be the
  same, i.e., $\abs{w_i = w'_i}$.
\end{proof}

\begin{genthm}{\cref{ax:fobsbot}}
  If $\rep{\scb{s}{i}}$ terminated without
  $\rep{\scr{s}{i} \robs \scb{s}{i}}$ then there exists $\abs{w_i}$
  such that $\abs{w_i \fobs \vrt{s}}$.
\end{genthm}

\begin{proof}
  Since we do not have $\rep{\scr{s}{i} \robs \scb{s}{i}}$, it is the
  case that $\rep{\scb{s}{i}}$ could not have observed any other
  $\rep{\scr{s'}{i}}$ since such $\abs{s'}$ must have occurred before
  $\abs{s}$ by single-scanner, implying
  $\rep{\scr{s'}{i} \rb \scr{s}{i}}$, meaning we would contradict
  \cref{ax:mem-nowrbetween} if we had
  $\rep{\scr{s'}{i} \robs \scb{s}{i}}$ by $\rep{\scr{s}{i}}$ occurring
  in between them. Thus, $\rep{\scb{s}{i}}$ must have observed some
  $\rep{\wrb{w_i}}$, i.e., $\rep{\wrb{w_i} \robs \scb{s}{i}}$,
  implying our goal $\abs{w_i \fobs s}$.
\end{proof}

\begin{genthm}{\cref{ax:fobshb}}
  If we have $\abs{w_i \fobs s}$ then we have
  $\rep{\wra{w_i} \hb \scb{s}{i}}$ without
  $\rep{\scr{s}{i} \robs \scb{s}{i}}$.
\end{genthm}

\begin{proof}
  Since $\abs{w_i \fobs s}$ is instantiated as
  $\rep{\wrb{w_i} \robs \scb{s}{i}}$, by
  $\rep{\wra{w_i} \rb \wrb{w_i} \robs \scb{s}{i}}$ we have
  $\rep{\wra{w_i} \hb \scb{s}{i}}$. For proving absence of
  $\rep{\scr{s}{i} \robs \scb{s}{i}}$, assume we have it, we derive a
  contradiction. By $\abs{w_i \fobs s}$ we have
  $\rep{\wrb{w_i} \robs \scb{s}{i}}$, which by \cref{ax:mem-robsuniq}
  we have $\rep{\wrb{w_i} = \scr{s}{i}}$, however that is impossible,
  since $\rep{\wrb{w_i}}$ is a rep event of $\abs{w_i}$, while
  $\rep{\scr{s}{i}}$ is a rep event of $\abs{s}$, thus a
  contradiction.
\end{proof}

Note that we now satisfies \cref{ax:fwd-scstruct,ax:fobshb}, meaning we
satisfy the preconditions to use \cref{lem:fwd-hbabsrep}. Before we
prove \cref{ax:forward1,ax:forward2a,ax:forward2b}, we prove some
helper lemmas.
\begin{lemma}%
  \label{lem:jay1obsX}
  For any execution, we have $\rep{\scon{s} \robs \wrx{w_i}}$ iff we
  have $\rep{\scon{s} \hb \wrx{w_i}}$ and
  $\rep{\scoff{s} \not\hb \wrx{w_i}}$.
\end{lemma}

\begin{proof}\leavevmode
  \begin{itemize}
  \item[$(\mathord{\implies})$] $\rep{\scon{s} \robs \wrx{w_i}}$
    trivially implies $\rep{\scon{s} \hb \wrx{w_i}}$, and if we were
    to have $\rep{\scoff{s} \hb \wrx{w_i}}$ we contradict
    \cref{ax:mem-nowrbetween} by $\rep{\scoff{s}}$ occurring in
    between $\rep{\scon{s} \robs \wrx{w_i}}$.
  \item[$(\mathord{\impliedby})$] By the fact that only scans change
    $\resX$ and there is a single-scanner constraint, we know that any
    write to $\resX$ other than $\rep{\scon{s}}$ and $\rep{\scoff{s}}$
    either happens before or after the scan $\abs{s}$. Since
    $\rep{\wrx{w_i}}$ has to observe something when it finishes, it
    must observe $\rep{\scon{s}}$, since if $\rep{e \robs \wrx{w_i}}$
    for some write event $\rep{e} \neq \rep{\scon{s}}$, then by
    \cref{ax:mem-wrtotal} we either have $\rep{e \hb \scon{s}}$, which
    contradicts \cref{ax:mem-nowrbetween} by $\rep{\scon{s}}$
    occurring in between $\rep{e \robs \wrx{w_i}}$; or
    $\rep{\scon{s} \hb e}$, which by the algorithm structure
    $\rep{e} = \rep{\scoff{s}}$ or some later write, however this
    contradicts $\rep{\scoff{s} \not\hb \wrx{w_i}}$. \qedhere
  \end{itemize}
\end{proof}

\begin{lemma}
  \label{lem:jay1rblinpoint}
  If $\abs{w'_i \rbhbeqrobs s}$ then
  $\rep{\scoff{s} \not\hb \wrx{w'_i}}$ and if $\rep{\wrb{w'_i}}$ was
  executed, then $\rep{\wrb{w'_i} \hb \scb{s}{i}}$.
\end{lemma}

\begin{proof}
  For $\abs{w'_i \rbhbeqrobs s}$, there must exist some $\abs{e}$
  and $\abs{w_j}$ such that $\abs{w'_i \rb e \hbeq w_j \robs s}$. Each
  of these events are populated except for $\abs{s}$, since each of
  the events are to the left of some $\abs{\hb}$. For
  $\abs{w'_i \rb e \hbeq w_j}$, let $\rep{e_r}$ be a rep event
  executed by $\abs{e}$, then by \cref{ax:rb-absrep,lem:fwd-hbabsrep}
  we have $\rep{\wrb{w'_i} \rb e_r \hbeq \wra{w_j}}$. We case split
  $\abs{w_j \robs s}$ into either $\rep{\wra{w_j} \robs \sca{s}{j}}$
  or $\rep{\wrb{w_j} \robs \scb{s}{j}}$.
  \begin{itemize}
  \item If we have $\rep{\wra{w_j} \robs \sca{s}{j}}$ then we have the chain
    \[
      \rep{\wrx{w'_i} \rb \wrb{w'_i} \hb \wra{w_j} \robs \sca{s}{j} \rb \scoff{s} \rb \scb{s}{i}}
    \]
    which trivially implies $\rep{\wrb{w'_i} \hb \scb{s}{i}}$ and
    $\rep{\scoff{s} \not\hb \wrx{w'_i}}$.
  \item If we have $\rep{\wrb{w_j} \robs \scb{s}{j}}$ then since
    $\rep{\wrb{w_j}}$ wrote, we know that $\rep{\wrx{w_j}}$ read
    $\true$, i.e., $\rep{\scon{s'} \robs \wrx{w_j}}$ for some
    $\abs{s'}$. By single-scanner, we have either $\abs{s \rb s'}$ or
    $\abs{s' \rbeq s}$, which in the former case we can form the cycle
    \[
      \rep{\scb{s}{j} \rb \scon{s'} \robs \wrx{w_j} \rb \wrb{w_j}
        \robs \scb{s}{j}}
    \]
    contradicting irreflexivity of $\rep{\hb}$, thus we must have
    $\abs{s' \rbeq s}$. By \cref{lem:jay1obsX} we have
    $\rep{\scoff{s'} \not\hb \wrx{w_j}}$, which by $\abs{s' \rbeq s}$
    implies $\rep{\scoff{s} \not\hb \wrx{w_j}}$. By \cref{ax:interval}
    over $\rep{\wrb{w'_i} \rb e_r}$ and
    $\rep{\scoff{s} \rb \scb{s}{i}}$ we have
    $\rep{\wrb{w'_i} \rb \scb{s}{i}}$ since if we were to have
    $\rep{\scoff{s} \rb e_r}$ then we contradict
    $\rep{\scoff{s} \not\hb \wrx{w_j}}$ by
    $\rep{\scoff{s} \rb e_r \hbeq \wrx{w_j}}$. For the
    $\rep{\scoff{s} \not\hb \wrx{w'_i}}$ goal, assume we have
    $\rep{\scoff{s} \hb \wrx{w'_i}}$, then we derive a contradiction
    by having $\rep{\scoff{s} \hb \wrx{w'_i} \hb \wrx{w_j}}$ which
    contradicts $\rep{\scoff{s} \not\hb \wrx{w_j}}$. \qedhere
  \end{itemize}
\end{proof}

\begin{genthm}{\cref{ax:forward1}}
  If $\abs{w'_i \rbhbeqrobs s}$ and $\rep{\scr{s}{i} \robs \scb{s}{i}}$
  then $\rep{\wra{w'_i} \hb \sca{s}{i}}$.
\end{genthm}

\begin{proof}
  By \cref{lem:jay1rblinpoint} we have
  $\rep{\scoff{s} \not\hb \wrb{w'_i}}$ and
  $\rep{\wrb{w'_i} \hb \scb{s}{i}}$ if $\rep{\wrb{w'_i}}$ was
  executed. Consider if we have $\rep{\scon{s} \hb \wrx{w'_i}}$, then
  by \cref{lem:jay1obsX} we have $\rep{\scon{s} \robs \wrx{w'_i}}$
  which in turn implies that $\rep{\wrb{w'_i}}$ must have executed,
  however that contradicts \cref{ax:mem-nowrbetween} by
  $\rep{\wrb{w'_i}}$ occurring in between
  $\rep{\scr{s}{i} \robs \scb{s}{i}}$, thus we must have
  $\rep{\scon{s} \not\hb \wrx{w'_i}}$.

  By \cref{ax:interval} over $\rep{\wra{w'_i} \rb \wrx{w'_i}}$ and
  $\rep{\scon{s} \rb \sca{s}{i}}$, since $\rep{\scon{s} \not\hb \wrx{w'_i}}$
  contradicts $\rep{\scon{s} \rb \wrx{w'_i}}$, we must have
  $\rep{\wra{w'_i} \rb \sca{s}{i}}$, which is our goal, thus we are done.
\end{proof}

\begin{genthm}{\cref{ax:forward2a,ax:forward2b}}
  If $\abs{w_i \fobs s}$ and either $\abs{w'_i \rbhbeqrobs s}$ or
  $\rep{\wra{w'_i} \hb \scr{s}{i}}$ then $\abs{w_i \not\hb w'_i}$.
\end{genthm}

\begin{proof}
  By our instantiation we have $\rep{\wrb{w_i} \robs \scb{s}{i}}$ from
  $\abs{w_i \fobs s}$. Assuming we have $\abs{w_i \hb w'_i}$ (thus we
  also have $\rep{\wra{w_i} \hb \wra{w'_i}}$ by
  \cref{lem:fwd-hbabsrep}), we derive a contradiction.

  Let us start with the case where we have
  $\rep{\wra{w'_i} \hb \scr{s}{i}}$. By the single-writer constraint
  we have
  $\rep{\wrb{w_i} \hb \wra{w'_i} \hb \scr{s}{i} \rb \scb{s}{i}}$ which
  contradicts \cref{ax:mem-nowrbetween} by $\rep{\scr{s}{i}}$
  occurring in between $\rep{\wrb{w_i} \robs \scb{s}{i}}$, thus this
  case is not possible. For the next case we assume
  $\rep{\wra{w'_i} \not\hb \scr{s}{i}}$.

  For the $\abs{w'_i \rbhbeqrobs s}$ case, by
  \cref{lem:jay1rblinpoint} we have
  $\rep{\scoff{s} \not\hb \wrb{w'_i}}$ and
  $\rep{\wrb{w'_i} \hb \scb{s}{i}}$ if $\rep{\wrb{w'_i}}$ was
  executed. By \cref{ax:interval} over
  $\rep{\wra{w'_i} \rb \wrx{w'_i}}$ and
  $\rep{\scon{s} \rb \scr{s}{i}}$, since we have
  $\rep{\wra{w'_i} \not\hb \scr{s}{i}}$ contradicting
  $\rep{\wra{w'_i} \rb \scr{s}{i}}$ we must have
  $\rep{\scon{s} \rb \wrx{w'_i}}$. By \cref{lem:jay1obsX} we thus have
  $\rep{\scon{s} \robs \wrx{w'_i}}$, which means that
  $\rep{\wrx{w'_i}}$ reads $\true$ and thus $\rep{\wrb{w'_i}}$ is
  executed. By the single-writer constraint we have
  $\rep{\wrb{w_i} \hb \wrb{w'_i} \hb \scb{s}{i}}$ which contradicts
  \cref{ax:mem-nowrbetween} by $\rep{\wrb{w'_i}}$ occurring in between
  $\rep{\wrb{w_i} \robs \scb{s}{i}}$.
\end{proof}

\section{Multi-writer Forwarding satisfies the Forwarding signature}
\label{apx:mwfwd2fwd}

We here give the complete proof that any history satisfying the
$\mwFwdSig$ signature (\cref{fig:forward2}) also satisfies the
$\fwdSig$ signature (\cref{fig:forward}). We instantiate $\vrt{\Vsc}$
and $\SMap{[-]}$ of $\fwdSig$ to be the same as the corresponding
instantiations of $\mwFwdSig$.
\cref{ax:fwd-vrtinscan,ax:fwd-io,ax:fwd-wrauniq,ax:fwd-scruniq,ax:fwd-sctotal}
are shared between the signatures, therefore they trivially hold.

\begin{genthm}{\cref{ax:fwd-scstruct}}
  Every virtual scan satisfies
  $\rep{\scr{\vsc}{i} \rb \sca{\vsc}{i} \rb \scb{\vsc}{i}}$.
\end{genthm}

\begin{proof}
  By \cref{ax:fwd2-scstruct}, we have
  \[
    \rep{\scr{\vsc}{i} \rb \scon{\vsc} \robseq \sconobs{\vsc} \rb
      \sca{\vsc}{i} \rb \scoff{\vsc} \robseq \scoffobs{\vsc} \rb
      \scb{\vsc}{i}}
  \]
  by \cref{ax:interval} over $\rep{\scr{\vsc}{i} \rb \scon{\vsc}}$ and
  $\rep{\sconobs{\vsc} \rb \sca{\vsc}{i}}$, we either have
  $\rep{\scr{\vsc}{i} \rb \sca{\vsc}{i}}$ or
  $\rep{\sconobs{\vsc} \rb \scon{\vsc}}$, however the latter
  contradicts $\rep{\scon{\vsc} \robseq \sconobs{\vsc}}$, thus we must
  have $\rep{\scr{\vsc}{i} \rb \sca{\vsc}{i}}$. Similarly, by
  \cref{ax:interval} over $\rep{\sca{\vsc}{i} \rb \scoff{\vsc}}$ and
  $\rep{\scoffobs{\vsc} \rb \scb{\vsc}{i}}$, we either have
  $\rep{\sca{\vsc}{i} \rb \scb{\vsc}{i}}$ or
  $\rep{\scoffobs{\vsc} \rb \scoff{\vsc}}$, where the latter
  contradicts $\rep{\scoff{\vsc} \robseq \scoffobs{\vsc}}$, thus we
  must have $\rep{\sca{\vsc}{i} \rb \scb{\vsc}{i}}$.
\end{proof}

\begin{genthm}{\cref{ax:fobsuniq}}
  If we have $\abs{w_i \fobs \vrt{\vsc}}$ and
  $\abs{w'_i \fobs \vrt{\vsc}}$ then $\abs{w_i}$ and $\abs{w'_i}$ must
  be equal.
\end{genthm}

\begin{proof}
  By the instantiation we have $\rep{\wra{w_i} \robs \wrfa{f}}$ and
  $\rep{\wrfbB{f} \robs \scb{\vsc}{i}}$ and
  $\rep{\wra{w'_i} \robs \wrfa{f'}}$ and
  $\rep{\wrfbB{f'} \robs \scb{\vsc}{i}}$ from
  $\abs{w_i \fobs \vrt{\vsc}}$ and $\abs{w'_i \fobs \vrt{\vsc}}$
  respectively. By \cref{ax:mem2-robsuniq} over
  $\rep{\wrfbB{f} \robs \scb{\vsc}{i}}$ and
  $\rep{\wrfbB{f'} \robs \scb{\vsc}{i}}$, we have
  $\rep{\wrfbB{f} = \wrfbB{f'}}$, implying $\vrt{f = f'}$ since each
  $\vrt{f}$ performs $\rep{\wrfbB{f}}$ at most once. Thus we have
  $\rep{\wra{w_i} \robs \wrfa{f}}$ and
  $\rep{\wra{w'_i} \robs \wrfa{f}}$, which by \cref{ax:mem-robsuniq}
  we have $\rep{\wra{w_i} = \wra{w'_i}}$, which again since we only
  have one write to $\resA[i]$ per abs write we have
  $\abs{w_i = w'_i}$.
\end{proof}

\begin{genthm}{\cref{ax:fobsbot}}
  If $\rep{\scb{\vsc}{i}}$ terminated without
  $\rep{\scr{\vsc}{i} \robs \scb{\vsc}{i}}$ then there exists
  $\abs{w_i}$ such that $\abs{w_i \fobs \vrt{\vsc}}$.
\end{genthm}

\begin{proof}
  Since we do not have $\rep{\scr{\vsc}{i} \robs \scb{\vsc}{i}}$, it is
  the case that $\rep{\scb{\vsc}{i}}$ could not have observed any
  other $\rep{\scr{\vsc'}{i}}$ since such $\abs{\vsc'}$ must have
  occurred before $\abs{\vsc}$ by virtual scanner total order
  \cref{ax:fwd-sctotal}, implying
  $\rep{\scr{\vsc'}{i} \rb \scr{\vsc}{i}}$, and
  $\rep{\scr{\vsc}{i} \hb \scb{\vsc}{i}}$ by \cref{ax:fwd2-scstruct},
  meaning we would contradict \cref{ax:mem2-nowrbetween} if we had
  $\rep{\scr{\vsc'}{i} \robs \scb{\vsc}{i}}$ by $\rep{\scr{\vsc}{i}}$
  occurring in between them.
  By \cref{ax:fwd-scruniq} this means that whatever
  $\rep{\scb{\vsc}{i}}$ read, it cannot have been $\bot$. By
  \cref{ax:fwd2-fbBuniq}, $\rep{\scb{\vsc}{i}}$ must have observed
  some $\rep{\wrfbB{f}}$, i.e., $\rep{\wrfbB{f} \robs \scb{\vsc}{i}}$.
  This forwarding $\vrt{f}$ must have read something with
  $\rep{\wrfa{f}}$ by the structure enforced by
  \cref{ax:fwd2-fwdstruct}, which by \cref{ax:fwd-wrauniq} must be
  some $\rep{\wra{w_i}}$. Thus we have
  $\rep{\wra{w_i} \robs \wrfa{f}}$ and
  $\rep{\wrfbB{f} \robs \scb{\vsc}{i}}$, matching the instantiation of
  $\abs{\fobs}$.
\end{proof}

\begin{genthm}{\cref{ax:fobshb}}
  If we have $\abs{w_i \fobs \vrt{\vsc}}$ then we have
  $\rep{\wra{w_i} \hb \scb{\vsc}{i}}$ without
  $\rep{\scr{\vsc}{i} \robs \scb{\vsc}{i}}$.
\end{genthm}

\begin{proof}
  Since $\abs{w_i \fobs \vrt{\vsc}}$ is instantiated as
  $\rep{\wra{w_i} \robs \wrfa{f}}$ and
  $\rep{\wrfbB{f} \robs \scb{\vsc}{i}}$, and since we have
  $\rep{\wrfa{f} \rb \wrfbB{f}}$ we have
  $\rep{\wra{w_i} \robs \wrfa{f} \rb \wrfbB{f} \robs \scb{\vsc}{i}}$,
  implying $\rep{\wra{w_i} \hb \scb{\vsc}{i}}$. For proving absence of
  $\rep{\scr{\vsc}{i} \robs \scb{\vsc}{i}}$, assume we have it, we
  derive a contradiction. By the instantiation of
  $\abs{w_i \fobs \vrt{\vsc}}$ we have
  $\rep{\wrfbB{f} \robs \scb{\vsc}{i}}$, which by
  \cref{ax:mem-robsuniq} gives us $\rep{\wrfbB{f} = \scr{\vsc}{i}}$,
  however that is impossible, since $\rep{\wrfbB{f}}$ is a rep event
  of $\vrt{f}$, while $\rep{\scr{\vsc}{i}}$ is a rep event of
  $\vrt{\vsc}$, thus a contradiction.
\end{proof}

We have now satisfied \cref{ax:fwd-scstruct,ax:fobshb}, so we satisfy
the preconditions to use \cref{lem:fwd-hbabsrep}. Before we prove
Properties~\eqref{ax:forward1} to \eqref{ax:forward2b}, we prove some
helper lemmas. First, we prove that if a virtual scan $\vrt{\vsc}$
observes a forwarding $\vrt{f}$, then $\vrt{f}$ must have observed
$\rep{\scon{\vsc}}$ both in $\rep{\wrx{w_i}}$ and $\rep{\wrfxB{f}}$,
meaning that the forward was concurrent to the scan.

\begin{lemma}
  \label{lem:jay2robsforward}
  If $\rep{\wrfbB{f} \robs \scb{\vsc}{i}}$ then
  $\rep{\scon{\vsc} \robs \wrx{w_i}}$ and
  $\rep{\scon{\vsc} \robs \wrfxB{f}}$.
\end{lemma}

\begin{proof}
  By $\rep{\wrfbB{f} \robs \scb{\vsc}{i}}$, we know that
  $\rep{\wrfbB{f}}$ successfully wrote to $\resBs[i]$, thus we know by
  \cref{ax:llsc-success} that there exists some $\rep{e_r}$ that wrote
  to $\resBs[i]$ such that $\rep{e_r \robs \wrfbA{f}}$ and
  $\rep{e_r \robs \wrfbB{f}}$. Also by the execution of $\vrt{f}$ we
  know that both $\rep{\wrx{w_i}}$ and $\rep{\wrfxB{f}}$ observed some
  $\rep{\scon{\vsc'}}$ of some $\vrt{\vsc'}$, i.e.,
  $\rep{\scon{\vsc'} \robs \wrx{w_i}}$ and
  $\rep{\scon{\vsc'} \robs \wrfxB{f}}$. To reach our goal, we prove by
  \cref{ax:fwd-sctotal} that we have $\vrt{\vsc = \vsc'}$ by showing
  that we can neither have $\vrt{\vsc \rb \vsc'}$ nor
  $\vrt{\vsc' \rb \vsc}$.
  \begin{itemize}
  \item For $\vrt{\vsc \rb \vsc'}$, we have
    $\rep{\scb{\vsc}{i} \rb \scr{\vsc'}{i}}$ by \cref{ax:rb-absrep}
    which leads to the cycle
    \[
      \rep{\scb{\vsc}{i} \rb \scr{\vsc'}{i} \rb \scon{\vsc'} \robs
        \wrfxB{f} \rb \wrfbB{f} \robs \scb{\vsc}{i}}
    \]
    which contradicts irreflexivity of $\rep{\hb}$ (\cref{lem:hb-irrefl}).
  \item For $\vrt{\vsc' \rb \vsc}$, we have
    $\rep{\scb{\vsc'}{i} \rb \scr{\vsc}{i}}$ by \cref{ax:rb-absrep}.
    By \cref{ax:mem2-wrtotal}, we have
    $\rep{\scr{\vsc}{i} \hb \wrfbB{f}}$ since if we instead had
    $\rep{\wrfbB{f} \hb \scr{\vsc}{i}}$ we have $\rep{\scr{\vsc}{i}}$
    occurring in between $\rep{\wrfbB{f} \robs \scb{\vsc}{i}}$,
    contradicting \cref{ax:mem2-nowrbetween}.
    Similarly by \cref{ax:mem2-wrtotal}, we have
    $\rep{e_r \hb \scr{\vsc}{i}}$, since consider if we have
    $\rep{\scr{\vsc}{i} \hbeq e_r}$: By \cref{ax:interval} over
    $\rep{\wrfbA{f} \rb \wrfxB{f}}$ and
    $\rep{\scoffobs{\vsc'} \rb \scb{\vsc'}{i}}$, we either have
    $\rep{\wrfbA{f} \rb \scb{\vsc'}{i}}$, which lets us construct the
    following cycle contradicting irreflexivity of $\rep{\hb}$
    (\cref{lem:hb-irrefl})
    \[
      \rep{\scr{\vsc}{i} \hbeq e_r \robs \wrfbA{f} \rb \scb{\vsc'}{i}
        \hb \scr{\vsc}{i}}
    \]
    or we have $\rep{\scoffobs{\vsc'} \rb \wrfxB{f}}$ which implies
    $\rep{\scoff{\vsc'} \hb \wrfxB{f}}$ by
    $\rep{\scoff{\vsc'} \hbeq \scoffobs{\vsc'} \rb \wrfxB{f}}$, thus
    $\rep{\scoff{\vsc'}}$ occurs in between
    $\rep{\scon{\vsc'} \robs \wrfxB{f}}$, contradicting
    \cref{ax:mem2-nowrbetween}.
    Thus we have $\rep{e_r \hb \scr{\vsc}{i}}$ and
    $\rep{\scr{\vsc}{i} \hb \wrfbB{f}}$, however this implies that
    $\rep{\scr{\vsc}{i}}$ occurs in between
    $\rep{e_r \robs \wrfbB{f}}$, contradicting
    \cref{ax:mem2-nowrbetween}. \qedhere
  \end{itemize}
\end{proof}

\begin{lemma}
  \label{lem:jay2rblinpoint}
  If $\abs{w'_i \rbhbeqrobs \vrt{\vsc}}$ then
  $\rep{\scoff{\vsc} \not\hb \wrx{w'_i}}$, and for any $\vrt{f'}$
  executed by $\abs{w'_i}$, we have
  $\rep{\scoff{\vsc} \not\hb \wrfxB{f'}}$ and if $\rep{\wrfbB{f'}}$
  was executed, then $\rep{\wrfbB{f'} \hb \scb{\vsc}{i}}$.
\end{lemma}

\begin{proof}
  For $\abs{w'_i \rbhbeqrobs \vrt{\vsc}}$, there exists some
  $\abs{e}$ and $\abs{w_j}$ such that
  $\abs{w'_i \rb e \hbeq w_j \robs \vrt{\vsc}}$. We know all of these
  events are populated except for $\vrt{\vsc}$ since they occur to the
  left of some $\abs{\hb}$. By \cref{lem:fwd-hbabsrep,ax:rb-absrep},
  we have every rep event of $\abs{w'_i}$ returns before some
  $\rep{e_r}$ in $\abs{e}$, and $\rep{e_r \hbeq \wra{w_j}}$. We case
  split $\abs{w_j \robs \vrt{\vsc}}$ into either
  $\rep{\wra{w_j} \robs \sca{\vsc}{j}}$, or
  $\rep{\wra{w_j} \robs \wrfa{f}}$ and
  $\rep{\wrfbB{f} \robs \scb{\vsc}{j}}$.
  \begin{itemize}
  \item If $\rep{\wra{w_j} \robs \sca{\vsc}{j}}$ then every rep event
    of $\abs{w'_i}$ happens before $\rep{\scoff{\vsc}}$ since we have
    \[
      \rep{e_r \hbeq \wra{w_j} \robs \sca{\vsc}{j} \rb \scoff{\vsc}}
    \]
    thus we have $\rep{\scoff{\vsc} \not\hb \wrx{w'_i}}$ and
    $\rep{\scoff{\vsc} \not\hb \wrfxB{f'}}$ and
    $\rep{\wrfbB{f'} \hb \scoff{\vsc} \hb \scb{\vsc}{i}}$ for all
    $\vrt{f'}$ executed by $\abs{w'_i}$.
  \item If $\rep{\wra{w_j} \robs \wrfa{f}}$ and
    $\rep{\wrfbB{f} \robs \scb{\vsc}{j}}$ then by
    \cref{lem:jay2robsforward} we have
    $\rep{\scon{\vsc} \robs \wrx{w_j}}$ and
    $\rep{\scon{\vsc} \robs \wrfxB{f}}$. By \cref{ax:fwd2-sconuniq} we
    have $\rep{\scoff{\vsc} \not\hb \wrfxB{f}}$, which in turn implies
    $\rep{\scoff{\vsc} \not\hb \wrx{w'_i}}$ and
    $\rep{\scoff{\vsc} \not\hb \wrfbB{f'}}$, since every rep event in
    $\abs{w'_i}$ occurs before $\rep{e_r}$ and
    $\rep{e_r \hbeq \wra{w_j} \robs \wrfa{f} \rb \wrfxB{f}}$. By
    \cref{ax:interval} over $\rep{\wrfbB{f'} \rb e_r}$ and
    $\rep{\scoffobs{\vsc} \rb \scb{\vsc}{i}}$, we have
    $\rep{\wrfbB{f'} \rb \scb{\vsc}{i}}$, since
    $\rep{\scoffobs{\vsc} \rb e_r}$ contradicts
    $\rep{\scoff{\vsc} \not\hb \wrfxB{f}}$ by
    \[
      \rep{\scoff{\vsc} \hbeq \scoffobs{\vsc} \rb e_r \hbeq \wra{w_j} \hb
        \wrfxB{f}}
    \]
    thus we satisfy all our goals. \qedhere
  \end{itemize}
\end{proof}

\begin{lemma}
  \label{lem:jay2fwd}
  If $\rep{\scon{\vsc} \robs \wrx{w_i}}$ with $\vrt{f}$ and $\vrt{f'}$
  executed by $\abs{w_i}$, where $\vrt{f \rb f'}$ and
  $\rep{\scon{\vsc} \robs \wrfxB{f}}$ and
  $\rep{\scon{\vsc} \robs \wrfxB{f'}}$ and
  $\rep{\wrfbB{f'} \hb \scb{\vsc}{i}}$, then there exists some
  $\vrt{f''}$ such that $\rep{\scr{\vsc}{i} \hb \wrfbB{f''}}$ and
  $\rep{\wrfbA{f''} \not\hb \wrfbA{f}}$ and
  $\rep{\wrfbB{f''} \hbeq \wrfbB{f'}}$ where $\rep{\wrfbB{f''}}$
  successfully wrote to $\resBs[i]$.
\end{lemma}

\begin{proof}
  Since we have $\rep{\scon{\vsc} \hb \wrx{w_i}}$ and
  $\rep{\wrfbB{f'} \hb \scb{\vsc}{i}}$, it is not possible for any
  $\rep{\scr{\vsc'}{i}}$ to occur in the intervals of $\vrt{f}$ and
  $\vrt{f}$, therefore we can apply \cref{lem:llsc-double}, thus by
  \cref{ax:fwd2-fbBuniq} there exists some $\vrt{f''}$ such that
  $\rep{\wrfbA{f''} \not\hb \wrfbA{f}}$ and
  $\rep{\wrfbB{f''} \hbeq \wrfbB{f'}}$ where $\rep{\wrfbB{f''}}$
  successfully wrote to $\resBs[i]$. By \cref{ax:interval} over
  $\rep{\wrx{w_i} \rb \wrfbA{f}}$ and
  $\rep{\wrfbA{f''} \rb \wrfbB{f''}}$, since
  $\rep{\wrfbA{f''} \not\hb \wrfbA{f}}$ contradicts
  $\rep{\wrfbA{f''} \rb \wrfbA{f}}$, we have
  $\rep{\wrx{w_i} \rb \wrfbB{f''}}$, which lets us construct
  $\rep{\scr{\vsc}{i} \hb \scon{\vsc} \robs \wrx{w_i} \rb
    \wrfbB{f''}}$ implying $\rep{\scr{\vsc}{i} \hb \wrfbB{f''}}$, thus
  we satisfy all the goals.
\end{proof}

\begin{genthm}{\cref{ax:forward1}}
  If $\abs{w'_i \rbhbeqrobs \vrt{\vsc}}$ and
  $\rep{\scr{\vsc}{i} \robs \scb{\vsc}{i}}$ then
  $\rep{\wra{w'_i} \hb \sca{\vsc}{i}}$.
\end{genthm}

\begin{proof}
  By \cref{lem:jay2rblinpoint} we have
  $\rep{\scoff{\vsc} \not\hb \wrx{w'_i}}$ and
  $\rep{\scoff{\vsc} \not\hb \wrfxB{f'}}$ and
  $\rep{\scoff{\vsc} \not\hb \wrfxB{f''}}$ and
  $\rep{\wrfbB{f''} \hb \scb{\vsc}{i}}$ where $\vrt{f'}$ and
  $\vrt{f''}$ where executed by $\abs{w'_i}$. Consider if we have
  $\rep{\scon{\vsc} \hb \wrx{w'_i}}$ (and thus also
  $\rep{\scon{\vsc} \hb \wrfxB{f'}}$ and
  $\rep{\scon{\vsc} \hb \wrfxB{f''}}$), then by
  \cref{ax:fwd2-sconuniq} we have $\rep{\scon{\vsc} \robs \wrx{w'_i}}$
  and $\rep{\scon{\vsc} \robs \wrfxB{f'}}$ and
  $\rep{\scon{\vsc} \robs \wrfxB{f''}}$. By \cref{lem:jay2fwd} we know
  there exists some $\vrt{f}$ such that
  $\rep{\scr{\vsc}{i} \hb \wrfbB{f}}$ and
  $\rep{\wrfbB{f} \hbeq \wrfbB{f''}}$ with $\rep{\wrfbB{f}}$
  successfully writing to $\resBs[i]$. However, this means we have
  $\rep{\scr{\vsc}{i} \hb \wrfbB{f} \hbeq \wrfbB{f''} \hb
    \scb{\vsc}{i}}$, contradicting \cref{ax:mem-nowrbetween} by
  $\rep{\wrfbB{f}}$ occurring in between
  $\rep{\scr{\vsc}{i} \robs \scb{\vsc}{i}}$.
  
  Thus we must have $\rep{\scon{\vsc} \not\hb \wrx{w'_i}}$. By
  \cref{ax:interval} over $\rep{\wra{w'_i} \rb \wrx{w'_i}}$ and
  $\rep{\sconobs{\vsc} \rb \sca{\vsc}{i}}$, since having
  $\rep{\scon{\vsc} \not\hb \wrx{w'_i}}$ contradicts
  $\rep{\scon{\vsc} \hbeq \sconobs{\vsc} \rb \wrx{w'_i}}$, we must have
  $\rep{\wra{w'_i} \rb \sca{\vsc}{i}}$, which is our goal, thus we are
  done.
\end{proof}

\begin{genthm}{\cref{ax:forward2a,ax:forward2b}}
  If $\abs{w_i \fobs \vrt{\vsc}}$ and either
  $\abs{w'_i \rbhbeqrobs \vrt{\vsc}}$ or
  $\rep{\wra{w'_i} \hb \scr{\vsc}{i}}$ then $\abs{w_i \not\hb w'_i}$.
\end{genthm}

\begin{proof}
  Assume we have $\abs{w_i \hb w'_i}$, we derive a contradiction. By
  \cref{lem:fwd-hbabsrep} we have $\rep{\wra{w_i} \hb \wra{w'_i}}$
  from $\abs{w_i \hb w'_i}$, since both writes are populated by
  occurring to the left of some $\abs{\hb}$. By the instantiation of
  $\abs{\fobs}$, there exists some $\vrt{f}$ such that
  $\rep{\wra{w_i} \robs \wrfa{f}}$ and
  $\rep{\wrfbB{f} \robs \scb{\vsc}{i}}$. By
  \cref{lem:jay2robsforward}, we have
  $\rep{\scon{\vsc} \robs \wrx{w_i}}$.
  
  We start with the $\rep{\wra{w'_i} \hb \scr{\vsc}{i}}$ case, where
  we can derive that $\rep{\wra{w'_i} \hb \wrfa{f}}$:
  \[
    \rep{\wra{w'_i} \hb \scr{\vsc}{i} \rb \scon{\vsc} \robs \wrx{w_i}
      \rb \wrfa{f}}
  \]
  which with $\rep{\wra{w_i} \hb \wra{w'_i}}$, contradicts
  \cref{ax:mem-nowrbetween} by $\rep{\wra{w'_i}}$ occurring in between
  $\rep{\wra{w_i} \robs \wrfa{f}}$.

  For the $\abs{w'_i \rbhbeqrobs \vrt{\vsc}}$ case, we start by
  deriving $\rep{\wrfbA{f} \rb \wrx{w'_i}}$ by \cref{ax:interval} over
  $\rep{\wrfbA{f} \rb \wrfa{f}}$ and
  $\rep{\wra{w'_i} \rb \wrx{w'_i}}$, since we cannot have
  $\rep{\wra{w'_i} \rb \wrfa{f}}$ since that contradicts
  \cref{ax:mem-nowrbetween} by $\rep{\wra{w'_i}}$ occurring in between
  $\rep{\wra{w_i} \robs \wrfa{f}}$. From this, we can derive that we
  have $\rep{\scon{\vsc} \hb \wrx{w'_i}}$ by
  $\rep{\scon{\vsc} \robs \wrx{w_i} \rb \wrfbA{f} \rb \wrx{w'_i}}$. By
  \cref{lem:jay2rblinpoint} we have
  $\rep{\scoff{\vsc} \not\hb \wrx{w'_i}}$, which by
  \cref{ax:fwd2-sconuniq} means we have
  $\rep{\scon{\vsc} \robs \wrx{w'_i}}$, thus we know that
  $\vrt{\wrfA{w'_i}}$ and $\vrt{\wrfB{w'_i}}$ will be executed, let us
  refer to these as $\vrt{f'}$ and $\vrt{f''}$ respectively. Thus
  again by \cref{lem:jay2rblinpoint,ax:fwd2-sconuniq} we have
  $\rep{\scon{\vsc} \robs \wrx{w'_i}}$ and
  $\rep{\scon{\vsc} \robs \wrfxB{f'}}$ and
  $\rep{\scon{\vsc} \robs \wrfxB{f''}}$.

  By \cref{lem:jay2fwd} over $\vrt{f'}$ and $\vrt{f''}$, there exists
  some $\vrt{f'''}$ such that $\rep{\wrfbA{f'''} \not\hb \wrfbA{f'}}$
  and $\rep{\wrfbB{f'''} \hbeq \wrfbB{f''}}$. By
  \cref{ax:mem2-wrtotal}, we either have
  $\rep{\wrfbB{f} \hb \wrfbB{f'''}}$ or
  $\rep{\wrfbB{f'''} \hbeq \wrfbB{f}}$. In the first case, we have
  \[
    \rep{\wrfbB{f} \hb \wrfbB{f'''} \hbeq \wrfbB{f''} \hb
      \scb{\vsc}{i}}
  \]
  which contradicts \cref{ax:mem2-nowrbetween} by $\rep{\wrfbB{f'''}}$
  occurring in between $\rep{\wrfbB{f} \robs \scb{\vsc}{i}}$. In the
  second case, if we have $\rep{\wrfbB{f'''} \hb \wrfbB{f}}$ then by
  \cref{lem:llsc-seq}, and $\rep{\wrfbB{f}}$ and $\rep{\wrfbB{f'''}}$
  being successful, we have $\rep{\wrfbB{f'''} \hb \wrfbA{f}}$, and
  this also follows in the case of $\rep{\wrfbB{f'''} = \wrfbB{f}}$.
  By \cref{ax:interval} over $\rep{\wra{w'_i} \rb \wrfbA{f'}}$ and
  $\rep{\wrfbA{f'''} \rb \wrfbB{f'''}}$ we have
  $\rep{\wra{w'_i} \rb \wrfbB{f'''}}$ since we cannot have
  $\rep{\wrfbA{f'''} \rb \wrfbA{f'}}$ by
  $\rep{\wrfbA{f'''} \not\hb \wrfbA{f'}}$. Thus we have
  \[
    \rep{\wra{w_i} \hb \wra{w'_i} \rb \wrfbB{f'''} \hb \wrfbA{f} \rb \wrfa{f}}
  \]
  which contradicts \cref{ax:mem-nowrbetween} by $\rep{\wra{w'_i}}$
  occurring in between $\rep{\wra{w_i} \robs \wrfa{f}}$.
\end{proof}

\section{Algorithms \ref{alg:jay2} and \ref{alg:jay3} Satisfy the
  Multi-Writer Forwarding Signature}\label{apx:jay3}

We will show the complete proofs that \cref{alg:jay2,alg:jay3} satisfy
the multi-writer forwarding signature (\cref{fig:forward2}). We start
by repeating the proof that \cref{alg:jay2} satisfies the signature.

\begin{genthm}{\cref{lem:jay2-fwd2}}
  Every execution of \cref{alg:jay2} satisfies the $\mwFwdSig$ signature
  (\cref{fig:forward2}).
\end{genthm}
\begin{proof}
  Since this algorithm is single scanner, we can simply define the set
  of virtual scans to be the same as the set of abs scans and map each
  abs scan to itself, i.e., $\vrt{\Vsc} = \abs{S}$ and
  $\SMap{\abs{s}} = \abs{s}$. Each rep event directly corresponds to
  their equivalent variant in \cref{alg:jay2}, except for
  $\rep{\sconobs{\vsc}}$ and $\rep{\scoffobs{\vsc}}$ which are set as
  $\rep{\sconobs{\vsc}} = \rep{\scon{\vsc}} = \rep{\scon{s}}$ and
  $\rep{\scoffobs{\vsc}} = \rep{\scoff{\vsc}} = \rep{\scoff{s}}$
  respectively. \cref{ax:fwd-vrtinscan} holds by each event being a
  subevent of itself ($\abs{e \subev e}$) and \cref{ax:fwd-sctotal} holds
  since the algorithm is single-scanner. Properties~\eqref{ax:fwd-io}
  to \eqref{ax:fwd2-fbBuniq} and \eqref{ax:fwd2-scstruct} to
  \eqref{ax:fwd2-fwdsccond} holds directly by the structure of the
  algorithm and \cref{ax:fwd2-sconuniq} can be proven by
  Lemma~\apndxLemJayIObsX{}, the same lemma that established this property
  for \cref{alg:jay1}, since the relative structure of
  $\rep{\scon{s}}$ and $\rep{\scoff{s}}$ is the same.
\end{proof}

We next prove that \cref{alg:jay3} satisfies signature. We first prove
that the phases of logical scans can only happen in sequence.

\begin{lemma}\label{lem:jay3xord}
  Each $\rep{\scvon}$ and $\rep{\scvoff}$ and $\rep{\scvend}$ that
  successfully wrote to $\resX$ are ordered such that there exists
  either $\rep{\scvon'}$ or $\rep{\scvoff'}$ or $\rep{\scvend'}$,
  where we have either $\rep{\scvon' \robs \scvoff}$ or
  $\rep{\scvoff' \robs \scvend}$ or $\rep{\scvend' \robs \scvon}$,
  with the exception of the first $\rep{\scvon}$, which observes the
  initial value of $\resX$.
\end{lemma}

\begin{proof}
  Looking at \cref{alg:jay3}, we see that the only operations writing
  to $\resX$ are successful $\SCins$ operations, which are only events
  $\rep{\scvon}$ and $\rep{\scvoff}$ and $\rep{\scvend}$. For either
  of these events to be successful, they need to observe the same
  write as the latest $\rep{\scvx}$ observed as per
  \cref{ax:llsc-success}.

  For $\rep{\scvon}$ to be successful, we must have
  $\rep{e_r \robs \scvx}$ and $\rep{e_r \robs \scvon}$. For the
  algorithm to execute $\rep{\scvon}$, the write $\rep{e_r}$ must have
  written 1 to $\xPhase$, which can only be either the original
  initialization of $\resX$ or $\rep{\scvend}$. Similarly, for
  $\rep{\scvoff}$, it can only be executed if the write $\rep{e_r}$
  wrote 2 to $\xPhase$, thus $\rep{e_r}$ can only be $\rep{\scvon}$,
  and for $\rep{\scvend}$, it can only be executed if the write
  $\rep{e_r}$ wrote 3 to $\xPhase$, thus $\rep{e_r}$ can only be
  $\rep{\scvoff}$.
\end{proof}

  We instantiate the write and forward rep events for the event
  structures of \cref{fig:struct2} with rep events of the same name.
  The structure of \writeProc{} and \forwardProc{} ensures that
  Properties~\eqref{ax:fwd-wrauniq}, \eqref{ax:fwd2-fbBuniq} and
  \eqref{ax:fwd2-wrstruct} to \eqref{ax:fwd2-fwdsccond} are satisfied.

  We instantiate virtual scans as a logical object corresponding to
  the rep events executed by \pushLSProc{} constructing a complete
  logical scan. The interval of a virtual scan $\vrt{\vsc}$ is
  instantiated as the smallest interval containing all its rep events.
  More formally, if we have $\vrt{\scv}$, $\vrt{\scv'}$ and
  $\vrt{\scv''}$ such that $\rep{\scvon \robs \scvx'}$ and
  $\rep{\scvoff' \robs \scvx''}$ and $\rep{\scvss''}$ wrote to
  $\resSS$, then there exists a virtual scan $\vrt{\vsc}$. Its rep
  events are instantiated using the rep events of $\vrt{\scv}$,
  $\vrt{\scv'}$ and $\vrt{\scv''}$ as follows and instantiate
  $\SMap{\abs{s}} = \vrt{\vsc}$ iff $\rep{\scvss'' \robs \scx}$.
  \begin{align*}
    \rep{\scr{\vsc}{i}} &\defeq \rep{\scvr_i}
    &\rep{\scon{\vsc}} &\defeq \rep{\scvon}
    &\rep{\sconobs{\vsc}} &\defeq \rep{\scvx'}
    &\rep{\sca{\vsc}{i}} &\defeq \rep{\scva'_i}
    &\rep{\scoff{\vsc}} &\defeq \rep{\scvoff'}
    &\rep{\scoffobs{\vsc}} &\defeq \rep{\scvx''}
    &\rep{\scb{\vsc}{i}} &\defeq \rep{\scvb''_i}
  \end{align*}
  We also let $\rep{\scxinit{\vsc} \defeq \rep{\scvx}}$ and
  $\rep{\scss{\vsc} \defeq \rep{\scvss''}}$ and
  $\rep{\scend{\vsc} \defeq \rep{\scvend'''}}$ where
  $\rep{\scvend'''}$ was successful and
  $\rep{\scvoff' \robs \scvend'''}$, however these do not need to be
  in the interval of $\vrt{\vsc}$. We can view the events as forming
  the following structure.
  \[\arraycolsep=0pt
    \begin{array}[c]{ccccccc}
      \multicolumn{2}{c}{\vrt{\scv}}&&\multicolumn{1}{c}{\vrt{\scv'}}&&\multicolumn{2}{c}{\vrt{\scv''}}\\[-5pt]
      \multicolumn{2}{c}{${\downbracefill}$}&&\multicolumn{1}{c}{${\downbracefill}$}&&\multicolumn{2}{c}{${\downbracefill}$}\\
      \rep{\scvx \rb {}} & \rep{\scvr_0 \cdots \scvr_{n-1} \rb \scvon} &
    \rep{{} \robs {}} &
    \rep{\scvx' \rb \scva'_0 \cdots \scva'_{n-1} \rb \scvoff'} &
    \rep{{} \robs {}} & \rep{\scvx'' \rb \scvb''_0 \cdots \scvb''_{n-1}} &
    \rep{{} \rb \scvss''}\\
    & \multicolumn{5}{c}{${\upbracefill}$}\\[-2pt]
    & \multicolumn{5}{c}{\vrt{\vsc}}
    \end{array}
  \]

  \cref{ax:fwd-io,ax:fwd-scruniq,ax:fwd2-sconuniq,ax:fwd2-scstruct}
  follows from the structure of virtual scans and the algorithm.

\begin{lemma}\label{lem:jay3scssend}
  For every virtual scan $\vrt{\vsc}$ such that $\rep{\scend{\vsc}}$
  have written to $\resX$, there exists exactly one
  $\rep{\scss{\vsc}}$ such that
  $\rep{\scoff{\vsc} \hb \scss{\vsc} \hb \scend{\vsc}}$.
\end{lemma}

\begin{proof}
  By induction on the number virtual scans performed. The base case
  corresponds to where $\vrt{\vsc}$ is the very first virtual scan. We
  simultaneously consider the inductive step case, where the inductive
  hypothesis states that all virtual scans
  $\vrt{\vsc'} \in \vrt{\Vsc'}$ has exactly one $\rep{\scss{\vsc'}}$
  such that $\rep{\scoff{\vsc'} \hb \scss{\vsc'} \hb \scend{\vsc'}}$,
  the inductive step follows that for $\vrt{\vsc}$ proceeding the
  virtual scans of $\vrt{\Vsc'}$, there exists exactly one
  $\rep{\scss{\vsc}}$ such that
  $\rep{\scoff{\vsc} \hb \scss{\vsc} \hb \scend{\vsc}}$.

  For the base case, since no other virtual scans has finished prior,
  we must have $\xToggle = \ssToggle$ before $\rep{\scoffobs{\vsc}}$
  has been executed by initialization of $\resX$ and $\resSS$. For the
  inductive step, we have $\xToggle = \ssToggle$ before
  $\rep{\scoffobs{\vsc}}$ executes, since it follows from each
  $\vrt{\vsc'} \in \vrt{\Vsc'}$ having
  $\rep{\scss{\vsc'} \hb \scend{\vsc'}}$ and there only existing one
  $\rep{\scss{\vsc'}}$ per $\vrt{\vsc'}$, since after
  $\rep{\scend{\vsc'}}$ has executed, $\xToggle = \ssToggle$ by
  $\rep{\scss{\vsc'}}$ negating $\ssToggle$ and $\rep{\scend{\vsc'}}$
  negating $\xToggle$, both exactly once.

  With $\rep{\scend{\vsc}}$ successfully writing, we know that the
  code of phase 3 (lines 31-43 of \cref{alg:jay3}) must have been
  executed. Let $\rep{\scend{\vsc}} = \rep{\scvend}$, we also know
  that $\rep{\scvxB}$ returned $\true$, since otherwise
  $\rep{\scend{\vsc}}$ must have failed to write. If the value of
  $\ssToggle$ read by $\rep{\scvssA}$ is different from the value of
  $\xToggle$ read by $\rep{\scoffobs{\vsc}}$, then there must exist
  some $\rep{\scvss'}$ that negated $\ssToggle$ before
  $\rep{\scvssA}$ (and thus before $\rep{\scend{\vsc}}$), but after
  $\rep{\scoff{\vsc}}$, which fits our goal.

  Otherwise, if $\ssToggle = \xToggle$, then $\rep{\scvss}$ will be
  executed. If $\rep{\scvss}$ successfully wrote, then we trivially
  satisfy our goal. Otherwise, by \cref{lem:llsc-interference}, there
  has to exist some $\rep{\scvss'}$ such that
  $\rep{\scvss' \not\hb \scvssA}$ and $\rep{\scvss' \hb \scvss}$.
  Since we must have $\rep{\scvoff'' \robs \scvx' \rb \scvss'}$ for
  some $\rep{\scvoff''}$, consider if we have
  $\rep{\scoff{\vsc} = \scvoff''}$, then we have
  $\rep{\scoff{\vsc} \hb \scvss' \hb \scend{\vsc}}$, satisfying our goal.
  Consider instead we have $\rep{\scoff{\vsc} \neq \scvoff''}$, by
  \cref{ax:mem2-wrtotal}, we either have
  $\rep{\scoff{\vsc} \hb \scvoff''}$ or
  $\rep{\scvoff'' \hb \scoff{\vsc}}$, where if we have the former then
  we must have
  $\rep{\scoff{\vsc} \hb \scend{\vsc} \hb \scvoff'' \hb \scvss'}$ by
  \cref{lem:jay3xord}, inducing the cycle
  $\rep{\scvss \rb \scend{\vsc} \hb \scvss' \hb \scvss}$ contradicting
  irreflexivity of $\rep{\hb}$, thus we must have the latter. By
  $\rep{\scvoff'' \hb \scoff{\vsc}}$, it must be the case that
  $\rep{\scvss'}$ correspond to some earlier executed virtual scan,
  which in the base case contradicts $\vrt{\vsc}$ being the first
  virtual scan, while in the inductive step $\rep{\scvoff''}$
  corresponds to some $\vrt{\vsc'} \in \vrt{\Vsc'}$, it must be the
  case that
  $\rep{\scvss' \hb \scend{\vsc'} \hb \scoff{\vsc} \hb \scvssA}$ by
  the inductive hypothesis and \cref{lem:jay3xord}, which contradicts
  $\rep{\scvss' \not\hb \scvssA}$.
\end{proof}

\newcommand{\PstrEv}{\mathit{se}}
\newcommand{\PendEv}{\mathit{ee}}
\newcommand{\PadvEv}{\mathit{ae}}

To efficiently reason about \pushLSProc{}, we define the intervals of
the phases. For an execution $\vrt{\scv}$ of \pushLSProc{},
$\vrt{\scv.p_1}$, and $\vrt{\scv.p_2}$, and $\vrt{\scv.p_3}$ are the
intervals of phases 1, 2, and 3. Each phase interval $\vrt{p}$ of
$\vrt{\scv}$ has a start event $\vrt{p}.\rep{\PstrEv}$ which is some
$\rep{\scvx}$ and an end event $\vrt{p}.\rep{\PendEv}$ which is a
$\SCins$ event of $\resX$, where start and end events need to satisfy
$\rep{\vrt{p}.\rep{\PstrEv} \llobs \vrt{p}.\rep{\PendEv}}$.
For $\vrt{p_1}$ we have $\vrt{p_1}.\rep{\PendEv} = \rep{\scvon}$ with
either $\vrt{p_1}.\rep{\PstrEv}$ observing the initial value of
$\resX$ or $\rep{\scend{\vsc} \robs \vrt{p_1}.\rep{\PstrEv}}$; for
$\vrt{p_2}$ we have $\vrt{p_2}.\rep{\PendEv} = \rep{\scvoff}$ with
$\rep{\scon{\vsc} \robs \vrt{p_2}.\rep{\PstrEv}}$; and for $\vrt{p_3}$
we have $\vrt{p_3}.\rep{\PendEv} = \rep{\scvend}$ with
$\rep{\scoff{\vsc} \robs \vrt{p_3}.\rep{\PstrEv}}$. Each phase
$\vrt{p}$ has an \emph{advancement event}, which we encode with
$\vrt{p}.\rep{\PadvEv}$, which is the rep event successfully writing
to $\resX$ which advances the logical scanner to the next phase. For
some $\vrt{\vsc}$, we have
$\vrt{p_1}.\rep{\PadvEv} = \rep{\scon{\vsc}}$ and
$\vrt{p_2}.\rep{\PadvEv} = \rep{\scoff{\vsc}}$ and
$\vrt{p_3}.\rep{\PadvEv} = \rep{\scend{\vsc}}$.

\begin{lemma}\label{lem:jay3phaseev}
  For each phase $\vrt{p}$, there exists a $\vrt{p}.\rep{\PadvEv}$
  such that
  $\rep{\vrt{p}.\rep{\PadvEv} \not\hb \vrt{p}.\rep{\PstrEv}}$ and
  $\rep{\vrt{p}.\rep{\PadvEv} \hbeq \vrt{p}.\rep{\PendEv}}$.
\end{lemma}

\begin{proof}
  By \cref{lem:llsc-interference}, there exists some successful
  writing event $\rep{e_r}$ of $\resX$ such that
  $\rep{e_r \not\hb \vrt{p}.\rep{\PstrEv}}$ and
  $\rep{e_r \hbeq \vrt{p}.\rep{\PendEv}}$. If
  $\rep{e_r} = \vrt{p}.\rep{\PadvEv}$, we are done, otherwise by
  \cref{lem:jay3xord} we know that $\rep{\scon{\vsc}}$ and
  $\rep{\scoff{\vsc}}$ and $\rep{\scend{\vsc}}$ occur in order, so for
  example in the case of $\vrt{p} = \vrt{p_2}$, we have
  $\rep{\scon{\vsc} \robs \vrt{p_2}.\rep{\PstrEv}}$, by
  \cref{ax:mem2-wrtotal} we must have $\rep{\scon{\vsc} \hb e_r}$,
  since $\rep{e_r \hb \scon{\vsc}}$ contradicts
  $\rep{e_r \not\hb \vrt{p}.\rep{\PstrEv}}$. Thus by
  \cref{lem:jay3xord}, there exists some
  $\rep{\scoff{\vsc} = \vrt{p_2}.\rep{\PadvEv}}$ in between
  $\rep{\scon{\vsc}}$ and $\rep{e_r}$, which satisfies
  $\rep{\scoff{\vsc} \hbeq \vrt{p}.\rep{\PendEv}}$ by
  $\rep{\scoff{\vsc} \hb e_r \hbeq \vrt{p}.\rep{\PendEv}}$ and
  $\rep{\scoff{\vsc} \not\hb \vrt{p}.\rep{\PstrEv}}$ since
  $\rep{\scoff{\vsc} \hb \vrt{p}.\rep{\PstrEv}}$ contradicts
  \cref{ax:mem2-nowrbetween} by $\rep{\scoff{\vsc}}$ occurring in
  between $\rep{\scon{\vsc} \robs \vrt{p_2}.\rep{\PstrEv}}$. The other
  phases follow by similar arguments.
\end{proof}

\begin{lemma}\label{lem:jay3pushLSend}
  For any terminated $\pushLSProc$ execution $\vrt{\scv}$, there
  exists some $\vrt{\vsc}$ such that for $\rep{\scvx}$ being the first
  rep event of $\vrt{\scv}$ and $\rep{e_r}$ being the last, then and
  $\rep{\scend{\vsc} \not\hb \rep{\scvx}}$ and
  $\rep{\scend{\vsc} \hbeq \rep{e_r}}$.
\end{lemma}

\begin{proof}
  Consider whether $\vrt{p_3}$ was executed or not by $\vrt{\scv}$. If
  $\vrt{p_3}$ was executed, then by \cref{lem:jay3phaseev} we satisfy
  our goal. Otherwise, if it was not executed, this means either
  $\vrt{p_1}$ or $\vrt{p_2}$ were executed, i.e., there is some
  $\rep{\scend{\vsc}}$ or $\rep{\scon{s}}$ setting $\xPhase$, thus
  there is another chance for $\vrt{p_3}$ to run. If again it was not
  run, we are either in $\vrt{p_1}$ or $\vrt{p_2}$ by
  $\rep{\scend{\vsc'}/\scon{\vsc'}}$, by \cref{lem:jay3phaseev} the
  execution of virtual scans must have progressed, and by
  \cref{lem:jay3xord} there must be some $\rep{\scend{\vsc''}}$
  between $\rep{\scend{\vsc}/\scon{\vsc}}$ and
  $\rep{\scend{\vsc'}/\scon{\vsc'}}$, where $\rep{\scend{\vsc''}}$ may
  be equal to $\rep{\scend{\vsc'}}$. Thus $\rep{\scend{\vsc''}}$
  satisfies our goal.
\end{proof}

\begin{lemma}\label{lem:jay3pushLSend2}
  For any terminated $\pushLSProc$ executions $\vrt{\scv}$ and
  $\vrt{\scv'}$ where $\vrt{\scv \rb \scv'}$, there exists some
  $\vrt{\vsc}$ and $\vrt{\vsc'}$ such that for $\rep{\scvx}$ being the
  first rep event of $\vrt{\scv}$, then
  $\rep{\scend{\vsc} \not\hb \rep{\scvx}}$ and
  $\rep{\scend{\vsc} \hb \scend{\vsc'}}$.
\end{lemma}

\begin{proof}
  By \cref{lem:jay3pushLSend} over $\vrt{\scv}$ and $\vrt{\scv'}$, we
  have $\rep{\scend{\vsc} \not\hb \rep{\scvx}}$ and
  $\rep{\scend{\vsc} \hbeq \rep{e_r}}$ and
  $\rep{\scend{\vsc'} \not\hb \rep{\scvx'}}$ and
  $\rep{\scend{\vsc'} \hbeq \rep{e'_r}}$ where $\rep{e_r}$ and
  $\rep{e'_r}$ are the last events of $\vrt{\scv}$ and $\vrt{\scv'}$
  respectively. By \cref{ax:mem2-wrtotal}, we must have
  $\rep{\scend{\vsc} \hb \scend{\vsc'}}$, since if we instead had
  $\rep{\scend{\vsc'} \hbeq \scend{\vsc}}$, we contradict
  $\rep{\scend{\vsc'} \not\hb \rep{\scvx'}}$ by having
  $ \rep{\scend{\vsc'} \hbeq \scend{\vsc} \hbeq \rep{e_r} \rb
    \rep{\scvx'}} $.
\end{proof}

\begin{genthm}{\cref{ax:fwd-vrtinscan}}
  For any abs scan $\abs{s}$, we have
  $\SMap{\abs{s}} \subev \abs{s}$.
\end{genthm}

\begin{proof}
  With \cref{lem:jay3pushLSend2}, we can show for $\abs{s}$ that there
  exists some $\rep{\scend{\vsc'} \hb \scend{\vsc''}}$ in between
  $\vrt{\scvA}$ and $\vrt{\scvB}$ such that
  $\rep{\scend{\vsc'} \not\hb \rep{\scvx}}$ where $\rep{\scvx}$ is the
  first rep event of $\vrt{\scvA}$. For $\vrt{\vsc} = \SMap{\abs{s}}$,
  we have $\rep{\scend{\vsc'} \hb \scxinit{\vsc}}$ since some new
  logical scan must have started after $\rep{\scend{\vsc'}}$ by
  $\rep{\scend{\vsc''}}$ occurring after it, and by the instantiation
  of $\SMap{[-]}$ we have $\rep{\scss{\vsc} \robs \scx}$, thus for all
  $i$ we have
  $ \rep{\scend{\vsc'} \hb \scxinit{\vsc} \rb \scr{\vsc}{i} \rb
    \scb{\vsc}{i} \rb \scss{\vsc} \robs \scx} $.
  Our goal is to establish
  $\abs{s}.\evStart \leq \rep{\scr{\vsc}{i}}.\evStart$ and
  $\rep{\scb{\vsc}{i}}.\evEnd \leq \abs{s}.\evEnd$.
  We show this by showing that we neither have
  $\abs{s}.\evStart > \rep{\scr{\vsc}{i}}.\evStart$ nor
  $\rep{\scb{\vsc}{i}}.\evEnd > \abs{s}.\evEnd$.
  \begin{itemize}
  \item Assume we have
    $\abs{s}.\evStart > \rep{\scr{\vsc}{i}}.\evStart$, we derive a
    contradiction. By definition of
    $\rep{\scxinit{\vsc} \rb \scr{\vsc}{i}}$, we have
    $\rep{\scxinit{\vsc}}.\evEnd < \rep{\scr{\vsc}{i}}.\evStart$,
    which implies $\rep{\scxinit{\vsc}}.\evEnd < \abs{s}.\evStart$ and
    thus $\rep{\scb{\vsc}{i}} \rb \abs{s}$. This in turn implies
    $\rep{\scb{\vsc}{i} \rb \scvx}$ since $\rep{\scvx} \subev \abs{s}$,
    which contradicts $\rep{\scend{\vsc'} \not\hb \rep{\scvx}}$ by
    $\rep{\scend{\vsc'} \hb \scxinit{\vsc} \rb \scvx}$.
  \item Assume we have $\rep{\scb{\vsc}{i}}.\evEnd > \abs{s}.\evEnd$,
    we derive a contradiction. By definition of
    $\rep{\scb{\vsc}{i} \rb \scss{\vsc}}$, we have
    $\rep{\scb{\vsc}{i}}.\evEnd < \rep{\scss{\vsc}}.\evStart$, which
    implies $\abs{s}.\evEnd < \rep{\scss{\vsc}}.\evStart$ and thus
    $\abs{s} \rb \rep{\scss{\vsc}}$. This in turn implies
    $\rep{\scx \rb \scss{\vsc}}$ since $\rep{\scx} \subev \abs{s}$, which
    contradict irreflexivity of $\rep{\hb}$ by
    $\rep{\scx \rb \scss{\vsc} \robs \scx}$. \qedhere
  \end{itemize}
\end{proof}

\begin{genthm}{\cref{ax:fwd2-sconuniq}}
  For any $\rep{e_r}$ reading $\resX$, we have $\rep{\scon{s} \robs e_r}$
  iff we have $\rep{\scon{s} \hb e_r}$ and $\rep{\scoff{s} \not\hb e_r}$.
\end{genthm}

\begin{proof}\leavevmode
  \begin{itemize}
  \item[$(\mathord{\implies})$] $\rep{\scon{s} \robs \wrx{w_i}}$ trivially
    implies $\rep{\scon{s} \hb \wrx{w_i}}$, and if we were to have
    $\rep{\scoff{s} \hb \wrx{w_i}}$ we contradict \cref{ax:mem-nowrbetween}
    since $\rep{\scoff{s}}$ would happen in-between $\rep{\scon{s}}$ and
    $\rep{e_r}$.
  \item[$(\mathord{\impliedby})$] By the fact that only virtual scans
    change $\resX$ and by \cref{lem:jay3xord} forcing these changes to
    be of an order of form
    $\rep{\scon{s} \robs \scoff{s} \robs \scend{\vsc} \robs \scon{s'} \dots}\ $, we
    know that any write to $\resX$ other than $\rep{\scon{s}}$ and
    $\rep{\scoff{s}}$ either happens before $\rep{\scon{s}}$ or after
    $\rep{\scoff{s}}$. Since $\rep{\wrx{w_i}}$ has to observe something when
    it finishes, it must observe $\rep{\scon{s}}$, since any other
    observation leads to a contradiction. \qedhere
  \end{itemize}

\end{proof}

\begin{genthm}{\cref{ax:fwd-sctotal}}
  For two distinct virtual scans $\vrt{\vsc}$ and $\vrt{\vsc'}$, we either
  have $\vrt{\vsc \rb \vsc'}$ or $\vrt{\vsc' \rb \vsc}$.
\end{genthm}

\begin{proof}
  By the instantiation of $\vrt{\vsc}$ and \cref{lem:jay3xord}, we
  have that each virtual scan must be separated by some
  $\rep{\scend{\vsc}}$, i.e., we either have
  $\rep{\scend{\vsc} \hb \scxinit{\vsc'}}$ or
  $\rep{\scend{\vsc'} \hb \scxinit{\vsc}}$. By \cref{lem:jay3scssend},
  we know that for each $\vrt{\vsc}$ there is exactly one
  $\rep{\scss{\vsc}}$ occurring before $\rep{\scend{\vsc}}$, thus we
  have either $\rep{\scss{\vsc} \hb \scend{\vsc} \hb \scxinit{\vsc'}}$
  or $\rep{\scss{\vsc'} \hb \scend{\vsc'} \hb \scxinit{\vsc}}$. If we
  have $\rep{\scss{\vsc} \hb \scend{\vsc} \hb \scxinit{\vsc'}}$ then
  we have
  $\rep{\scb{\vsc}{i} \rb \scss{\vsc} \hb \scend{\vsc} \hb
    \scxinit{\vsc'} \rb \scr{\vsc'}{i}}$, and dually for
  $\rep{\scss{\vsc'} \hb \scend{\vsc'} \hb \scxinit{\vsc}}$, thus we
  either have $\rep{\scb{\vsc}{i} \rb \scr{\vsc'}{i}}$ or
  $\rep{\scb{\vsc'}{i} \rb \scr{\vsc}{i}}$ for each $i$, which matches
  the defined bounds of $\vrt{\vsc}$, i.e., we either have
  $\vrt{\vsc \rb \vsc'}$ or $\vrt{\vsc' \rb \vsc}$, which is our goal.
\end{proof}

\section{Algorithm \ref{alg:afek} Satisfies the Snapshot Signature}\label{apx:afek}

We show that \cref{alg:afek} satisfies the $\snapshotSig$ signatures
(\cref{fig:snapshot}). We employ virtual scans as a simplification
mechanism that allows us to elide $\writeProc$ views from
consideration when establishing $\snapshotSig$. A virtual scan will
denote a scan that detected no change in the array, i.e., it returned
at line~\ref{alg:afek-ret1}. Formally, a virtual scan $\vrt{\vsc}$
consists of rep events $\rep{\sca{\vsc}{i}}$ and $\rep{\scb{\vsc}{i}}$
for each $i$, where both $\rep{\sca{\vsc}{i}}$ and
$\rep{\scb{\vsc}{i}}$ observe the same write (i.e., no change
detected), and every $\rep{\sca{\vsc}{i}}$ occurs before any
$\rep{\scb{\vsc}{j}}$:
\[
  \renewcommand{\arraystretch}{\eqSpacing}%
  \begin{tabular*}{\textwidth}{@{\hskip 2em}l@{\hskip 10em}c@{\extracolsep{\fill}}r}
    $\forall \vrt{\vsc}, i.$ &$\exists \abs{w_i}.\ \rep{\wra{w_i} \robs
                               \sca{\vsc}{i}} \land \rep{\wra{w_i} \robs \scb{\vsc}{i}}$ &(\ref{ax:afek-success} revisited)\\
    $\forall \vrt{\vsc}, i, j.$ &$\rep{\sca{\vsc}{i} \rb \scb{\vsc}{j}}$ &(\ref{ax:afek-scrb} revisited)
  \end{tabular*}
\]
To each abs scan, we can associate a virtual scan as follows. If the
abs scan terminated because it detected no change, than it immediately
is a virtual scan. If the abs scan terminated by returning a view from
a $\writeProc$, i.e., returned at line~\ref{alg:afek-ret2}, then that
view itself is an abs scan which can, recursively, be associated with
a virtual scan. We formally define this by $\SMap{\abs{[-]}}$ as
follows:
\[
  \SMap{\abs{s}} \defeq
  \begin{cases}
    \vrt{\vsc} & \text{where}\quad \exists k.\ \forall i.\ \rep{\sca{\vsc}{i} =
      \scanth{s}{i}{k}} \land \rep{\scb{\vsc}{i} =
      \scbnth{s}{i}{k}}\\
    \hline\\[-1em]
    \SMap{\rep{\wrs{w_i}}} & \text{where}\quad
    \begin{aligned}
      \exists & 
      \abs{w'_i}, \abs{w''_i}, j, k.\ j < k
      \land \abs{w''_i \rb w'_i \rb w_i} \land {}\\
      &\rep{\wra{w''_i} \robs \scanth{s}{i}{j}} \land
      \rep{\wra{w'_i} \robs \scanth{s}{i}{k}} \land
      \rep{\wra{w_i} \robs \scbnth{s}{i}{k}}
    \end{aligned}
  \end{cases}
\]
We instantiate $\abs{\WrEff_i}$ to be the set of all writes
$\abs{w_i} \in \abs{W_i}$ where $\rep{\wra{w_i}}$ is defined; these
are the writes that executed their effect. We also instantiate the
visibility relations $\abs{\obs}$ and $\abs{\robs}$ as follows, using
the helper relation $\abs{\hlrobs}$.
\begin{align*}
  \abs{w_i \hlrobs \vrt{\vsc}} &\wideDefeq \rep{\wra{w_i} \robs \sca{\vsc}{i}} &
  \abs{w_i \robs s} &\wideDefeq \abs{w_i \hlrobs \SMap{\abs{s}}} &
  \abs{\obs} &\wideDefeq \abs{\robs}
\end{align*}
In English, the scan $\abs{s}$ reads from, and also observes,
$\abs{w_i}$ iff there is an appropriate rep event in the virtual scan
$\SMap{\abs{s}}$ that reads from $\rep{\wra{w_i}}$ at the level of rep
events.

\begin{lemma}
  $\SMap{\abs{[-]}}$ is a well-founded recursive definition.
\end{lemma}

\begin{proof}
  The first case of $\SMap{\abs{[-]}}$ is immediately well-founded.
  For the second case it is well-founded because the ending time of
  the scan $\rep{\wrs{w_i}}$ on the right is smaller than the ending
  time of $\abs{s}$ on the left, and the ending times are bounded from
  below by $0$. To see that the ending time of $\rep{\wrs{w_i}}$ is
  below that of $\abs{s}$, suppose otherwise. Then it must be
  $\abs{s} \rb \rep{\wra{w_i}}$, because from the code of $\writeProc$
  we have that $\rep{\wrs{w_i} \rb \wra{w_i}}$. In particular,
  $\rep{\scbnth{s}{i}{k}} \rb \rep{\wra{w_i}}$ for every $k$. But,
  because $\rep{\wra{w_i}}$ is observed by some rep read in $\abs{s}$,
  we also have $\rep{\wra{w_i} \robs \scbnth{s}{i}{k}}$ for some $k$.
  Thus, we have an event $\rep{\scbnth{s}{i}{k}}$ that terminated
  before the event $\rep{\wra{w_i}}$ that it observes, which
  contradicts~\cref{ax:wfobs}.
\end{proof}

\begin{lemma}\label{lem:afek-vrtinscan}
  For each scan $\abs{s}$, we have $\SMap{\abs{s}} \subev \abs{s}$.
\end{lemma}

\begin{proof}
  If $\SMap{\abs{s}} = \vrt{\vsc}$ corresponds to the first case, then
  this is trivial, since $\vrt{\vsc}$ consists of the selected rep
  events of $\abs{s}$, i.e., $\rep{\sca{\vsc}{i} = \scanth{s}{i}{k}}$
  and $\rep{\scb{\vsc}{i} = \scbnth{s}{i}{k}}$.
  Otherwise, $\SMap{\abs{s}} = \SMap{(\rep{\wrs{w_i}})}$ corresponds
  to the second case for some write view $\rep{\wrs{w_i}}$. By
  recursion on the definition of $\SMap{[-]}$ it must be
  $\SMap{\abs{s}} = \SMap{(\rep{\wrs{w_i}})} \subev \rep{\wrs{w_i}}$,
  so it suffices to show $\rep{\wrs{w_i}} \subev \abs{s}$.
  By definition, if we have
  $\abs{s}.\evStart \leq \rep{\wrs{w_i}}.\evStart$ and
  $\rep{\wrs{w_i}}.\evEnd \leq \abs{s}.\evEnd$ then we have
  $\rep{\wrs{w_i}} \subev \abs{s}$. During the scan, at least three
  different writes must have been observed:
  $\rep{\wra{w_i} \robs \scbnth{s}{i}{k}}$ and
  $\rep{\wra{w'_i} \robs \scanth{s}{i}{k}}$ and
  $\rep{\wra{w''_i} \robs \scanth{s}{i}{j}}$ where we have
  $\abs{w''_i \rb w'_i \rb w_i}$ by the algorithm being single-writer.
  If we have $\rep{\wra{w'_i} \rb \scanth{s}{i}{j}}$ we
  contradict \cref{ax:mem-nowrbetween} for
  $\rep{\wra{w''_i} \robs \scanth{s}{i}{j}}$ since we have
  $\rep{\wra{w''_i} \rb \wra{w'_i} \rb \scanth{s}{i}{j}}$, thus
  we have $\rep{\wra{w'_i} \not\rb \scanth{s}{i}{j}}$, which by
  definition gives us
  $\rep{\scanth{s}{i}{j}}.\evStart \leq \rep{\wra{w'_i}}.\evEnd$.
  Similarly,
  if we have $\rep{\wra{w'_i} \rb \scanth{s}{i}{j}}$ we
  contradict \cref{ax:mem-nowrbetween} for
  $\rep{\wra{w''_i} \robs \scanth{s}{i}{j}}$ since we have
  $\rep{\wra{w''_i} \rb \wra{w'_i} \rb \scanth{s}{i}{j}}$, thus
  we have $\rep{\scbnth{s}{i}{k}} \not\rb \rep{\wra{w_i}}$, which by
  definition gives us
  $\rep{\wra{w_i}}.\evStart \leq \rep{\scbnth{s}{i}{k}}.\evEnd$.
  Together, we have
  \[
    \abs{s}.\evStart < \rep{\scanth{s}{i}{j}}.\evStart \leq
    \rep{\wra{w'_i}}.\evEnd < \rep{\wrs{w_i}}.\evStart\ ,
  \]
  and
  \[
    \rep{\wrs{w_i}}.\evEnd < \rep{\wra{w_i}}.\evStart \leq
    \rep{\scbnth{s}{i}{k}}.\evEnd < \abs{s}.\evEnd\ ,
  \]
  which gives us
  $\abs{s}.\evStart \leq \rep{\wrs{w_i}}.\evStart$ and
  $\rep{\wrs{w_i}}.\evEnd \leq \abs{s}.\evEnd$, i.e.,
  $\rep{\wrs{w_i}} \subev \abs{s}$.
\end{proof}

\begin{genthm}{\cref{ax:wfobs}}\label{lem:afek-wfobs}
  If $\abs{e \obs^+ e'}$ then we cannot have $\abs{e' \rbeq e}$.
\end{genthm}

\begin{proof}
  If $\abs{e \obs^+ e'}$ and $\abs{e' \rbeq e}$ then
  $\abs{e \obs^+ e'}$ can only hold if $\abs{e \robs e'}$, i.e., we
  have $\abs{e} = \abs{w_i}$ and $\abs{e'} = \abs{s}$ and
  $\vrt{\vsc} = \SMap{\abs{s}}$ such that
  $\rep{\wra{w_i} \robs \sca{\vsc}{i}}$. We have that
  $\abs{e' \rbeq e}$ can only be $\abs{e' \rb e}$, since $\abs{e}$ and
  $\abs{e'}$ cannot be equal events by $\abs{e}$ being a write and
  $\abs{e'}$ being a scan. But then, by
  \cref{ax:rb-absrep,lem:afek-vrtinscan}, from $\abs{e' \rb e}$ we
  derive $\rep{\sca{\vsc}{i} \rb \wra{w_i}}$, which contradicts
  \cref{ax:wfobs} for registers.
\end{proof}

\begin{genthm}{\cref{ax:ss-io}}
  For every terminated scan $\abs{s}$ and index $i$, there exists some
  write $\abs{w_i}$ such that $\abs{w_i \robs s}$ and
  $\abs{w_i}.\evIn = \abs{s}.\evOut[i]$.
\end{genthm}

\begin{genthm}{\cref{ax:ss-robsuniq}}
  \label{lem:afek-robsuniq}
  If $\abs{w_i \robs s}$ and $\abs{w'_i \robs s}$ then we must have
  $\abs{w_i = w'_i}$.
\end{genthm}

\begin{proof}
  If $\abs{w_i \robs s}$ and $\abs{w'_i \robs s}$, then
  $\rep{\wra{w_i}}$ and $\rep{\wra{w'_i}}$ are both observed by some
  $\rep{\sca{\vsc}{i}}$, i.e., the same read. Thus, by
  \cref{ax:mem-robsuniq}, $\rep{\wra{w_i} = \wra{w'_i}}$, and since
  each write executes $\rep{\wra{w_i}}$ only once by
  Property~\ref{ax:evsig-uniq}, it must be $\abs{w_i = w'_i}$.
\end{proof}

\begin{proof}
  The structure of the algorithm ensures that at every return of
  $\scanProc$, writes of every memory cell are observed and returned.
\end{proof}

\begin{genthm}{\cref{ax:ss-wrtotal}}
  For two distinct writes $\abs{w_i}$ and $\abs{w'_i}$, we have either
  $\abs{w_i \hb w'_i}$ or $\abs{w'_i \hb w_i}$.
\end{genthm}

\begin{proof}
  Holds by the assumption that the algorithm is single-writer.
\end{proof}

\begin{genthm}{\cref{ax:ss-wrterm}}
  Every terminated write is effectful, i.e.,
  $\term(\abs{W_i}) \subseteq \abs{\WrEff_i}$.
\end{genthm}

\begin{proof}
  Since each terminated write must have $\rep{\wra{w_i}}$ defined,
  which is how we define the set $\abs{\WrEff_i}$, thus every
  terminated write must be in some $\abs{\WrEff_i}$.
\end{proof}

\begin{lemma}\label{lem:hb-afek}
  If $\abs{w_i \hb s}$ then $\abs{w_i \hb \SMap{\abs{s}}}$.
\end{lemma}

\begin{proof}
  This proof follows similarly to the proof of \cref{lem:hb-ss2fwd}.
  By viewing $\abs{w_i \hb s}$ as a chain of $\abs{\hb_1}$, we have
  $\abs{w_i \hb_1 \dots \hb_1 s}$. By induction on the length of the
  chain. In our base case we have $\abs{w_i \hb_1 s}$, which splits
  into cases $\abs{w_i \rb s}$ and $\abs{w_i \obs s}$, where the
  latter can only be $\abs{w_i \robs s}$. For $\abs{w_i \rb s}$, by
  \cref{ax:rb-absrep,lem:afek-vrtinscan}, we have
  $\abs{w_i \rb \SMap{\abs{s}}}$. For $\abs{w_i \robs s}$, we have
  $\abs{w_i \hlrobs \SMap{\abs{s}}}$ directly from the definition of
  $\abs{\robs}$. The inductive step is similar.
\end{proof}

\begin{lemma}\label{ax:afek-hbabsrep}
  Choose some $j$. For events
  $\abs{e}, \abs{e'} \in \bigcup_i \abs{W_i} \cup \vrt{\Vsc}$ where
  $\abs{e'}$ is either a write $\abs{e'} = \abs{w'_i}$ such that
  $\rep{\wra{w'_i}}$ is defined or a virtual scan
  $\abs{e'} = \vrt{\vsc'}$, if $\abs{e \hb e'}$ holds then there
  exists $\rep{e_r}$ and $\rep{e'_r}$ such that $\rep{e_r \hb e'_r}$,
  where if $\abs{e} = \abs{w_i} \in \abs{W_i}$ then
  $\rep{e_r} = \rep{\wra{w_i}}$ and if
  $\abs{e} = \vrt{\vsc} \in \vrt{\Vsc}$ then
  $\rep{e_r} = \rep{\scb{\vsc}{j}}$, and dually for $\abs{e'}$ and
  $\rep{e'_r}$.
\end{lemma}

\begin{proof}
  This proof follows similarly to the proof of
  \cref{lem:fwd-hbabsrep}. We can view $\abs{e \hb e'}$ as a chain of
  $\abs{\hb_1}$ relations connecting $\abs{e}$ and $\abs{e'}$; that
  is, $\abs{e \hb_1 \cdots \hb_1 e'}$, where by definition
  $\abs{\hb_1} = (\abs{\rb} \cup \abs{\obs}) = (\abs{\rb} \cup
  \abs{\hlrobs})$. It is easy to see that each abs event in this chain
  is populated with a rep event. Indeed, if an abs event's relation
  $\abs{\hb_1}$ to its successor is realized by $\abs{\rb}$, then the
  event is terminated and all of its rep events are executed (and by
  the structure of the algorithm, each abs event must have at least
  one rep event). Alternatively, if the abs event's relation
  $\abs{\hb_1}$ to its successor is realized by $\abs{\obs}$, then the
  event must be populated as per the definitions of reads-from. This
  leaves out $\abs{e'}$, which has no successor, but $\abs{e'}$ is
  populated by either being a write, thus having $\rep{\wra{w_i}}$
  defined by assumption, or being a virtual scan, which are first
  defined when a scan returns.
  
  Now the proof is by induction on the length of the above chain. The
  base case is when $\abs{e \hb_1 e'}$. If $\abs{e \rb e'}$, the proof
  is by \cref{ax:rb-absrep} giving us $\rep{e_r \rb e'_r}$ where if
  $\abs{e} = \abs{w_i} \in \abs{W_i}$ then
  $\rep{e_r} = \rep{\wra{w_i}}$ and if
  $\abs{e} = \vrt{\vsc} \in \vrt{\Vsc}$ then
  $\rep{e_r} = \rep{\scb{\vsc}{j}}$, and dually for $\abs{e'}$ and
  $\rep{e'_r}$. For $\abs{e \hlrobs e'}$, we have
  $\abs{e } = \abs{w_i}$ and $\vrt{\vsc'} = \abs{e'}$ with
  $\rep{\wra{w_i} \robs \sca{\vsc'}{i}}$. By \cref{ax:afek-scrb} we
  construct
  $\rep{\wra{w_i} \robs \sca{\vsc'}{i} \rb \scb{\vsc'}{j}}$ which
  gives us our goal $\rep{\wra{w_i} \hb \scb{\vsc'}{j}}$. The
  inductive step is similar, with the addition that we need
  transitivity of $\rep{\hb}$ to join events.
\end{proof}

\begin{genthm}{\cref{ax:ss-nowrbetween}}
  \label{lem:afek-nowrbetween}
  If $\abs{w_i \robs s}$ then there does not exist a write
  $\abs{w'_i}$ such that $\abs{w_i \hb w'_i \hb s}$ holds.
\end{genthm}

\begin{proof}
  Given $\abs{w_i \robs s}$, we assume there exists $\abs{w'_i}$ such
  that $\abs{w_i \hb w'_i \hb s}$, we derive a contradiction. By
  assumption, we have $\abs{w_i \hlrobs \vrt{\vsc}}$ for
  $\vrt{\vsc} = \SMap{\abs{s}}$. By \cref{lem:hb-afek} and
  $\abs{w_i \hb w'_i \hb s}$, we can derive
  $\abs{w_i \hb w'_i \hb \vrt{\vsc}}$, and then by
  \cref{ax:afek-hbabsrep}, also
  $\rep{\wra{w_i} \hb \wra{w'_i} \hb \scb{\vsc}{i}}$. By the
  definition of $\abs{\hlrobs}$ and \cref{ax:afek-success},
  $\abs{w_i \hlrobs \vrt{\vsc}}$ implies
  $\rep{\wra{w_i} \robs \sca{\vsc}{i}, \scb{\vsc}{i}}$. But then
  $\rep{\wra{w'_i}}$ occurrs between $\rep{\wra{w_i}}$ and
  $\rep{\scb{\vsc}{i}}$, contradicting \cref{ax:mem-nowrbetween} and
  $\rep{\wra{w_i} \robs \scb{\vsc}{i}}$.
\end{proof}

\begin{genthm}{\cref{ax:ss-mono}}\label{lem:afek-mono}
  If $\abs{w_i, w_j \robs s}$ and $\abs{w'_i, w'_j \robs s'}$
  with $\abs{w_i \hb w'_i}$ then we cannot have $\abs{w'_j \hb w_j}$.
\end{genthm}

\begin{proof}
  Given $\abs{w_i, w_j \robs s}$ and $\abs{w'_i, w'_j \robs s'}$ with
  $\abs{w_i \hb w'_i}$, we show that also having $\abs{w'_j \hb w_j}$
  lead to a contradiction. We unfold the assumptions
  $\abs{w_i, w_j \robs s}$ and $\abs{w'_i, w'_j \robs s'}$ into
  $\abs{w_i, w_j \hlrobs \vrt{\vsc}}$ and
  $\abs{w'_i, w'_j \hlrobs \vrt{\vsc'}}$, where
  $\vrt{\vsc} = \SMap{\abs{s}}$ and $\vrt{\vsc'} = \SMap{\abs{s'}}$.
  By the definition of $\abs{\hlrobs}$ and \cref{ax:afek-success},
  from $\abs{w_j \hlrobs \vrt{\vsc}}$ we have
  $\rep{\wra{w_j} \robs \sca{\vsc}{j}, \scb{\vsc}{j}}$, and similarly
  for $\abs{w'_j \hlrobs \vrt{\vsc'}}$. By \cref{ax:afek-scrb} we have
  $\rep{\sca{\vsc}{j} \rb \scb{\vsc}{i}}$ and
  $\rep{\sca{\vsc'}{i} \rb \scb{\vsc'}{j}}$. By \cref{ax:interval}, it
  is either $\rep{\sca{\vsc}{j} \rb \scb{\vsc'}{j}}$ or
  $\rep{\sca{\vsc'}{i} \rb \scb{\vsc}{i}}$. In the first case (the
  second is symmetric) we can construct
  $\rep{\wra{w'_j} \hb \wra{w_j} \robs \sca{\vsc}{j} \rb
    \scb{\vsc'}{j}}$, meaning that we have
  $\rep{\wra{w'_j} \hb \wra{w_j} \hb \scb{\vsc'}{j}}$ contradicting
  \cref{ax:mem-nowrbetween} by $\rep{\wra{w_j}}$ occurring between
  $\rep{\wra{w'_j} \robs \scb{\vsc'}{j}}$.
\end{proof}

\fi

\end{document}